\documentclass[12pt]{article}
\usepackage[margin=1in]{geometry}
\usepackage[small]{titlesec}
\usepackage{natbib,amsmath,amssymb,amsthm,bm,graphicx,caption,subcaption,booktabs,quoting,xr-hyper}
\usepackage[flushleft,neverdecrease]{paralist}
\usepackage[doublespacing]{setspace}
\usepackage[dvipsnames]{xcolor}
\quotingsetup{font=small,leftmargin=\parindent,rightmargin=\parindent}
\AtBeginDocument{%
\abovedisplayskip=5pt plus 2pt minus 2pt
\belowdisplayskip=\abovedisplayskip
\abovedisplayshortskip=2pt plus 2pt minus 2pt
\belowdisplayshortskip=\belowdisplayskip}
\titlespacing*{\section}{0pt}{*1}{*1}
\titlespacing*{\subsection}{0pt}{*1}{*1}
\setdefaultleftmargin{\parindent}{}{}{}{}{}
\bibpunct{(}{)}{;}{a}{}{,}
\allowdisplaybreaks

\graphicspath{{fig/}}

\newtheorem{condition}{Condition}
\newtheorem{lemma}{Lemma}
\newtheorem{proposition}{Proposition}
\newtheorem{theorem}{Theorem}
\newtheorem{corollary}{Corollary}

\def\bA{\mathbf{A}}
\def\bD{\mathbf{D}}
\def\bE{\mathbf{E}}
\def\bF{\mathbf{F}}
\def\bG{\mathbf{G}}
\def\bH{\mathbf{H}}
\def\bI{\mathbf{I}}
\def\bV{\mathbf{V}}
\def\bW{\mathbf{W}}
\def\bX{\mathbf{X}}
\def\bY{\mathbf{Y}}
\def\bZ{\mathbf{Z}}
\def\ba{\mathbf{a}}
\def\be{\mathbf{e}}
\def\bs{\mathbf{s}}

\def\bv{\mathbf{v}}

\def\bSigma{\bm\Sigma}
\def\bOmega{\bm\Omega}
\def\ve{\varepsilon}
\def\blambda{\bm\lambda}
\def\bomega{\bm\omega}
\def\bone{\mathbf{1}}
\def\bzero{\mathbf{0}}
\def\cS{\mathcal{S}}
\def\cU{\mathcal{U}}
\DeclareMathOperator\CV{CV}
\DeclareMathOperator\Var{Var}
\DeclareMathOperator\diag{diag}
\DeclareMathOperator\sgn{sgn}
\DeclareMathOperator\tr{tr}

\begin{document}
\begin{onehalfspace}

\title{CARE: Large Precision Matrix Estimation for Compositional Data}
\author{Shucong Zhang$^\text{a}$, Huiyuan Wang$^\text{b}$, and Wei Lin$^\text{c}$}
\date{}
\maketitle

\footnotetext{$^\text{a}$School of Statistics, University of International Business and Economics, Beijing, China; $^\text{b}$Department of Biostatistics, Epidemiology and Informatics, Perelman School of Medicine, University of Pennsylvania, Philadelphia, PA; $^\text{c}$School of Mathematical Sciences and Center for Statistical Science, Peking University, Beijing, China}

\footnotetext{\emph{Contact}: Wei Lin, School of Mathematical Sciences and Center for Statistical Science, Peking University, Beijing 100871, China (weilin@math.pku.edu.cn).}

\begin{abstract}
\noindent High-dimensional compositional data are prevalent in many applications. The simplex constraint poses intrinsic challenges to inferring the conditional dependence relationships among the components forming a composition, as encoded by a large precision matrix. We introduce a precise specification of the compositional precision matrix and relate it to its basis counterpart, which is shown to be asymptotically identifiable under suitable sparsity assumptions. By exploiting this connection, we propose a composition adaptive regularized estimation (CARE) method for estimating the sparse basis precision matrix. We derive rates of convergence for the estimator and provide theoretical guarantees on support recovery and data-driven parameter tuning. Our theory reveals an intriguing trade-off between identification and estimation, thereby highlighting the blessing of dimensionality in compositional data analysis. In particular, in sufficiently high dimensions, the CARE estimator achieves minimax optimality and performs as well as if the basis were observed. We further discuss how our framework can be extended to handle data containing zeros, including sampling zeros and structural zeros. The advantages of CARE over existing methods are illustrated by simulation studies and an application to inferring microbial ecological networks in the human gut.\bigskip

\end{abstract}
\end{onehalfspace}

\newpage
\section{Introduction}
High-dimensional compositional data, which lie in a high-dimensional unit simplex, arise in a wide spectrum of contemporary applications. For example, in metagenomic studies, the composition of microbial communities can be comprehensively quantified using 16S rRNA gene or shotgun metagenomic sequencing, which yields the relative abundances of thousands of bacterial taxa \citep{li2015}. Other prominent examples include the chemical composition of particulate matter in environmental studies \citep{Donkelaar.etal2019}, market shares of industries in economics \citep{Sutton2007}, and word frequencies of texts in machine learning \citep{Blei.Lafferty2007}. A central task in such applications is to infer the conditional dependence relationships among the components forming a composition, which are often represented by an undirected graphical model and encoded by the inverse covariance or precision matrix \citep{Lauritzen1996}.

The simplex constraint creates intrinsic challenges to inference for high-dimensional graphical models and precision matrices. A conventional method starts with applying the log-ratio transformation \citep{aitchison1982,aitchison2003} to the compositional data, and proceeds with inference as if the transformed data lie in a Euclidean space. This approach, however, is unsatisfactory for our problem in the following respects.

\begin{compactenum}
\item The simplex constraint induces negative correlations between originally independent components, and hence a ``null'' structure of the compositional precision matrix away from a diagonal matrix. It is therefore not clear how to impose the sparsity structure as in the case of Euclidean data.
\item Among the conventional matrix specifications of compositional covariance structures \citep{aitchison2003}, the log-ratio and isometric log-ratio covariance matrices treat the components asymmetrically, while the variation matrix is not a covariance matrix and the centered log-ratio covariance matrix is not invertible. Thus, none of these specifications give rise to a compositional precision matrix that is permutation invariant under sparsity assumptions.
\item Since each log-ratio involves more than one component, conditional dependence relationships among the log-ratios are more difficult to interpret, especially when model selection or support recovery is the primary purpose.
\end{compactenum}

A potential way to resolve these difficulties is to return to the basis precision matrix $\bOmega_0$, that is, the precision matrix of the log-basis variables that generate a composition; see Section \ref{sec:ident} for the formal definition. Since the log-basis variables are free of constraints, $\bOmega_0$ can be assumed to be sparse and interpreted in the usual manner. However, owing to the many-to-one relationship between the basis and the composition, $\bOmega_0$ is unidentifiable from compositional data in any finite dimension, making it a seemingly unrealistic goal.

A closely related question was addressed by \citet{cao2019}. They showed that the basis covariance matrix $\bSigma_0=\bOmega_0^{-1}$ belonging to a suitable sparsity class is approximately identifiable and indistinguishable from the centered log-ratio covariance matrix $\bSigma_c$ as the dimensionality increases. Their idea relies on a fundamental decomposition, originally due to \citet{aitchison2003}, that relates the basis and compositional covariance structures. On the other hand, such a relationship for the precision matrices has long been lacking in the literature, preventing the problem of precision matrix estimation from being solved. Any attempt to apply \citet{cao2019}'s argument directly to the precision matrix setting would face the necessity of imposing sparsity assumptions simultaneously on the covariance and precision matrices, which is unreasonable and impractical.

This article bridges the conceptual gap by introducing a precise specification of the compositional precision matrix and establishing a transparent relationship between the basis and compositional precision matrices. This novel connection inherits a similar low-rank plus sparse form from the covariance decomposition mentioned above, allowing us to prove the approximate identifiability of $\bOmega_0$ under suitable sparsity assumptions. These results are of independent interest and would be useful for other inference problems for compositional data where precision matrices play a role. Building on our precision matrix specification and its nice properties, we adopt the idea of constrained $\ell_1$-minimization \citep{cai2011} to develop a composition adaptive regularized estimation (CARE) method for estimating the sparse basis precision matrix. We further derive rates of convergence for the resulting estimator under various matrix norms, and provide theoretical guarantees on support recovery and data-driven parameter tuning.

Our theoretical results reveal an intriguing trade-off between identification and estimation. More specifically, as the dimension $p$ grows, the identification error decreases, while the estimation error increases. Remarkably, when $p\log p\gg n$ with $n$ being the sample size, the estimation error dominates the identification error and the CARE estimator achieves the minimax optimal rates for unconstrained data in \citet{cai2016}. This entails that, in sufficiently high dimensions, our method performs as well as if the basis were observed. This blessing of dimensionality has direct implications for practical data analysis. For example, in microbiome studies, the number of bacterial taxa depends on the taxonomic level under examination, which should be judiciously chosen to facilitate inference and interpretation.

\subsection{Related Work}
There is a vast and rapidly growing literature on high-dimensional graphical models and precision matrix estimation. The existing methods fall roughly into three categories. The neighborhood-based approach exploits a nodewise regression formulation of Gaussian graphical models, and fits the high-dimensional regressions using regularization methods such as the lasso, Dantzig selector, and scaled lasso \citep{meinshausen2006,yuan2010,ren2015,Fan.Lv2016}. Instead, the likelihood-based approach estimates the precision matrix as a whole by directly penalizing the joint Gaussian likelihood with the graphical lasso \citep{yuan2007,friedman2008} or nonconvex penalties \citep{lam2009}. Finally, the CLIME-related methods are motivated by modified score functions and solve constrained $\ell_1$-minimization or penalized empirical risk minimization problems \citep{cai2011,zhang2014,liu2015,cai2016}. Their ideas are to simplify the loss function and yield theoretical and computational advantages.

Extensions of the problem to compositional data have attracted considerable attention in the computational biology literature. The SPIEC-EASI procedure proposed by \citet{kurtz2015} treats the centered log-ratio transformed data as Euclidean data and directly applies the usual neighborhood selection and graphical lasso methods. \citet{fang2017} developed an $\ell_1$-penalized likelihood method called gCoda by parametrizing the logistic normal likelihood with $\bOmega_0$. The identifiability issue under the parametrization was noted but not rigorously addressed. A related method using the Bayesian graphical lasso was considered by \citet{Schwager.etal2017}. \citet{yuan2019} introduced an $\ell_1$-penalized CD-trace approach by borrowing the idea of D-trace loss from \citet{zhang2014}. The derivation, however, relies on an exchangeability condition, which is equivalent to the row sums of $\bSigma_0$ being all equal. \citet{Zhang.He2019} employed the sparse columnwise inverse operator (SCIO) method of \citet{liu2015} to find an approximate inverse of the sample centered log-ratio covariance matrix. Their method is based on the same heuristic argument as that of SPIEC-EASI, and would suffer from similar conceptual issues.

The difficulties and ambiguities in deriving and lending support to the aforementioned methods showcase the lack of convenient mathematical tools for the study of compositional precision matrices. As commented by \citet[p.~64]{aitchison2003}, before introducing the tools of compositional covariance matrices,
\begin{quoting}
we would not expect that excellent tool of the wide open spaces (or $\mathbb{R}^d$) of North America, namely the barbecue, necessarily to be an appropriate concept for cooking in the confined space (or $\cS^d$) of a low-cost housing flatlet in Hong Kong. If our concepts fail to serve us in new situations we must invent new concepts.
\end{quoting}
Similarly, we would not expect the oven for roasting Cantonese ducks (or compositional covariance matrices) to be perfect for roasting Peking ducks (or compositional precision matrices). We shall therefore suggest a new concept, whose mathematical form was first defined but not explored for inferential use by \citet{aitchison2003}, along with the associated tools that are more convenient and powerful for our purpose.

\subsection{Organization of the Article}
In Section \ref{sec:ident}, we present the precision matrix specification and its properties, thereby addressing the identifiability problem. Section \ref{sec:method} motivates and describes the CARE methodology. Section \ref{sec:theory} establishes theoretical guarantees for our proposed estimator in terms of rates of convergence, support recovery, and data-driven choice of tuning parameters. Extensions to data containing zeros and a perturbation theory are developed in Section 5. Simulation studies and an application to human gut microbiome data are presented in Sections 6 and 7, respectively. Section 8 contains some final discussion. Proofs and additional discussion and numerical results are provided in the supplementary materials.

We close this section by introducing some notation to be used throughout the article. For any vector $\ba$, matrix $\bA$, and $1\le r\le\infty$, we denote by $\|\ba\|_r$ the vector $\ell_r$-norm, and $\|\bA\|_{\max}$, $\|\bA\|_{L_r}$, and $\|\bA\|_F$ the entrywise $\ell_{\infty}$-norm, matrix $\ell_r$-norm, and Frobenius norm, respectively. Also, $\lambda_{\min}(\bA)$ and $\lambda_{\max}(\bA)$ denote the smallest and largest eigenvalues of $\bA$, respectively, and $\bA\succ0$ stands for $\bA$ being symmetric and positive definite. Finally, we write $\bI_p$ for the $p\times p$ identity matrix, $\bone_p$ the $p$-vector of 1s, and $\be_j$ the vector with 1 in the $j$th component and 0 elsewhere.

\section{Precision Matrix Specification and Identifiability}\label{sec:ident}
Consider a composition $\bX=(X_1,\dots,X_p)^T$ that takes values in the $(p-1)$-dimensional unit simplex $\cS^{p-1}=\{(x_1,\dots,x_p):x_j>0\text{ for }j=1,\dots,p,\sum_{j=1}^p x_j=1\}$. Let $\bW=(W_1,\dots,W_p)^T$ with $W_j>0$ for all $j$ be a latent vector, called the basis, that generates the observed composition via the normalization
\[
X_j=\frac{W_j}{\sum_{i=1}^{p}W_i},\qquad j=1,\dots,p.
\]
To specify a basis covariance structure compatible with the log-ratio form of compositional covariance structure, we define the log-basis vector $\bY=(Y_1,\dots,Y_p)^T$ with $Y_j=\log W_j$ for all $j$. Denote by $\bOmega_0=(\omega_{ij}^0)$ the precision matrix of $\bY$, which we call the basis precision matrix. In the important case where $\bY$ is multivariate normal, so that $\bX$ is logistic normal \citep[Property 6.1]{aitchison2003}, it follows from nonparanormal graphical model theory \citep{liu2009} that $W_i$ and $W_j$ are conditionally independent given the other $p-2$ variables if and only if $\omega_{ij}^0=0$. The graphical model interpretation of the precision matrix has also been extended to certain discrete and non-Gaussian cases \citep{loh2013,morrison2022}. Hence, recovering the graphical model structure among the $p$ basis variables amounts to estimating the support of $\bOmega_0$. In microbiome studies, $\bX$ represents the relative abundances of bacterial taxa measured by 16S rRNA gene or shotgun metagenomic sequencing, while $\bW$ represents the absolute abundances that are of primary interest but not directly observed. It has been found that sample abundance measurements after accounting for undersampling are well approximated by a log-normal distribution \citep{limpert2001,paulson2013}, and thus the nonparanormal graphical model seems plausible. In such applications, our goal is to estimate the basis precision matrix $\bOmega_0$ from the observed data on $\bX$.

Denote by $\bSigma_0=(\sigma_{ij}^0)=\bOmega_0^{-1}$ the basis covariance matrix, that is, the covariance matrix of $\bY$. The centered log-ratio transformation \citep{aitchison2003} gives the transformed data $\bZ=(Z_1,\dots,Z_p)^T$ with $Z_j=\log\{X_j/g(\bX)\}$, where $g(\bX)=\bigl(\prod_{j=1}^pX_j\bigr)^{1/p}$ is the geometric mean of $\bX$. Accordingly, the centered log-ratio covariance matrix $\bSigma_c=(\sigma_{ij}^c)$ is the covariance matrix of $\bZ$. Although $\bSigma_0$ is generally unidentifiable, \citet{cao2019} showed that it is close to its compositional counterpart $\bSigma_c$ in the sense that
\[
\|\bSigma_0-\bSigma_c\|_{\max}=O\left(\frac{\|\bSigma_0\|_{L_1}}{p}\right).
\]
This identifiability bound derives from the representation of $\bSigma_c$ in terms of $\bSigma_0$:
\begin{equation}\label{eq:sigma_c}
\bSigma_c=\bSigma_0-\bv_p\bv_p^T\bSigma_0-\bSigma_0\bv_p\bv_p^T+\bv_p\bv_p^T\bSigma_0\bv_p\bv_p^T,
\end{equation}
where $\bv_p=\bone_p/\sqrt p$. As a result, under sparsity assumptions on $\bSigma_0$ such that $\|\bSigma_0\|_{L_1}=o(p)$, $\bSigma_0$ is asymptotically identifiable and indistinguishable with $\bSigma_c$ as $p\to\infty$. This fact entails that $\bSigma_c$ serves as a good proxy for $\bSigma_0$ in sparse covariance estimation. Intuitively, we would expect a well-defined inverse of $\bSigma_c$ to be a good approximation of $\bOmega_0$. However, owing to the zero-sum constraint on $\bZ$, $\bSigma_c$ has an eigenvalue 0 and the corresponding eigenvector $\bv_p$ \citep[Theorem 4.6]{aitchison2003}.

A conceptually simple remedy for the singularity of $\bSigma_c$ is to find a generalized or pseudo-inverse. Let $\bSigma_c=\bV\bD_c\bV^T$ be the spectral decomposition of $\bSigma_c$, where $\bV$ is a $p\times(p-1)$ matrix of eigenvectors and $\bD_c=\diag(d_1,\dots,d_{p-1})$ with $d_1\ge\dots\ge d_{p-1}>0$. The rank-deficient matrix $\bSigma_c$ can be made invertible by a rank-1 correction. More precisely, for any constant $\rho>0$, $\bSigma_c+\rho\bv_p\bv_p^T$ is nonsingular and has the spectral decomposition
\[
\bSigma_c+\rho\bv_p\bv_p^T=(\bV,\bv_p)
\Biggl(\begin{matrix}
\bD_c & \bzero\\
\bzero & \rho
\end{matrix}\Biggr)
(\bV,\bv_p)^T,
\]
whose inverse is given by
\[
\bigl(\bSigma_c+\rho\bv_p\bv_p^T\bigr)^{-1}=(\bV,\bv_p)
\Biggl(\begin{matrix}
\bD_c^{-1} & \bzero\\
\bzero & \rho^{-1}
\end{matrix}\Biggr)
(\bV,\bv_p)^T.
\]
The Moore--Penrose inverse of $\bSigma_c$ can then be expressed as
\begin{equation}\label{eq:inverse}
\bSigma_c^+=\bV\bD_c^{-1}\bV^T=\bigl(\bSigma_c+\rho\bv_p\bv_p^T\bigr)^{-1}-\frac{1}{\rho}\bv_p\bv_p^T.
\end{equation}
We now define $\bOmega_c=(\omega_{ij}^c)=\bSigma_c^+$ as the \emph{compositional precision matrix}, since it provides a full parametrization of compositional inverse covariance structure and is identifiable from the observed compositional data. The definition of $\bSigma_c^+$ and expression \eqref{eq:inverse} were mentioned by \citet[sec.\ 5.6]{aitchison2003}, but its use for modeling and inference was not pursued therein. Indeed, since $\bOmega_c$ is indirectly defined through $\bSigma_c$, it is less interpretable and may not be very useful without an explicit link to the basis precision matrix $\bOmega_0$. Fortunately, by substituting \eqref{eq:sigma_c} into \eqref{eq:inverse} and using the Sherman--Morrison formula, we can establish the following relationship between $\bOmega_c$ and $\bOmega_0$.

\begin{theorem}\label{thm:rel}
The following relationship between $\bOmega_c$ and $\bOmega_0$ holds:
\[
\bOmega_c=\bOmega_0-\frac{\bOmega_0\bone_p\bone_p^T\bOmega_0}{\bone_p^T\bOmega_0\bone_p}.
\]
\end{theorem}

While there can be many possible ways to specify the compositional precision matrix, Theorem \ref{thm:rel} illustrates a major advantage of our particular specification: it admits a representation as the sum of its basis counterpart $\bOmega_0$ and a low-rank correction, which inherits a similar form from the decomposition of $\bSigma_c$ in \eqref{eq:sigma_c}. Here, the rank-1 component reflects a global effect of compositionality across all rows and columns of $\bOmega_0$, and is related to the information loss due to the normalization operation. In light of this decomposition, statistical wisdom learned from the literature on low-rank plus sparse matrix recovery problems \citep{candes2011,Chandrasekaran.etal2012,Fan.etal2013} suggests that a suitably sparse $\bOmega_0$ can be separated from the compositionality effect. The following proposition collects some useful properties of $\bOmega_c$.

\begin{proposition}\label{prop:omega_c}
The matrix $\mathbf{\Omega}_c$ satisfies the following properties:
\begin{compactenum}[(a)]
\item $\bSigma_c\bOmega_c=\bG$, where $\bG=\bI_p-p^{-1}\bone_p\bone_p^T$;
\item $\bG\bOmega_c=\bOmega_c$ and $\bOmega_c\bone_p=\bzero$;
\item $\lambda_{\min}(\bOmega_0)\le\lambda_{\min}^+(\bOmega_c)\le\lambda_{\max}(\bOmega_c)\le\lambda_{\max}(\bOmega_0)$, where $\lambda_{\min}^+(\bOmega_c)$ denotes the smallest positive eigenvalue of $\bOmega_c$;
\item $\sigma_{jj}^c\omega_{jj}^c\ge(1-1/p)^2$ for $j=1,\dots,p$.
\end{compactenum}
\end{proposition}

Part (a) in Proposition \ref{prop:omega_c} indicates that $\bOmega_c$ is an approximate inverse of $\bSigma_c$ with $\bG$ playing the role of the identity matrix for centered log-ratio vectors. In fact, any $p$-vector $\ba$ satisfies $\ba^T\bone_p=0$ if and only if $\bG\ba=\ba$ \citep[app.\ A]{aitchison2003}. Part (b) parallels Theorem 4.6 of \citet{aitchison2003}, part (c) confines the range of the positive eigenvalues of $\bOmega_c$ to within that of $\bOmega_0$, and part (d) is analogous to the inequality that $\sigma_{jj}^0\omega_{jj}^0\ge1$. These properties are essential for the development of methodology and theory based on $\bOmega_c$.

To impose sparsity on the target matrix $\bOmega_0$, we consider the class of sparse precision matrices
\[
\cU_q(s_0(p),M_p)=
\begin{Bmatrix}
\displaystyle\bOmega=(\omega_{ij}):\bOmega\succ0,\max_{1\le j\le p}\sum_{i=1}^p|\omega_{ij}|^{q}\le s_0(p),\\
\displaystyle\|\bOmega\|_{L_1}\le M_p,1/R\le\lambda_{\min}(\bOmega)\le\lambda_{\max}(\bOmega)\le R
\end{Bmatrix},
\]
where $0\le q<1$, $R>1$ is a constant, and $s_0(p)$ and $M_p$ are bounded away from 0 and may diverge with $p$. This class covers a wide range of sparse precision matrices and is similar to those considered for unconstrained high-dimensional data \citep{cai2011,ren2015,cai2016}. Under the assumption that $\bOmega_0$ belongs to the class $\cU_q(s_0(p),M_p)$, $\bOmega_c$ is generally not sparse, but remains close to $\bOmega_0$ and preserves some of its properties, as formalized by the following results.

\begin{proposition}\label{prop:ident}
Suppose that $\bOmega_0\in\cU_q(s_0(p),M_p)$. Then
\[
\frac{R^{-3}}{p}\le\|\bOmega_0-\bOmega_c\|_{\max}\le\frac{RM_p^2}{p}.
\]
\end{proposition}

\begin{proposition}\label{prop:bound}
Suppose that $\bOmega_0\in\cU_q(s_0(p),M_p)$. Then
\[
\|\bOmega_c\|_{L_1}\le(1+R^2)M_p,\qquad 1/R\le\lambda_{\min}^+(\bOmega_c)\le\lambda_{\max}(\bOmega_c)\le R.
\]
\end{proposition}

Proposition \ref{prop:ident} entails that the sparse basis precision matrix $\bOmega_0$ is approximately identifiable provided that $M_p=o(\sqrt p)$. The identification error due to approximating $\bOmega_0$ by $\bOmega_c$ vanishes asymptotically as $p\to\infty$. This blessing of dimensionality motivates us to use $\bOmega_c$ as a proxy for $\bOmega_0$ and develop procedures for estimating $\bOmega_0$ through $\bOmega_c$. In particular, when $M_p$ does not grow with $p$, the identification error decays at the rate of $O(p^{-1})$, which is sharp and cannot be improved. Proposition \ref{prop:bound} gives bounds on the matrix $\ell_1$-norm and positive eigenvalues of $\bOmega_c$, which are analogous to those for $\bOmega_0$. Finally, we observe that if the eigenvalue condition in the definition of $\cU_q(s_0(p),M_p)$ is replaced by $\lambda_{\max}(\bOmega)/\lambda_{\min}(\bOmega)\le R$ as in \citet{cai2016}, then the $\|\cdot\|_{\max}$-distance between $\bOmega_0$ and $\bOmega_c$ in Proposition \ref{prop:ident} can be similarly bounded as
\[
\|\bOmega_0-\bOmega_c\|_{\max}\le\frac{R\|\bOmega_0\|_{L_1}}{\sqrt p}\le\frac{RM_p}{\sqrt p}.
\]
In this case, the approximate identifiability of $\bOmega_0$ still holds as long as $M_p=o(\sqrt p)$.

\section{Composition Adaptive Regularized Estimation}\label{sec:method}
Suppose that $(\bX_1,\bZ_1),\dots,(\bX_n,\bZ_n)$ are independent realizations of $(\bX,\bZ)$. Denote by $\widehat\bSigma_c=(\hat\sigma_{ij}^c)= n^{-1}\sum_{k=1}^n(\bZ_k-\bar\bZ)(\bZ_k-\bar\bZ)^T$ the sample centered log-ratio covariance matrix, where $\bar\bZ=n^{-1}\sum_{k=1}^n\bZ_k$. In view of Proposition \ref{prop:ident}, we shall develop an estimator of $\bOmega_0$ via the proxy matrix $\bOmega_c$ based on the observed data or $\widehat\bSigma_c$. However, $\bOmega_c$ is in general not sparse, although it is entrywise close to the sparse matrix $\bOmega_0$. Also, we shall not impose restrictive distributional assumptions. Thus, it is unclear whether existing neighborhood or likelihood-based methods can effectively leverage the proxy $\bOmega_c$. A further look at the CLIME estimator reveals that its rate of convergence under the entrywise $\ell_{\infty}$-norm depends crucially on the matrix $\ell_1$-norm bound $M_p$ \citep[Theorem 6]{cai2011}. This, together with the fact from Proposition \ref{prop:ident} that the identification error between $\bOmega_0$ and $\bOmega_c$ can be tightly bounded in terms of $M_p$, motivates us to consider a CLIME-type estimator for $\bOmega_0$.

Without loss of generality, we assume that the log-basis vector $\bY$ has mean $\bzero$. Since our inference goal is to estimate the precision matrix $\bOmega_0$ of $\bY$, it is conventional to impose the following two types of tail conditions on $\bY$.

\begin{condition}[Sub-Gaussian tail]\label{cond:subG}
Assume that $\log p=o(n)$ and there exist some constants $\eta>0$ and $K>0$ such that
\[
\sup_{\bs\in\mathbb{R}^p,\,\|\bs\|_2=1}E[\exp\{\eta(\bs^T\bY)^2/\Var(\bs^T\bY)\}]\le K.
\]
\end{condition}

\begin{condition}[Polynomial-type tail]\label{cond:poly}
Assume that $p=O(n^\gamma)$ for some constant $\gamma>0$ and there exist some constants $\ve>0$ and $K'>0$ such that
\[
\sup_{\bs\in\mathbb{R}^p,\,\|\bs\|_2=1}E|\bs^T\bY|^{4\gamma+4+\ve}\le K'.
\]
\end{condition}

If the log-basis $\bY$ were observable, then $\widehat\bSigma\bOmega_0$ would concentrate around $\bI_p$, where $\widehat\bSigma$ is the sample basis covariance matrix. In our setting, however, $\bY$ is not observable and our estimator should be based on $\widehat\bSigma_c$. Since $\bOmega_0$ and $\bOmega_c$ are asymptotically indistinguishable when $M_p=o(\sqrt p)$, by part (a) of Proposition \ref{prop:omega_c} we expect that $\widehat\bSigma_c\bOmega_0$, similar to $\widehat\bSigma_c\bOmega_c$, would concentrate around $\bG$. The following lemma shows that this is indeed the case.

\begin{lemma}\label{lem:conc}
Suppose that $\bOmega_0\in\cU_q(s_0(p),M_p)$ with $M_p=o(\sqrt p)$. Under Condition \ref{cond:subG} (or Condition \ref{cond:poly}), for any $\xi>0$, there exists some constant $C_0>0$ such that
\[
\|\widehat\bSigma_c\bOmega_0-\bG\|_{\max}\le C_0\biggl(\sqrt{\frac{\log p}{n}}+\frac{M_p}{\sqrt p}\biggr)
\]
with probability at least $1-O(p^{-\xi})$ (or $1-O(p^{-\xi/2}+n^{-\ve/8})$).
\end{lemma}

Lemma \ref{lem:conc} allows us to construct a constrained $\ell_1$-minimization procedure for estimating $\bOmega_0$ in a similar spirit to the CLIME approach of \citet{cai2011}. More specifically, we consider the following optimization problem:
\begin{equation}\label{eq:opt}
\text{minimize }\|\bOmega\|_1\quad\text{subject to}\quad\|\widehat\bSigma_c\bOmega-\bG\|_{\max}\le\lambda,
\end{equation}
where $\lambda>0$ is a tuning parameter and can be easily chosen to ensure the feasibility of $\bOmega_0$ provided that $M_p=o(\sqrt p)$. It is important to note that the second term $M_p/\sqrt p$ of the error bound in Lemma \ref{lem:conc} is due to the departure of $\bOmega_c$ from $\bOmega_0$. The main differences from the case of unconstrained data are that problem \eqref{eq:opt} incorporates the identifiability gap between $\bOmega_0$ and $\bOmega_c$, and that the matrix $\bG$ plays the role of the identity matrix. The overall optimization problem \eqref{eq:opt} can be further decomposed into $p$ independent vector minimization problems. As a result, the estimator $\widetilde\bOmega=(\tilde\omega_{ij})=(\widetilde\bomega_1,\dots,\widetilde\bomega_p)$ is defined through the solutions $\widetilde\bomega_j$ to the optimization problems
\begin{equation}\label{eq:opt_col}
\text{minimize }\|\bomega_j\|_1\quad\text{subject to}\quad\left\|\widehat\bSigma_c\bomega_j-\left(\be_j-\frac{\bone_p}{p}\right)\right\|_{\max}\le\lambda_j,
\end{equation}
where $\lambda_j>0$ are tuning parameters. To adapt to the variability and sparsity of the columns of $\bOmega_0$, $\lambda_j$ are allowed to vary from column to column, thus differing from the CLIME method where the tuning parameter is fixed. In practice, $\lambda_j$ can be chosen by cross-validation, which will be further discussed in Section \ref{sec:cv}. Finally, we define our CARE estimator $\widehat\bOmega=(\hat\omega_{ij})$ by symmetrizing $\widetilde\bOmega$:
\[
\hat\omega_{ij}=\hat\omega_{ji}=\tilde\omega_{ij}I(|\tilde\omega_{ij}|\le|\tilde\omega_{ji}|)+\tilde\omega_{ji} I(|\tilde\omega_{ij}|>|\tilde\omega_{ji}|).
\]

It is worth mentioning that our methodology differ from those of \citet{cao2019} and \citet{cai2016} in some important ways. First, while the COAT method is optimization-free, our method requires solving an optimization problem, whose dependence on the basis--composition relationship is implicit. This would make the theoretical development technically more challenging. Second, our method consists of a single-step constrained $\ell_1$-minimization procedure with column-specific tuning parameters, which departs markedly from the two-step ACLIME procedure and adapts better to the identification error in our problem. An alternative method that follows the ACLIME approach more closely is discussed in Supplementary Section \ref{sec:aclime}.

\section{Theoretical Properties}\label{sec:theory}
We investigate the theoretical properties of the CARE estimator. In Section \ref{sec:rate_supp}, we derive rates of convergence under various matrix norms and provide a support recovery guarantee for the estimator $\widehat\bOmega$. Data-driven choice of the tuning parameters $\lambda_j$ is justified in Section \ref{sec:cv}. In deriving the rates of convergence, we explicitly characterize the impact of dimensionality on the degree of identifiability and decompose the total error into an estimation error and an identification error. Results of this kind are rare in the literature, but see \citet{Fan.etal2013} and \citet{cao2019}.

\subsection{Rates of Convergence and Support Recovery}\label{sec:rate_supp}
We first present the rates of convergence for the estimator $\widehat\bOmega$ under different matrix norms.

\begin{theorem}\label{thm:omega}
Suppose that $\bOmega_0\in\cU_q(s_0(p),M_p)$ with $M_p=o(\sqrt p)$ and the tuning parameters are chosen as $\lambda_j\asymp\sqrt{(\log p)/n}+M_p/\sqrt p$ for all $j$. Under Condition \ref{cond:subG} (or Condition \ref{cond:poly}), for any $\xi>0$, there exists some constant $C>0$ such that the estimator $\widehat\bOmega$ satisfies
\begin{align*}
\|\widehat\bOmega-\bOmega_0\|_{\max}&\le C\biggl(M_p\sqrt{\frac{\log p}{n}}+\frac{M_p^2}{\sqrt p}\biggr),\\
\|\widehat\bOmega-\bOmega_0\|_{L_1}&\le Cs_0(p)\biggl(M_p\sqrt{\frac{\log p}{n}}+\frac{M_p^2}{\sqrt p}\biggr)^{1-q},\\
\frac{1}{p}\|\widehat\bOmega-\bOmega_0\|_F^2&\le Cs_0(p)\biggl(M_p\sqrt{\frac{\log p}{n}}+\frac{M_p^2}{\sqrt p}\biggr)^{2-q}
\end{align*}
with probability at least $1-O(p^{-\xi})$ (or $1-O(p^{-\xi/2}+n^{-\ve/8})$).
\end{theorem}

Note that the error bound under the matrix $\ell_1$-norm also holds under any matrix $\ell_r$-norm with $1\le r\le\infty$, since $\widehat\bOmega$ and $\bOmega_0$ are symmetric \citep[Lemma 7.2]{cai2016}. All the error bounds in Theorem \ref{thm:omega} have an appealing form that decomposes into two terms. The first term represents the estimation error due to the estimation of $\bOmega_c$ based on the observed data, while the second term arises from the identification error between $\bOmega_c$ and $\bOmega_0$. Note that the latter is due to the intrinsic difficulty of nonidentifiability in estimating $\bOmega_0$; as a result, it cannot be eliminated by any other methods. Theoretically, this can be seen from Proposition \ref{prop:ident}, where the identification error between $\bOmega_c$ and $\bOmega_0$ is bounded away from zero. Moreover, it is important to note that the dimension $p$ plays opposite roles in these two terms: it contributes a factor of $\log p$ to the former but a factor of $p^{-1/2}$ to the latter. This leads to a clear trade-off between identification and estimation, or a trade-off between the blessing and the curse of dimensionality, a striking phenomenon not observed in precision matrix estimation for unconstrained data. The quantity $M_p$ also has a role to play in the trade-off and has a larger impact on the identification error than on the estimation error. As $p$ becomes sufficiently large, the estimation error will dominate the identification error, as summarized in the following corollary.

\begin{corollary}\label{cor:minimax}
Assume that the conditions of Theorem \ref{thm:omega} hold. If $M_p=o(\sqrt{p(\log p)/n})$, then the estimator $\widehat\bOmega$ satisfies
\begin{align*}
\|\widehat\bOmega-\bOmega_0\|_{\max}&\le CM_p\sqrt{\frac{\log p}{n}},\\
\|\widehat\bOmega-\bOmega_0\|_{L_1}&\le CM_p^{1-q}s_0(p)\biggl(\frac{\log p}{n}\biggr)^{(1-q)/2},\\
\frac{1}{p}\|\widehat\bOmega-\bOmega_0\|_F^2&\le CM_p^{2-q}s_0(p)\biggl(\frac{\log p}{n}\biggr)^{1-q/2}
\end{align*}
with probability at least $1-O(p^{-\xi})$ or $1-O(p^{-\xi/2}+n^{-\ve/8})$.
\end{corollary}

Corollary \ref{cor:minimax} requires that $M_p=o(\sqrt{p(\log p)/n})$, which is stronger than the identifiability condition $M_p=o(\sqrt{p})$ and suggests a phase transition phenomenon for optimal estimation. In particular, the resulting rates of convergence under the matrix $\ell_r$-norm with $1\le r\le\infty$ and Frobenius norm match the minimax optimal rates for unconstrained data \citep{cai2016}. This indicates that, in sufficiently high dimensions, the CARE estimator performs as well as if the basis were observed, highlighting the blessing of dimensionality in compositional data analysis.

Based on the entrywise $\ell_{\infty}$-error bound in Theorem \ref{thm:omega}, we define the hard-thresholded estimator
$\widehat\bOmega^t=(\hat\omega_{ij}^t)$ with $\hat\omega_{ij}^t=\hat\omega_{ij}I(|\hat\omega_{ij}|>\tau_{np})$, where $\tau_{np}=C(M_p\sqrt{(\log p)/n}+M_p^2/\sqrt p)$. Our next result shows that, under an additional minimum signal assumption, the hard-thresholded estimator $\widehat\bOmega^t$ recovers the support of $\bOmega_0$ successfully.

\begin{theorem}\label{thm:supp}
Assume that the conditions of Theorem \ref{thm:omega} hold. If $\min_{(i,j):\omega_{ij}^0\ne0}|\omega_{ij}^0|>2\tau_{np}$, then the estimator $\widehat\bOmega^t$ satisfies
\[
\sgn(\hat\omega_{ij}^t)=\sgn(\omega_{ij}^0)
\]
for all $i,j=1,\dots,p$ with probability at least $1-O(p^{-\xi})$ or $1-O(p^{-\xi/2}+n^{-\ve/8})$.
\end{theorem}

\subsection{Data-Driven Choice of $\lambda_j$}\label{sec:cv}
We describe the cross-validation procedure for choosing the tuning parameters $\lambda_j$ in the $p$ vector optimization problems \eqref{eq:opt_col}. We randomly split the whole dataset into a training set with sample size $n_1$ and a test set with sample size $n_2=n-n_1$, where $n_1\asymp n_2\asymp n$. To measure the predictive performance of $\widetilde\bomega_j$, we consider the loss function
\begin{equation}\label{eq:loss}
L(\bomega_j,\bSigma_c)=\frac{1}{2}\bomega_j^T\bSigma_c\bomega_j-\left(\be_j-\frac{\bone_p}{p}\right)^T\bomega_j,
\end{equation}
whose derivative $\partial L/\partial\bomega_j=\bSigma_c\bomega_j-(\be_j-\bone_p/p)$ corresponds to the constraint in problem \eqref{eq:opt_col}. In a similar spirit to the D-trace loss of \citet{zhang2014}, we write the loss function \eqref{eq:loss} in matrix form as
\[
L(\bOmega,\bSigma_c)=\frac{1}{2}\tr(\bOmega\bSigma_c\bOmega)-\tr(\bG\bOmega).
\]
There is a subtle connection between $L(\bOmega,\bSigma_c)$ and the CD-trace loss function \citep{yuan2019}
\[
L_D(\bOmega,\bSigma_0)=\frac{1}{2}\tr(\bG\bOmega\bG\bSigma_0\bG\bOmega)-\tr(\bG\bOmega),
\]
where $\bOmega$ is a positive definite basis precision matrix. This loss function is minimized at $\bOmega_0$ only when the exchangeability condition $\bG\bSigma_0=\bSigma_0\bG$ holds. Similarly, if $\bG\bOmega_0=\bOmega_0\bG$, from $\bG\bSigma_0\bG=\bSigma_c$ and $\bG^2=\bG$ we see that the two loss functions coincide at $\bOmega_0$.

For $j=1,\dots,p$ and the $b$th split, denote by $\widehat\bomega_j^{(1b)}(\lambda_j)$ the estimator with tuning parameter $\lambda_j$ based on the training set, and by $\widehat\bSigma_c^{(2b)}$ the sample centered log-ratio covariance matrix based on the test set.  We choose the optimal value $\hat\lambda_j$ of $\lambda_j$ by minimizing
\[
\CV(\lambda_j)=\frac{1}{B}\sum_{b=1}^BL\{\widehat\bomega_j^{(1b)}(\lambda_j),\widehat\bSigma_c^{(2b)}\},
\]
and compute the estimator $\widehat\bomega_j(\hat\lambda_j)$ based on the whole dataset. Let $\widehat\blambda=(\hat\lambda_1,\dots,\hat\lambda_j)$ and then the data-driven estimator $\widehat\bOmega(\widehat\blambda)$ is obtained by combining $\widehat\bomega_j(\hat\lambda_j)$ for all columns. The procedure searches for the optimal $\hat\lambda_j$ through a grid of points $\delta_j\ell/N$, $\ell=1,\dots,N$, for some sufficiently large $\delta_j>0$. In practice, $\delta_j$ is set to the smallest value such that $\widehat\bomega_j=\bzero$; it is seen immediately from problem \eqref{eq:opt_col} that such a choice is $\delta_j=1-p^{-1}$. To establish theoretical guarantees for the data-driven estimator, we consider for simplicity the case of $B=1$ and omit the superscript $b$. The rate of convergence for the data-driven estimator $\widehat\bOmega^{(1)}(\widehat\blambda)= (\widehat\bomega_1^{(1)}(\hat\lambda_1),\dots,\widehat\bomega_p^{(1)}(\hat\lambda_p))$ is given by the following result.

\begin{theorem}\label{thm:cv}
Assume that Condition \ref{cond:subG} holds, $\bOmega_0\in\cU_q(s_0(p),M_p)$ with $M_p=o(\sqrt p)$, $\log N=o(n)$, and $\min_{1\le j\le p}\delta_j\ge NC_0(\sqrt{(\log p)/n}+M_p/\sqrt p)$ for $C_0>0$ defined in Lemma \ref{lem:conc}. Then the data-driven estimator $\widehat\bOmega^{(1)}(\widehat\blambda)$ satisfies
\[
\frac{1}{p}\|\widehat\bOmega^{(1)}(\widehat\blambda)-\bOmega_0\|_F^2=O_p\biggl\{s_0(p)\biggl(M_p\sqrt{\frac{\log p}{n}}+\frac{M_p^2}{\sqrt p}\biggr)^{2-q}\biggr\}.
\]
\end{theorem}

\section{Extensions to Data Containing Zeros}
We have so far confined our discussion to data lying in the strictly positive simplex. In many applications, however, the compositions are not directly observable but instead estimated from count data. When the count data are sparse, as is the case in microbiome studies, the compositional data obtained by simple normalization may contain an abundance of zeros and thus prevent the direct application of log-ratio transformations. In such situations, it is useful to distinguish between sampling zeros, which are due to undersampling of rare yet present taxa, and structural zeros, which represent truly absent taxa regardless of the sampling or sequencing depth \citetext{\citealp{agresti2013}, Sec.\ 10.6}. We describe how our framework can be extended to accommodate these two types of zeros in Sections \ref{ssec:sampl} and \ref{sec:struct}.

\subsection{Sampling Zeros}\label{ssec:sampl}
In the presence of only sampling zeros, the underlying proportions are still positive and can be estimated by many existing methods, ranging from easy-to-implement zero replacement and variable correction procedures \citep{shi2022} to more sophisticated approaches that borrow information across samples and taxa \citep{Zhang.Lin2019,cao2020}. This amounts to observing a noisy version of the compositional data, which are subject to measurement error that should be accounted for in subsequent analysis. Since the optimization problem \eqref{eq:opt} depends on the data only through $\widehat\bSigma_c$, it suffices to consider a perturbed covariance matrix $\check\bSigma_c=\widehat\bSigma_c+\bE$. For simplicity, consider the CARE estimator $\widehat\bOmega\equiv\widehat\bOmega(\lambda)$ with $\lambda_j=\lambda$ for all $j$, and let $\check\bOmega$ be the perturbed version with $\widehat\bSigma_c$ replaced by $\check\bSigma_c$. Let $\lambda_0=\|\widehat\bSigma_c\bOmega_0-\bG\|_{\max}$. The following perturbation result characterizes the impact of the measurement error $\bE$ on the estimator $\check\bOmega$.

\begin{theorem}\label{thm:perturb}
Suppose that $\bOmega_0\in\cU_q(s_0(p),M_p)$. If $\lambda=\lambda_0+M_p\|\bE\|_{\max}$, then the estimator $\check\bOmega$ satisfies
\[
\|\check\bOmega-\widehat\bOmega\|_{\max}\le 2M_p(\lambda_0+p^{-1})+3M_p^2\|\bE\|_{\max}.
\]
\end{theorem}

It is readily seen from Theorem \ref{thm:perturb} that as long as $\bE$ is small enough so that $M_p\|\bE\|_{\max}=O(\lambda_0)$, the second term in the perturbation bound is small. Furthermore, by Lemma \ref{lem:conc} we can choose $\lambda_0\asymp\sqrt{(\log p)/n}+M_p/\sqrt{p}$ with high probability, so that the term $p^{-1}$ is also negligible. Then the bound becomes
\[
\|\check\bOmega-\widehat\bOmega\|_{\max}=O\biggl(M_p\sqrt{\frac{\log p}{n}}+\frac{M_p^2}{\sqrt p}\biggr),
\]
matching the entrywise $\ell_\infty$-error bound in Theorem \ref{thm:omega}. This indicates that the measurement error would have no notable effect on the performance of the CARE estimator.

\subsection{Structural Zeros}\label{sec:struct}
Structural zeros may occur when the frequency of zeros in the count data is higher than the sampling model predicts. A popular approach to this problem is zero-inflated models, which are mixtures of a point mass at zero (zero part) and a count data distribution with positive mean (nondegenerate part). One can, in principle, use a zero-inflated model for multivariate count data \citep[e.g.,][]{tang2019,xu2021,zeng2022} to obtain strictly positive composition estimates from the nondegenerate part of the model; our approach can then be applied to these positive estimates as described earlier. The low-rank assumptions underlying some of these models, however, may lead to rank-deficient covariance structures, which render the estimation of a sparse precision matrix impossible. A remedy for the rank degeneracy problem is to construct an unbiased estimator of the centered log-ratio covariance matrix by inverse probability weighting. More specifically, suppose that $X_{kp}>0$ for all $k$ and let $\bZ_k^{(p)}=(Z_{k1}^{(p)},\dots,Z_{k,p-1}^{(p)})^T$ with $Z_{kj}^{(p)}=\log(X_{kj}/X_{kp})$ be the additive log-ratio transformation of $\bX_k$. Let $\Delta_{kj}$ be the latent variable indicating whether $X_{kj}$ comes from the nondegenerate part, and $\pi_j$ the corresponding probability. In practice, $\Delta_{kj}$ and $\pi_j$ are unknown but can be estimated from the adopted zero-inflated model. The covariance matrix $\bSigma_{(p)}$ of $\bZ_k^{(p)}$ can then be estimated by $\widehat\bSigma_{(p)}=(\hat\sigma_{ij}^{(p)})$ with
\[
\hat\sigma_{ij}^{(p)}=\frac{1}{n}\sum_{k=1}^n\frac{\Delta_{ki}\Delta_{kj}}{\pi_{ij}}Z_{ki}^{(p)}Z_{kj}^{(p)},
\]
where $\pi_{ij}=\pi_i\pi_j$ if $i\ne j$ and $\pi_i$ otherwise. By the relationship between $\bSigma_c$ and $\bSigma_{(p)}$ \citep[eq.\ (4.28)]{aitchison2003}, we obtain an unbiased estimator of $\bSigma_c$:
\[
\check\bSigma_c=\bF^T\bH^{-1}\widehat\bSigma_{(p)}\bH^{-1}\bF,
\]
where $\bF=(\bI_{p-1},-\bone_{p-1})$ and $\bH=\bI_{p-1}+\bone_{p-1}\bone_{p-1}^T$. Now we can substitute $\check\bSigma_c$ for $\widehat\bSigma_c$ in the optimization problem \eqref{eq:opt_col}. The performance of this extension, of course, hinges on the quality of the estimates of $\Delta_{kj}$ and $\pi_j$. A detailed investigation is needed but beyond the scope of this article.

\subsection{Limitations}
Despite the extensions outlined above, important limitations of our method exist. First, since structural zeros cannot be handled by the log transformation per se, our method should be combined with a zero-inflated model to estimate the covariance structure of its nondegenerate part. More generally, the covariance structure of the zero part and the interaction between the zero part and the nondegenerate part are also of interest, whose modeling remains an open problem. Second, our measurement error framework for dealing with sampling zeros requires the perturbation term $\bE$ to be sufficiently small, which may not be satisfied when the true proportions are extremely small and the total counts are limited. Developing more robust and efficient methods for these scenarios is an interesting topic.

\section{Simulation Studies}
In this section, we examine the numerical performance of the CARE estimator through simulation studies under different scenarios and compare it with a variety of existing methods. To implement our procedure, by splitting $\bomega_j$ into a positive part $\bomega_j^+$ and a negative part $\bomega_j^-$, we reformulate problem \eqref{eq:opt_col} as the following linear programming problem:
\begin{gather*}
\begin{aligned}
&\text{minimize }\bone_p^T(\bomega_j^++\bomega_j^-)\quad\text{subject to}\\
&\Biggl(\begin{matrix}
\widehat\bSigma_c&-\widehat\bSigma_c\\
-\widehat\bSigma_c&\widehat\bSigma_c
\end{matrix}\Biggr)
\Biggl(\begin{matrix}
\bomega_j^+\\
\bomega_j^-
\end{matrix}\Biggr)
\le
\Biggl(\begin{matrix}
\lambda_j\bone_p+\be_j-\bone_p/p\\
\lambda_j\bone_p-\be_j+\bone_p/p
\end{matrix}\Biggr),
\end{aligned}\\
\bomega_j^+\ge\bzero,\quad\bomega_j^-\ge\bzero,
\end{gather*}
where the inequalities apply componentwise. These problems have the same size as, and hence comparable computational efficiency to, those in the CLIME procedure, and can be solved by many existing algorithms. Here we use an implementation of the parametric simplex method \citep{vanderbei2020}, which produces the entire solution path by solving the linear programming problem only once. Compared with interior-point and simplex methods, the parametric simplex method is more efficient for large-scale sparse learning problems \citep{pang2017}.

We compare our estimator to three previously proposed methods: CD-trace \citep{yuan2019}, gCoda \citep{fang2017}, and SPIEC-EASI with the graphical lasso \citep{kurtz2015}. The oracle estimator, which is obtained by applying the CLIME approach to the basis as it were observed, is also included for comparison. This ideal estimator serves as the benchmark for assessing whether the compositionality effect has been removed.

\subsection{Simulation Results for Compositional Data}
We first consider the case with positive compositional data and no zero correction is needed. The compositional data were generated as follows. First, we independently sampled $\bY_k=(Y_{k1},\dots,Y_{kp})^T$, $k=1,\dots,n$, from a multivariate normal distribution $N_p(\bzero,\bOmega_0^{-1})$. Then we generated $\bW_k=(W_{k1},\dots,W_{kp})^T$ and $\bX_k=(X_{k1},\dots,X_{kp})^T$, $k=1,\dots,n$, through the transformations $W_{kj}=\exp(Y_{kj})$ and $X_{kj}=W_{kj}/\sum_{i=1}^pW_{ki}$ for $j=1,\dots,p$. As a result, $\bW_k$ and $\bX_k$ follow multivariate log-normal and logistic-normal distributions, respectively. We considered the following four models for generating $\bOmega_0$ with diverse network structures.

\begin{compactenum}[(a)]
\item \emph{Band graph}: Let $\bOmega_1=(\omega_{ij}^1)$ and $\bOmega_0=\bOmega_1+(|\lambda_{\min}(\bOmega_1)|+0.01)\bI_p$, where
$\omega_{i,i+1}^1=\omega_{i+1,i}^1=0.8$, $\omega_{i,i+2}^1=\omega_{i+2,i}^1=0.5$, $\omega_{ij}^1=0$ for $|i-j|\ge3$, and $\omega_{ii}^1$ were drawn from a uniform distribution $U(1,2)$.
\item \emph{Hub graph}: The $p$ nodes were divided into blocks of size 5. For each block, one hub was selected and connected to the other nodes in the same block, and each edge weight was set to 0.8 or 0.5 with equal probability. The diagonal entries were set as in model (a).
\item \emph{Block graph}: The $p$ nodes were equally divided into 5 blocks. Each pair of nodes in the same block were connected with probability $20/p$, and each edge weight was set as in model (b). The diagonal entries were set large enough as in model (a) so that  $\bOmega_0$ was positive definite.
\item \emph{Random graph}: Each pair of nodes were connected with probability $4/p$. The edge weights were set as in model (b) and the diagonal entries were set as in model (a).
\end{compactenum}

Throughout the simulations, we took the sample size $n=200$ and dimension $p=50,100,200,400$. Five performance measures are adopted: the spectral norm, matrix $\ell_1$-norm and Frobenius norm losses for evaluating the estimation accuracy, and the true positive and false positive rates for assessing the support recovery performance. For the CARE and oracle methods, the tuning parameters $\lambda_j$ were chosen by fivefold cross-validation.

\begin{table}
\caption{Means and standard errors (in parentheses) of performance measures for different methods in model (a) over 100 replications.}\label{tab:model_a}
\def~{\phantom{0}}
\begin{tabular*}{\textwidth}{@{}c*{5}{@{\extracolsep{\fill}}c}@{}}
\hline
& \multicolumn{5}{c}{Method}\\
\cline{2-6}
$p$ & CARE & Oracle & CD-trace & gCoda & SPIEC-EASI\\
\hline
\multicolumn{6}{c}{Spectral norm loss}\\
~50 & ~2.46 (0.12) & ~2.29 (0.12) & ~3.21 (0.03) & ~3.47 (0.05) & ~3.35 (0.04)\\
100 & ~2.70 (0.10) & ~2.59 (0.10) & ~3.24 (0.02) & ~3.62 (0.04) & ~3.44 (0.02)\\
200 & ~2.96 (0.07) & ~2.91 (0.08) & ~3.27 (0.02) & ~3.70 (0.02) & ~3.37 (0.02)\\
400 & ~3.25 (0.03) & ~3.24 (0.04) & ~3.34 (0.02) & ~3.76 (0.02) & ~3.50 (0.03)\\
\multicolumn{6}{c}{Matrix $\ell_1$-norm loss}\\
~50 & ~3.28 (0.20) & ~3.05 (0.23) & ~3.53 (0.06) & ~3.86 (0.06) & ~3.93 (0.06)\\
100 & ~3.45 (0.15) & ~3.36 (0.16) & ~3.61 (0.07) & ~3.98 (0.05) & ~3.89 (0.03)\\
200 & ~3.61 (0.12) & ~3.56 (0.11) & ~3.64 (0.05) & ~4.07 (0.03) & ~3.91 (0.03)\\
400 & ~3.77 (0.07) & ~3.76 (0.07) & ~3.76 (0.07) & ~4.18 (0.03) & ~3.99 (0.03)\\
\multicolumn{6}{c}{Frobenius norm loss}\\
~50 & ~6.79 (0.29) & ~6.31 (0.25) & ~9.83 (0.08) & 10.93 (0.21) & 10.37 (0.13)\\
100 & 10.30 (0.25) & ~9.84 (0.27) & 14.22 (0.07) & 16.35 (0.24) & 14.87 (0.04)\\
200 & 16.68 (0.19) & 16.33 (0.19) & 20.21 (0.07) & 23.98 (0.18) & 20.59 (0.07)\\
400 & 27.04 (0.15) & 26.85 (0.16) & 29.12 (0.12) & 34.50 (0.18) & 29.71 (0.33)\\
\multicolumn{6}{c}{True positive rate (\%)}\\
~50 &  90.5 (2.6)  &  94.5 (2.3)  &  74.9 (2.8)  &  74.5 (6.4)  &  54.8 (4.0)\\
100 &  89.1 (2.2)  &  91.8 (2.0)  &  60.4 (2.4)  &  49.2 (5.8)  &  65.6 (2.3)\\
200 &  83.7 (1.6)  &  85.9 (1.4)  &  60.9 (1.5)  &  39.6 (3.5)  &  66.9 (1.6)\\
400 &  68.2 (1.4)  &  69.6 (1.5)  &  53.0 (1.1)  &  34.2 (2.2)  &  65.4 (3.4)\\
\multicolumn{6}{c}{False positive rate (\%)}\\
~50 &  ~7.4 (0.7)  &  ~6.0 (0.7)  &  ~3.6 (0.2)  &  ~9.3 (1.0)  &  ~5.9 (0.4)\\
100 &  ~3.6 (0.3)  &  ~3.3 (0.2)  &  ~1.1 (0.1)  &  ~2.9 (0.4)  &  ~2.3 (0.1)\\
200 &  ~1.3 (0.1)  &  ~1.2 (0.1)  &  ~0.5 (0.0)  &  ~1.2 (0.1)  &  ~1.2 (0.1)\\
400 &  ~0.4 (0.0)  &  ~0.4 (0.0)  &  ~0.2 (0.0)  &  ~0.5 (0.1)  &  ~0.5 (0.1)\\
\hline
\end{tabular*}
\end{table}

\begin{table}
\caption{Means and standard errors (in parentheses) of performance measures for different methods in model (b) over 100 replications.}\label{tab:model_b}
\def~{\phantom{0}}
\begin{tabular*}{\textwidth}{@{}c*{5}{@{\extracolsep{\fill}}c}@{}}
\hline
& \multicolumn{5}{c}{Method}\\
\cline{2-6}
$p$ & CARE & Oracle & CD-trace & gCoda & SPIEC-EASI\\
\hline
\multicolumn{6}{c}{Spectral norm loss}\\
~50 & ~2.04 (0.17) & ~1.91 (0.20) & ~2.29 (0.11) & ~3.47 (0.07) & ~3.68 (0.09)\\
100 & ~1.99 (0.12) & ~1.92 (0.14) & ~2.23 (0.05) & ~3.77 (0.12) & ~3.55 (0.12) \\
200 & ~2.23 (0.15) & ~2.16 (0.12) & ~2.83 (0.09) & ~3.82 (0.06) & ~3.75 (0.04)\\
400 & ~2.44 (0.15) & ~2.36 (0.14) & ~2.94 (0.07) & ~4.02 (0.01) & ~3.43 (0.02) \\
\multicolumn{6}{c}{Matrix $\ell_1$-norm loss}\\
~50 & ~3.60 (0.35) & ~3.17 (0.43) & ~3.77 (0.22) & ~6.03 (0.18) & ~6.35 (0.16)\\
100 & ~4.22 (0.47) & ~3.76 (0.56) & ~4.38 (0.35) & ~7.97 (0.30) & ~7.21 (0.34)\\
200 & ~4.59 (0.42) & ~4.35 (0.33) & ~6.39 (0.23) & ~7.49 (0.06) & ~7.05 (0.11)\\
400 & ~5.69 (0.50) & ~5.48 (0.44) & ~7.15 (0.17) & ~8.69 (0.03) & ~6.91 (0.09)\\
\multicolumn{6}{c}{Frobenius norm loss}\\
~50 & ~5.49 (0.24) & ~5.18 (0.23) & ~6.13 (0.16) & 11.27 (0.31) & 13.16 (0.33)\\
100 & ~7.78 (0.23) & ~7.54 (0.24) & ~9.03 (0.14) & 16.47 (0.45) & 17.84 (0.43)\\
200 & 11.68 (0.25) & 11.45 (0.27) & 13.97 (0.13) & 25.03 (0.75) & 26.94 (0.04)\\
400 & 16.84 (0.12) & 16.69 (0.14) & 20.18 (0.10) & 34.12 (0.08) & 35.95 (0.02)\\
\multicolumn{6}{c}{True positive rate (\%)}\\
~50 &  87.6 (3.8)  &  89.4 (3.8)  &  79.7 (3.3)  &  56.9 (2.0)  &  54.4 (4.6)\\
100 &  85.8 (3.3)  &  87.4 (3.0)  &  65.6 (3.1)  &  30.5 (2.6)  &  63.6 (8.9)\\
200 &  81.7 (2.4)  &  83.2 (2.3)  &  53.4 (2.4)  &  14.8 (2.2)  &  71.9 (2.2)\\
400 &  82.4 (1.5)  &  83.1 (1.5)  &  48.5 (1.1)  &  10.2 (0.4)  &  90.4 (1.1)\\
\multicolumn{6}{c}{False positive rate (\%)}\\
~50 &  ~2.1 (0.5)  &  ~1.4 (0.5)  &  ~2.1 (0.3)  &  ~4.4 (0.4)  &  10.0 (1.8)\\
100 &  ~1.2 (0.2)  &  ~1.0 (0.2)  &  ~0.4 (0.1)  &  ~1.1 (0.1)  &  ~4.5 (0.4)\\
200 &  ~0.7 (0.1)  &  ~0.7 (0.1)  &  ~0.1 (0.0)  &  ~0.3 (0.0)  &  ~1.1 (0.1)\\
400 &  ~0.2 (0.0)  &  ~0.2 (0.0)  &  ~0.0 (0.0)  &  ~0.1 (0.0)  &  ~0.5 (0.0)\\
\hline
\end{tabular*}
\end{table}

The simulation results for models (a) and (b) over 100 replications are reported in Tables \ref{tab:model_a} and \ref{tab:model_b}, and those for models (c) and (d) in Supplementary Tables \ref{tab:model_c} and \ref{tab:model_d}. Overall, we see that the CARE and oracle methods perform nearly equally well and outperform the other three competitors by a large margin in almost all settings. In particular, the performance of CARE tends to be closer to that of the oracle estimator in higher dimensions, confirming the blessing of dimensionality revealed by our theory. Among the three previously proposed methods, CD-trace seems to outperform gCoda and SPIEC-EASI in terms of estimation accuracy. In addition, SPIEC-EASI appears to have some advantages for recovering more edges in high dimensions over CD-trace and gCoda, but it does so at the expense of an inflated false positive rate.

\begin{figure}[!t]
\centering
\includegraphics[width=.25\textwidth]{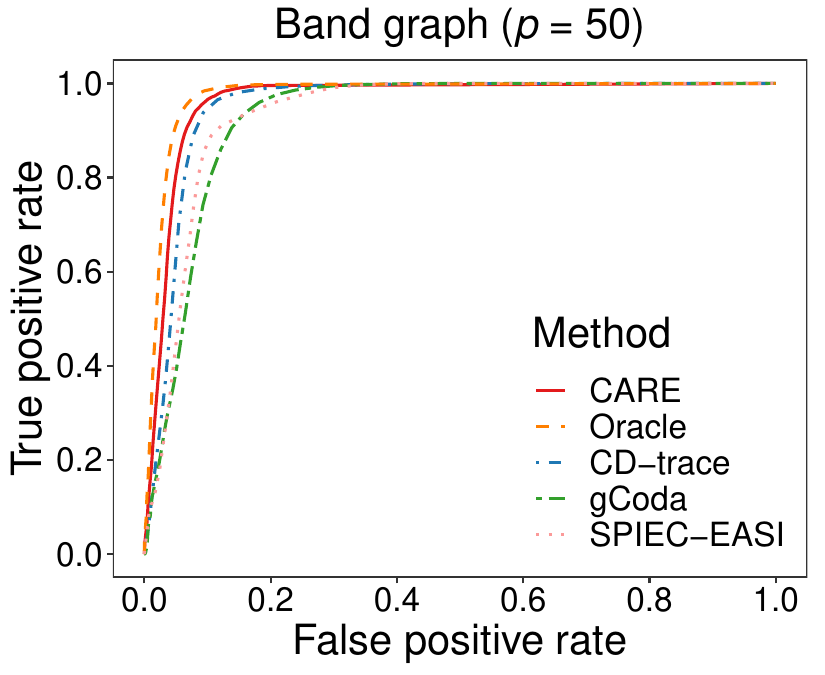}%
\includegraphics[width=.25\textwidth]{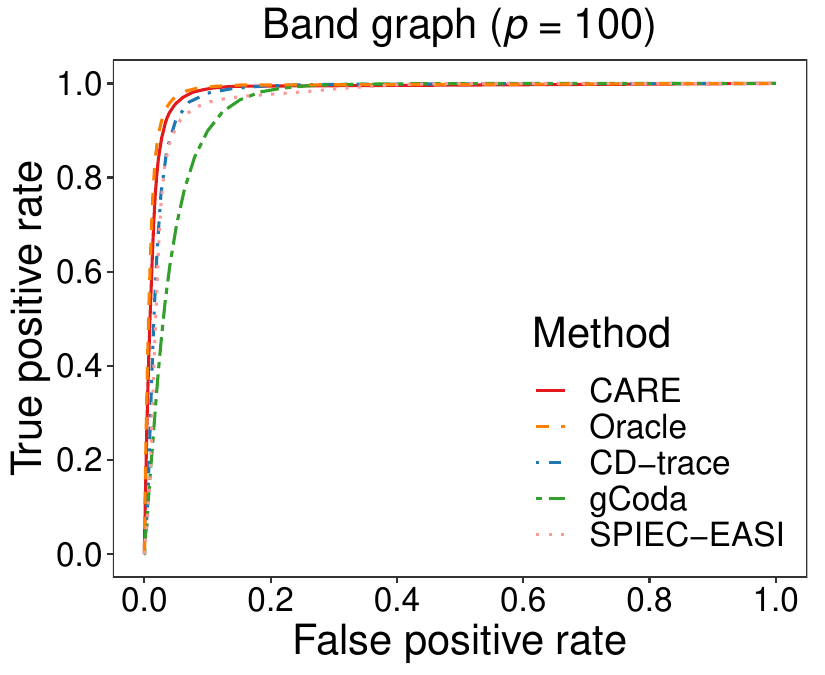}%
\includegraphics[width=.25\textwidth]{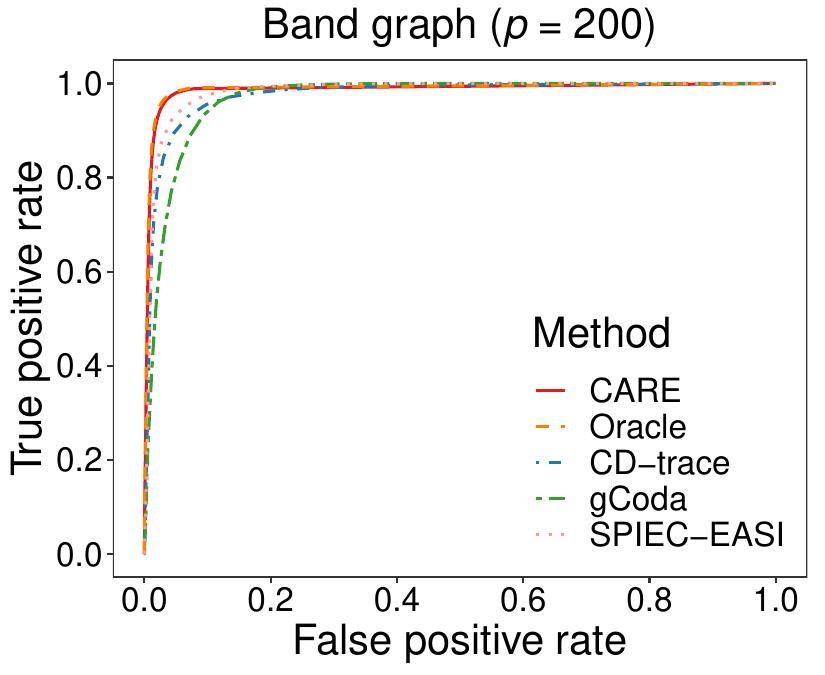}%
\includegraphics[width=.25\textwidth]{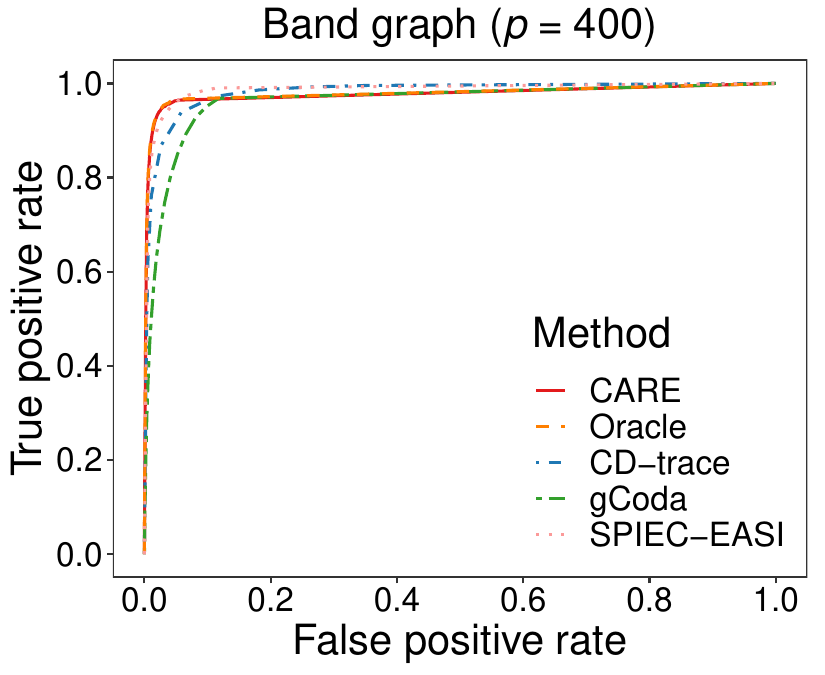}\\
\includegraphics[width=.25\textwidth]{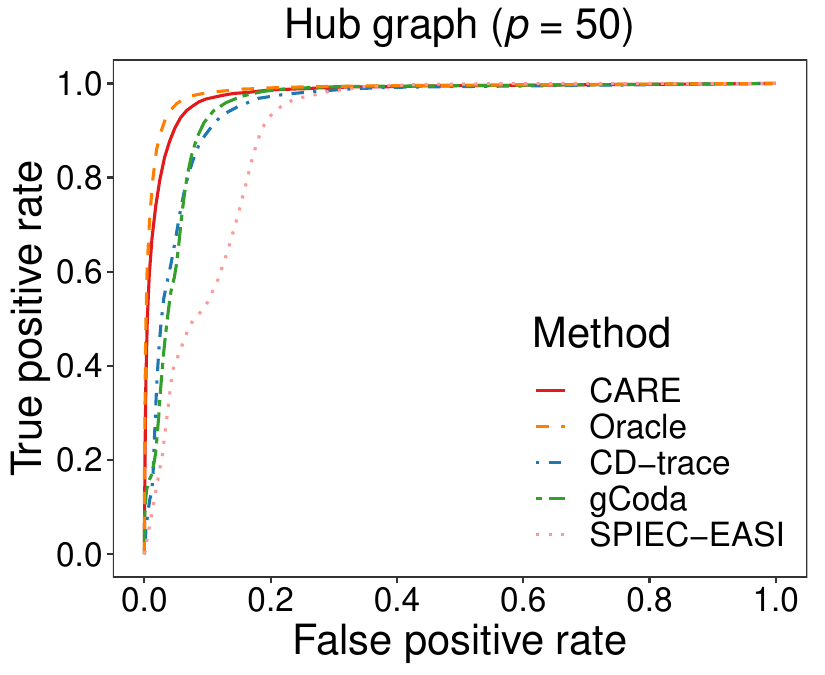}%
\includegraphics[width=.25\textwidth]{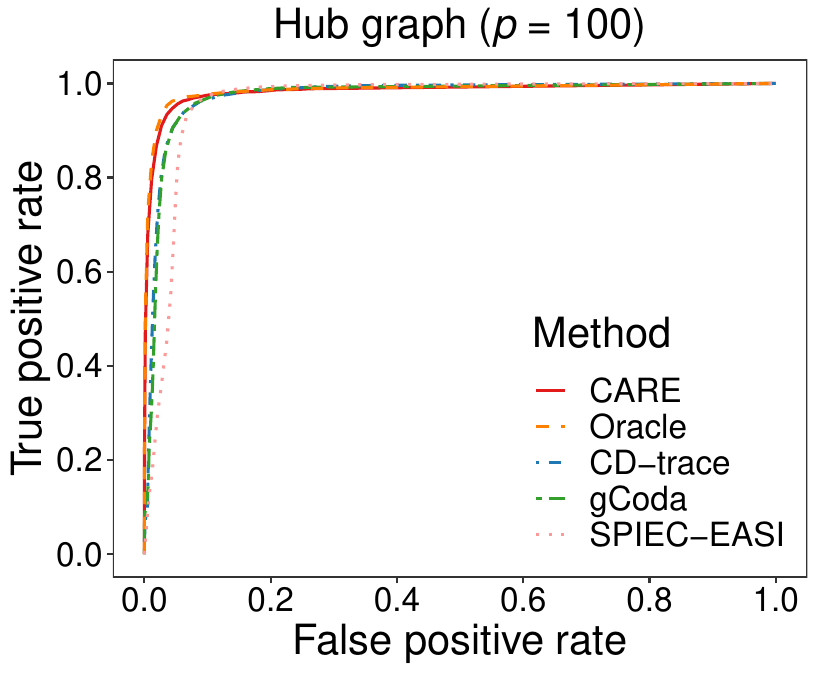}%
\includegraphics[width=.25\textwidth]{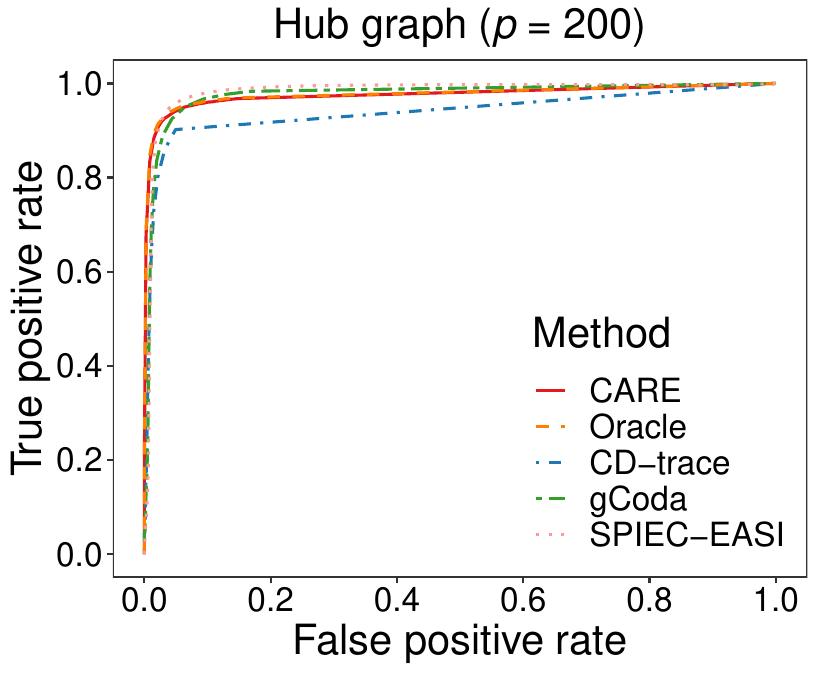}%
\includegraphics[width=.25\textwidth]{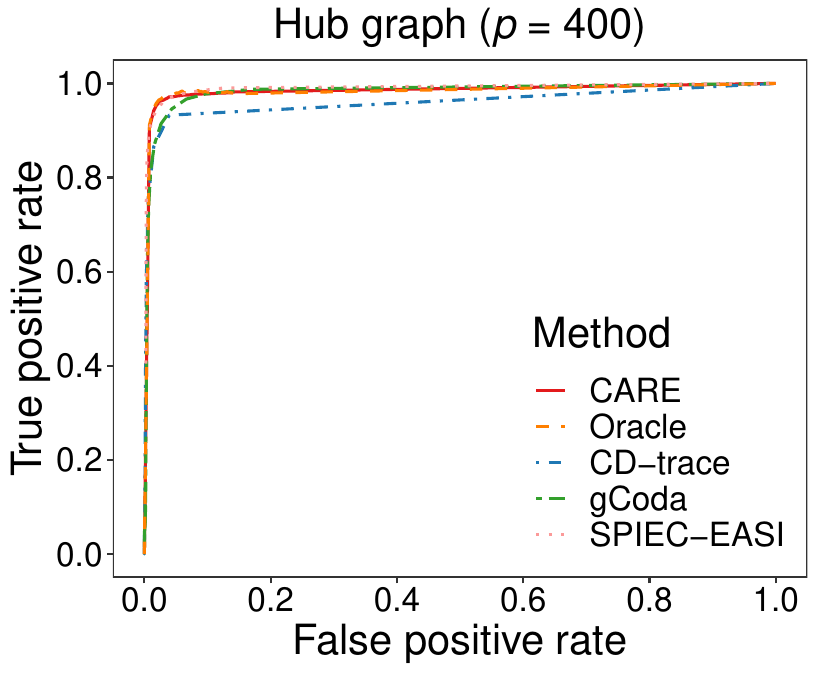}\\
\includegraphics[width=.25\textwidth]{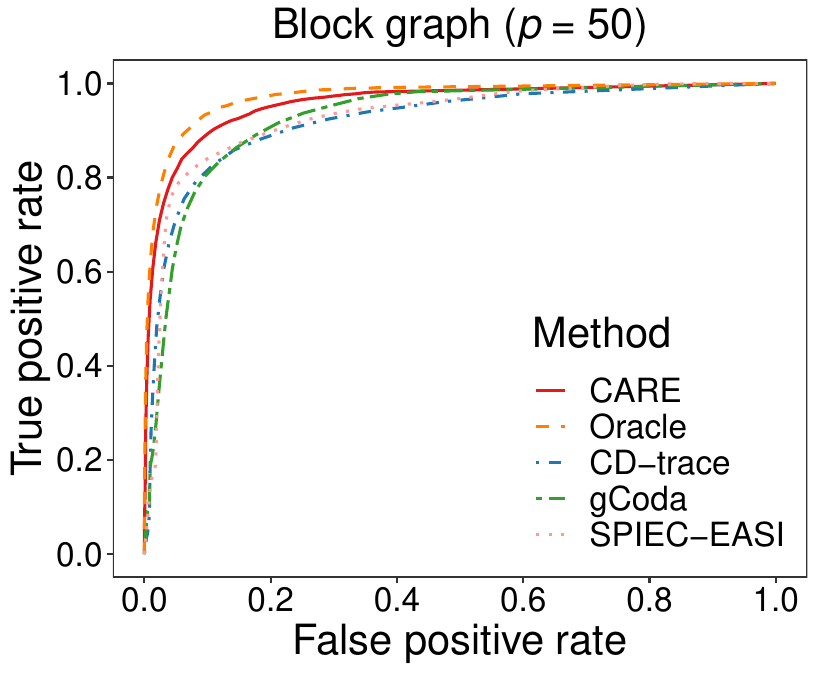}%
\includegraphics[width=.25\textwidth]{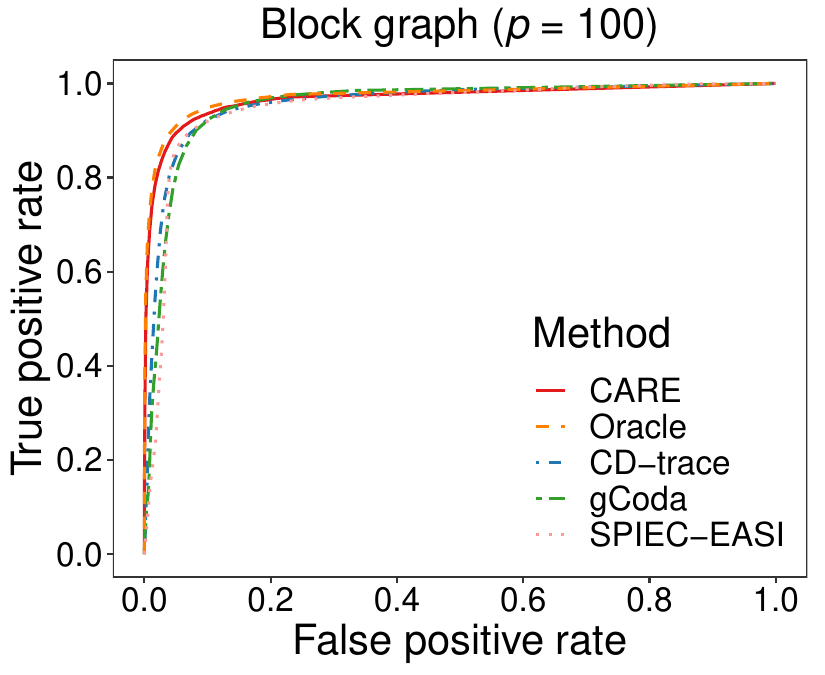}%
\includegraphics[width=.25\textwidth]{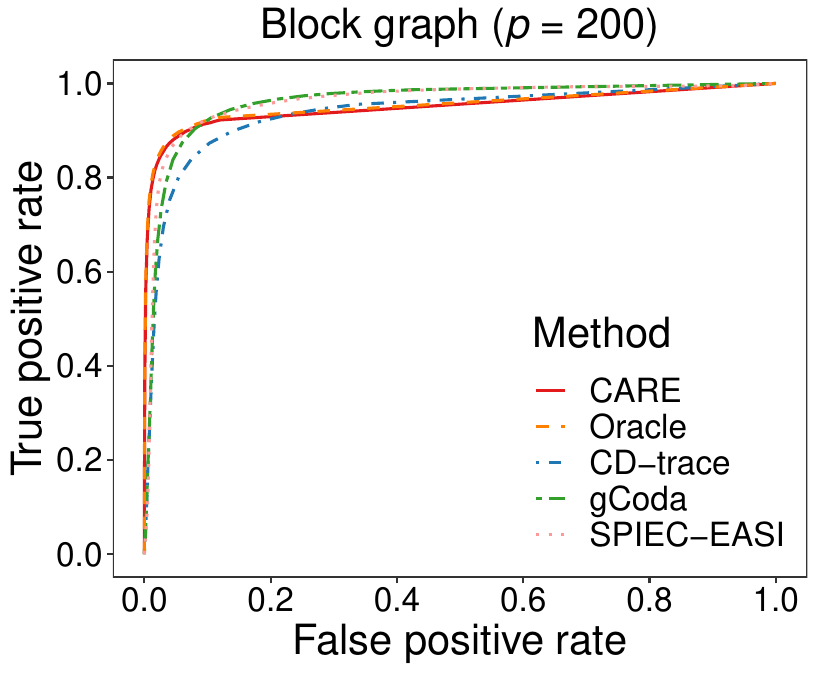}%
\includegraphics[width=.25\textwidth]{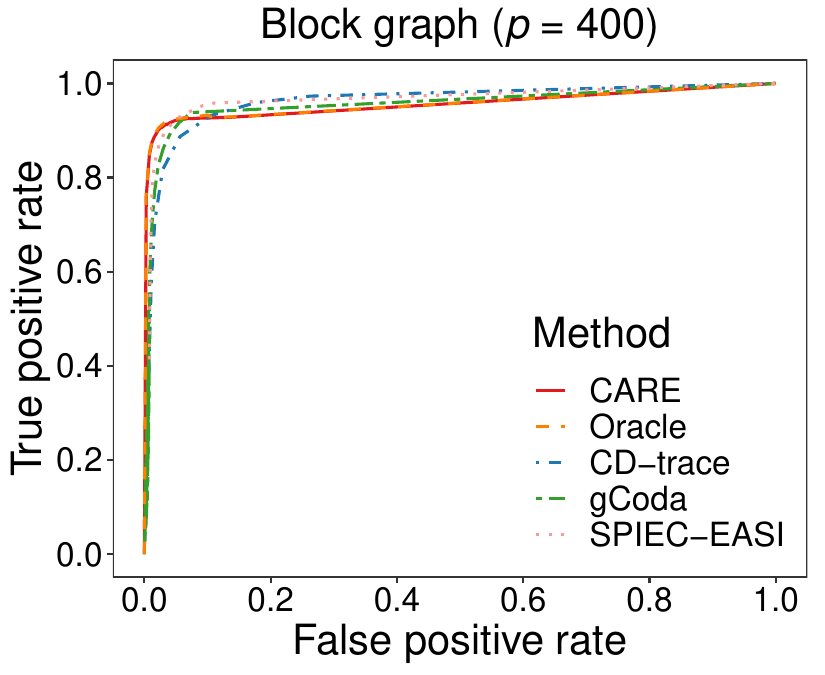}\\
\includegraphics[width=.25\textwidth]{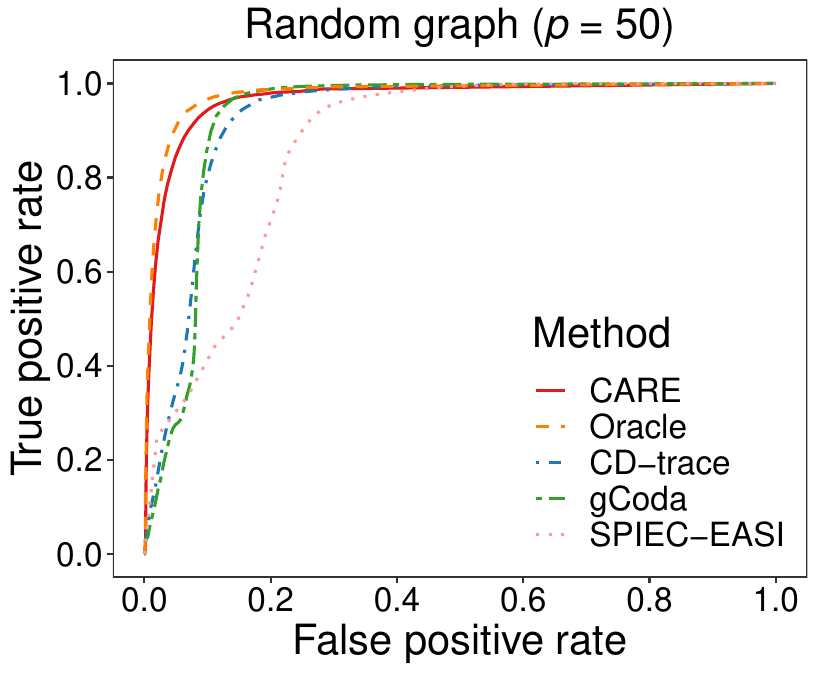}%
\includegraphics[width=.25\textwidth]{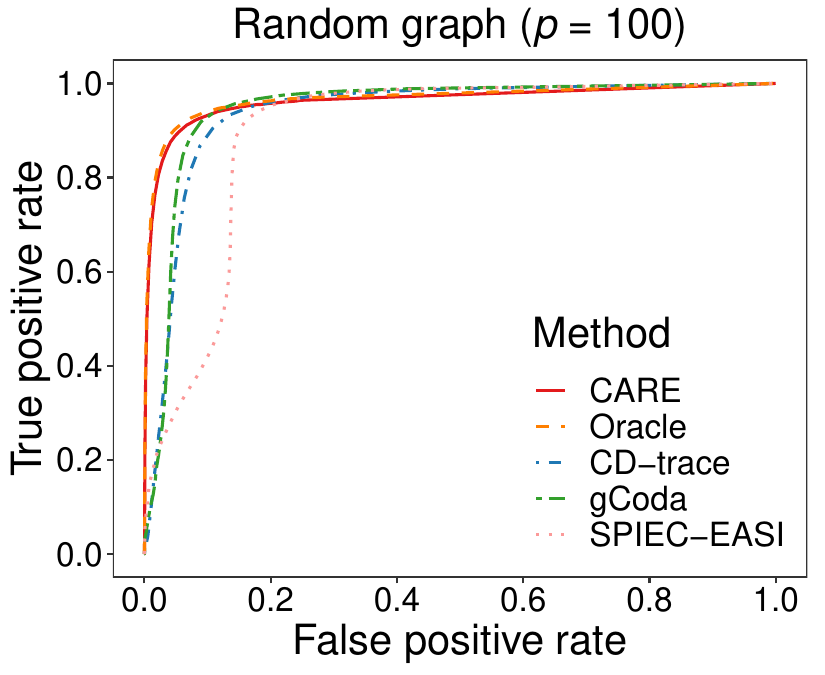}%
\includegraphics[width=.25\textwidth]{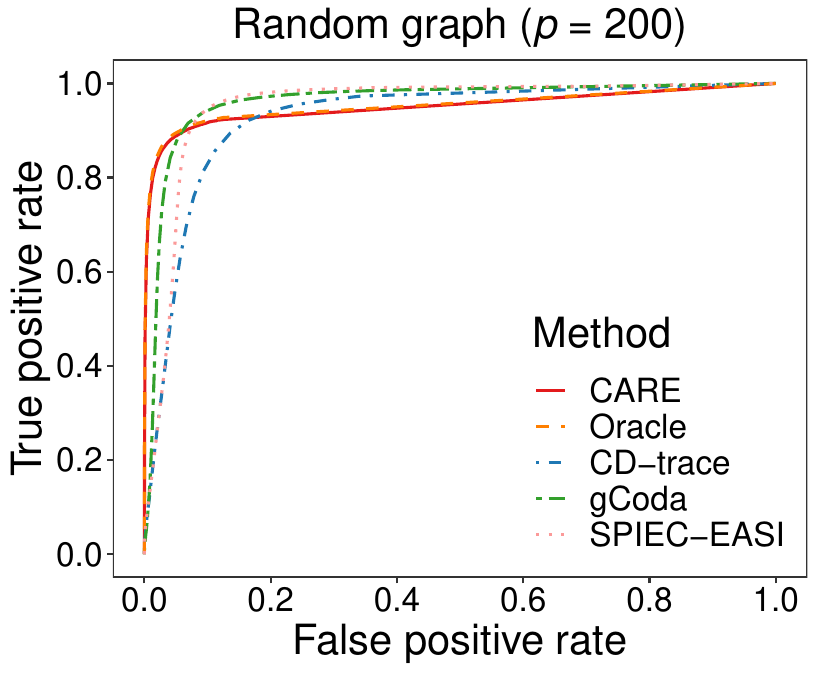}%
\includegraphics[width=.25\textwidth]{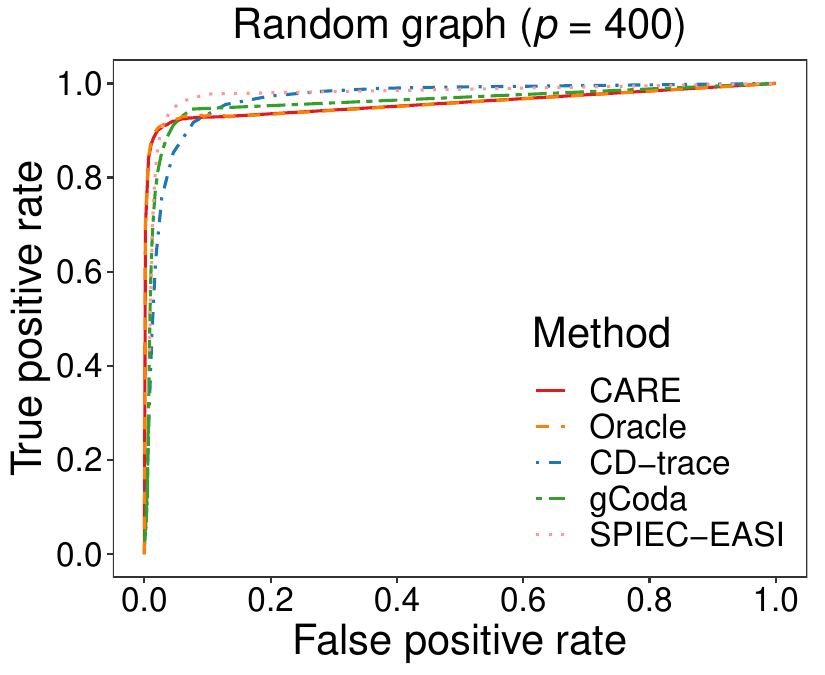}
\caption{The ROC curves for different methods in models (a)--(d) with $p=50,100,200,400$.}\label{fig:roc}
\end{figure}

To compare the support recovery performance without using a specific tuning approach, the receiver operating characteristic (ROC) curves for all methods under models (a)--(d) are shown in Figure \ref{fig:roc}. We observe that the ROC curves of the CARE and oracle methods dominate those of the other three methods uniformly in low to moderate dimensions, as well as in high dimensions when the false positive rates are small. These results embody the superiority of CARE in terms of support recovery. Moreover, SPIEC-EASI performs better for models (a) and (c) than for models (b) and (d) in low dimensions, indicating that its performance may depend more critically on the network topology. The higher right tails of ROC curves for the competing methods in certain high-dimensional settings may not be relevant in practice, since only the sparse regime is of interest.

\subsection{Simulation Results for Count Data}
To illustrate how our method work with zeros, we next examine the case where the generated data are counts and include many sampling zeros. We adopted the graph structure of model (a) and the same settings as before for generating $(\bY_k, \bX_k)$, $k=1,\dots,n$, except that the mean vector of $\bY_k$ was drawn from a uniform distribution $U(0,5)$ to allow a higher diversity. The count data were then generated from a multinomial distribution $\text{Mult}(m_k,\bX_k)$, where $m_k$ were sampled uniformly from $15p,\dots,15p+500$ and $p,\dots,2p$, resulting in about 35\% and 70\% zero counts, respectively, for $p=50$ and 100.

The following methods are included in our comparisons: CARE using the compositional data directly (Comp+CARE), CARE, CD-trace, and SPIEC-EASI using composition estimates from the logistic normal multinomial model of \citet{Zhang.Lin2019} (LNM+CARE, LNM+CD-trace, and LNM+SPIEC-EASI, respectively), and CARE using the $+0.5$ variable correction procedure of \citet{shi2022} (VC+CARE). The simulation results reported in Table \ref{tab:count} suggest that the proposed method works reasonably well with count data that contain many zeros, with the performance diminishing as the proportion of zeros increases. Moreover, model-based, more computationally intensive methods for composition estimation such as LNM tend to improve the performance of our method. All combined with LNM, CARE still exhibits a performance gap over CD-trace and SPIEC-EASI by identifying substantially more true positives.

\begin{table}
\caption{Simulation results for count data: means and standard errors (in parentheses) of performance measures for different methods in model (a) over 100 replications.}\label{tab:count}
\def~{\phantom{0}}
\begin{tabular*}{\textwidth}{@{}c*{6}{@{\extracolsep{\fill}}c}@{}}
\hline
\small
& & \multicolumn{5}{c}{Method}\\
\cline{3-7}
$p$ & Zeros & Comp+CARE & LNM+CARE & VC+CARE & LNM+CD-trace & LNM+SPIEC-EASI\\
\hline
\multicolumn{7}{c}{Spectral norm loss}\\
~50 & 35\% & ~2.46 (0.12) & ~3.21 (0.09) & ~3.40 (0.08) & ~3.41 (0.06) & ~3.42 (0.04)\\
	& 70\% & --           & ~3.78 (1.09) & 14.78 (3.80) & ~3.67 (0.79) & ~3.44 (0.03)\\
100 & 35\% & ~2.91 (0.09) & ~3.44 (0.04) & ~3.47 (0.06) & ~3.47 (0.04) & ~3.60 (0.02)\\
	& 70\% & --           & ~3.82 (0.97) & 14.18 (1.59) & ~3.76 (0.88) & ~3.60 (0.01)\\
\multicolumn{7}{c}{Matrix $\ell_1$-norm loss}\\
~50 & 35\% & ~3.23 (0.19) & ~3.83 (0.11) & ~4.07 (0.14) & ~3.82 (0.09) & ~3.97 (0.06)\\
    & 70\% & --           & ~4.61 (1.24) & 17.16 (3.93) & ~4.44 (1.09) & ~3.94 (0.08)\\
100 & 35\% & ~3.60 (0.11) & ~3.99 (0.07) & ~4.09 (0.10) & ~3.87 (0.07) & ~3.98 (0.05)\\
	& 70\% & --           & ~4.52 (1.14) & 16.59 (1.63) & ~4.39 (1.18) & ~3.99 (0.05)\\
\multicolumn{7}{c}{Frobenius norm loss}\\
~50 & 35\% & ~6.79 (0.21) & ~9.43 (0.24) & 10.13 (0.13) & 10.31 (0.12) & 10.74 (0.21)\\
	& 70\% & --           & 10.61 (0.65) & 28.74 (5.87) & 10.66 (0.50) & 10.98 (0.13)\\
100 & 35\% & 11.16 (0.20) & 14.79 (0.20) & 14.81 (0.15) & 15.18 (0.14) & 15.84 (0.10)\\
	& 70\% & --           & 15.58 (0.45) & 42.97 (5.17) & 15.41 (0.43) & 16.04 (0.04)\\
\multicolumn{7}{c}{True positive rate (\%)}\\
~50 & 35\% & 91.2 (3.0)   & 63.8 (5.4)   & 43.3 (5.0)   & 42.7 (4.0)   & 37.9 (7.0)\\
	& 70\% & --           & 24.3 (4.1)   & 10.3 (2.6)   & 15.9 (3.2)   & 17.5 (4.0)\\
100 & 35\% & 83.3 (2.5)   & 47.7 (3.7)   & 32.5 (3.4)   & 28.7 (4.0)   & 28.8 (3.6)\\
	& 70\% & --           & 16.8 (2.7)   & ~6.1 (1.9)   & ~9.4 (1.9)   & 12.1 (2.0)\\
\multicolumn{7}{c}{False positive rate (\%)}\\
~50 & 35\% & ~7.3 (0.8)   & 10.8 (1.2)   & ~8.2 (1.0)   & ~4.0 (0.5)   & ~5.1 (0.8)\\
	& 70\% & --           & ~9.1 (1.2)   & ~4.1 (0.7)   & ~4.0 (0.8)   & ~3.4 (0.8)\\
100 & 35\% & ~2.8 (0.2)   & ~3.7 (0.4)   & ~2.7 (0.4)   & ~1.0 (0.2)   & ~1.6 (0.2)\\
	& 70\% & --           & ~3.6 (0.4)   & ~1.2 (0.2)   & ~1.0 (0.3)   & ~1.4 (0.3)\\
\hline
\end{tabular*}
\end{table}

\section{Application to Gut Microbiome Data}
The gut microbiome is considered to participate in many host physiological processes and have a tremendous impact on human health. Reconstructing ecological networks from metagenomic data holds the potential to reveal complex microbial interaction patterns in the gut. Some studies \citep[e.g.,][]{busiello2017} have found that species interaction networks are sparse and have related the property to explorability and dynamical robustness, thus meeting the requirement of our method. Here we illustrate our method by applying it to a dataset in \citet{wu2011}, which was previously analyzed by \citet{cao2019} using covariance estimation. In this study, DNA from fecal samples of 98 healthy subjects were quantified by 454/Roche pyrosequencing of 16S rRNA gene segments, yielding 87 genera that appeared in at least one sample. These subjects were divided into a lean group of $n=63$ samples with body mass index (BMI) $<25$ and an obese group of $n=35$ samples with BMI $\ge25$. To ensure stable detection of microbial interactions, we filtered out too rare genera and retained $p=40$ genera that were present in at least four samples in each group. We adopted the multisample approach of \citet{cao2020} to deal with excess zeros and obtained a positive composition matrix.

For each group, we estimated the basis precision matrix using different methods. The results are represented as microbial interaction networks among the genera. For the CARE method, the tuning parameters were selected by tenfold cross-validation. To assess the stability of support recovery, we randomly subsampled 80\% of the subjects and recorded the proportion of edges that reoccurred in the subsample. We repeated the subsampling procedure 100 times and adopt the average proportion of reproduced edges as a measure of network stability. Finally, only those edges reproduced in at least 80\% of the subsamples were retained in the networks. The networks identified by the CARE method are shown in Figure \ref{fig:net}, and those by the other methods in Supplementary Figures \ref{fig:CD_net}--\ref{fig:SPIEC_net}. The numbers of positive and negative edges and stability of networks for all methods are summarized in Table \ref{tab:net}.

\begin{figure}
\begin{subfigure}{\textwidth}
\centering
\includegraphics[width=.75\textwidth]{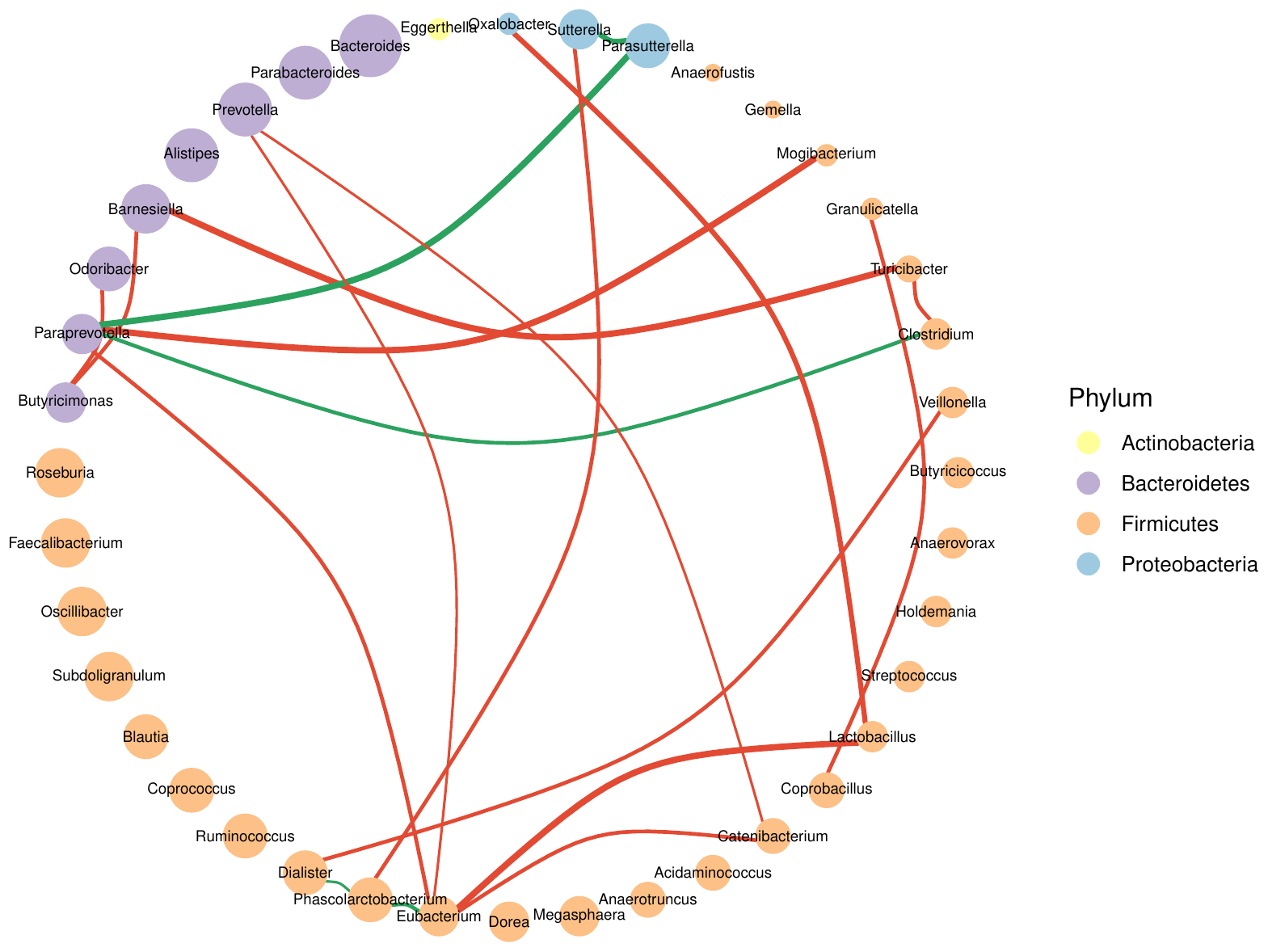}
\caption{Lean}
\end{subfigure}
\begin{subfigure}{\textwidth}
\centering
\includegraphics[width=.75\textwidth]{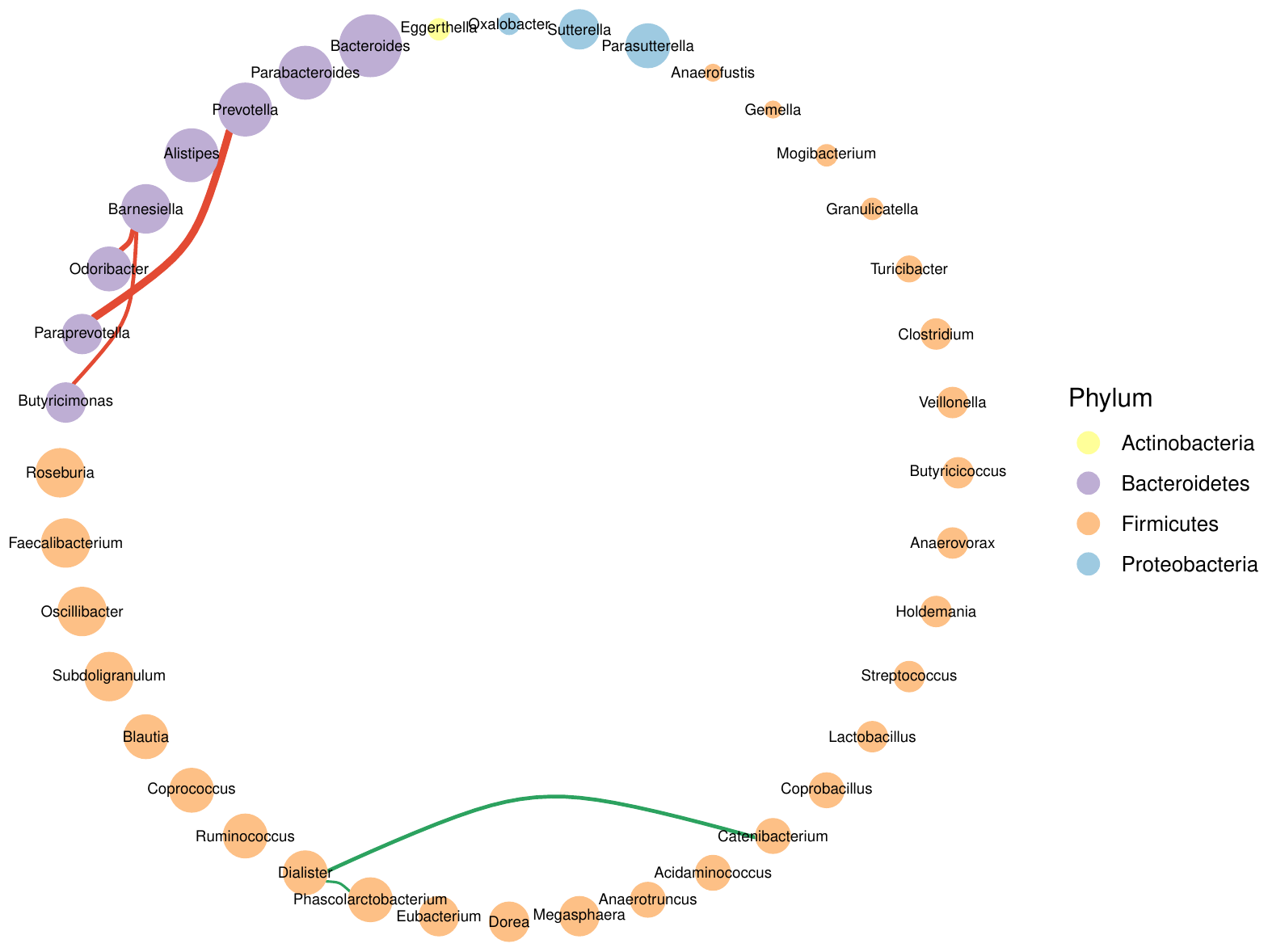}
\caption{Obese}
\end{subfigure}
\caption{Microbial interaction networks identified by the CARE method for the (a) lean and (b) obese groups in the gut microbiome data. Positive and negative edges are displayed in green and red, respectively, with thicknesses proportional to their strengths. Node sizes are proportional to the relative abundances of genera among all samples.}\label{fig:net}
\end{figure}

\begin{table}
\caption{Numbers of positive and negative edges and stability of networks for different methods applied to the gut microbiome data.}\label{tab:net}
\def~{\phantom{0}}
\begin{tabular*}{\textwidth}{@{}l*{5}{@{\extracolsep{\fill}}c}@{}}
\hline
& \multicolumn{2}{@{}c@{}}{Number of all edges} & \multicolumn{2}{@{}c@{}}{Number of stable edges} &\\
\cline{2-3}\cline{4-5}
Method & \multicolumn{1}{c}{Positive} & \multicolumn{1}{c}{Negative} & \multicolumn{1}{c}{Positive} & \multicolumn{1}{c}{Negative} & Network stability\\
\hline
\multicolumn{6}{c}{Lean group}\\
CARE       & ~9 & 19 & 5 & 14 & 0.766\\
CD-trace   & ~9 & 10 & 4 & ~1 & 0.582\\
gCoda      & 12 & 16 & 2 & ~5 & 0.577\\
SPIEC-EASI & 12 & 18 & 4 & ~6 & 0.680\\
\multicolumn{6}{c}{Obese group}\\
CARE       & ~4 & ~8 & 2 & ~3 & 0.708\\
CD-trace   & 16 & 15 & 1 & ~1 & 0.538\\
gCoda      & ~4 & ~6 & 0 & ~3 & 0.636\\
SPIEC-EASI & 13 & ~9 & 1 & ~3 & 0.551\\
\hline
\end{tabular*}
\end{table}

As depicted in Figure \ref{fig:net}, the network structures revealed by the CARE method for the lean and obese groups look markedly different. The genus--genus interactions for the obese group are substantially fewer and less complex than those for the lean group, agreeing with the previous finding that obese microbiomes induce a less modular metabolic network than lean microbiomes \citep{greenblum2012}. From Table \ref{tab:net} we see that CARE maintains the highest network stability and detects more stable edges than the other methods. Interestingly, the networks constructed by CARE involve more negative than positive interactions; a similar phenomenon was observed by \citet{cao2019} in microbial correlation networks and tends to be supported by the ecological theory of microbiome stability \citep{coyte2015}. More discussion on the identified microbial interactions can be found in Supplementary Section \ref{ssec:genera}.

\section{Discussion}
The specification of the compositional precision matrix and the associated tools developed in this article establish a natural link to the basis inverse covariance structure. The blessing of dimensionality derived from this inverse covariance relationship formalizes the intuition that the nonidentifiability due to not knowing the magnitude of the basis spreads over all $p$ components and becomes negligible as the dimensionality grows. These insights are of both theoretical and practical importance, opening up the possibility of investigating existing and developing new methodology for estimating the basis precision matrix. The general idea is to use $\bOmega_c$ as a proxy for $\bOmega_0$ and develop procedures for estimating $\bOmega_0$ through $\bOmega_c$. Although $\bOmega_c$ is not sparse, its matrix $\ell_1$-norm can be tightly controlled, a fact that can be directly exploited by a compositionally adjusted CLIME procedure. Extensions to neighborhood and likelihood-based methods seem possible but remain to be explored. Despite the similarities we have discussed, our problem also differs from general low-rank plus sparse matrix recovery problems in important ways. In particular, the rank in our problem is known and need not be estimated from the data. Also, the usual spiked eigenvalue assumption \citep[e.g.,][]{Fan.etal2013} is not met. These features prevent methods based on nuclear norm regularization or principal components from being applied to our context.

Our conditional dependence modeling approach relies on the log transformation, which is not artificial but natural in at least two ways. First, as a consequence of a large number of independent multiplicative effects, the log-normal distribution is usually a better model than the normal for many real data with small means, large variances, and positive values \citep{limpert2001}. This is, for instance, the case for microbiome data since bacterial abundance generally follows an exponential growth model with additional variability. Second, methods and theory induced by the log transformation, such as multinomial logistic and log-linear models, have been fundamental to categorical data analysis \citep{agresti2013}. In light of the intimate connections between the two areas, the log transformation is a sensible choice that is compatible with many related methods and theory. Nevertheless, it is not the only viable way to deal with the simplex constraint. Notably, square root and power transformations have been applied to compositional data in regression and principal component analysis \citep{Scealy.Welsh2011,Scealy.etal2015}. It would be worthwhile to develop notions of conditional dependence and graphical models under these alternative transformations, which we leave for future work.

\section*{Supplementary Materials}
The supplementary materials contain the proofs of theoretical results, additional discussion and numerical results, and R code and data.

\section*{Acknowledgements}
The authors thank the Associate Editor and two reviewers for valuable comments that have led to an improved article.

\section*{Disclosure Statement}
The authors report there are no competing interests to declare.

\section*{Funding}
Lin's research was supported by National Natural Science Foundation of China (12171012, 12292980, and 12292981). Zhang's research was supported by National Natural Science Foundation of China (12201111 and 12371264) and the Fundamental Research Funds for the Central Universities in UIBE (CXTD14-05).

\bibliographystyle{jasa}
\bibliography{care_ref}

\end{document}


\title{Supplementary Materials for ``CARE: Large Precision Matrix Estimation for Compositional Data''}
\author{Shucong Zhang, Huiyuan Wang, and Wei Lin}
\date{}
\maketitle

In these supplementary materials, we prove all the theoretical results in the main article and present additional discussion and numerical results.

\section{Proofs}

\subsection{Proof of Theorem \ref{thm:rel}}
By the representation of $\bSigma_c$ in \eqref{eq:sigma_c}, we have the decomposition
\begin{align*}
\bSigma_c+\rho\bv_p\bv_p^T&=\bSigma_0-\bv_p\bv_p^T\bSigma_0-\bSigma_0\bv_p\bv_p^T+\bv_p\bv_p^T\bSigma_0\bv_p\bv_p^T+\rho\bv_p\bv_p^T \\
&=\bSigma_0+(\rho\bI_p-\bSigma_0)\bv_p\bv_p^T+\bv_p\bv_p^T\bSigma_0(\bv_p\bv_p^T-\bI_p)\\
&=\bSigma_0+\bu_p\bv_p^T+\bv_p\bz_p^T,
\end{align*}
where $\bu_p=(\rho\bI_p-\bSigma_0)\bv_p$ and $\bz_p=(\bv_p\bv_p^T-\bI_p)\bSigma_0\bv_p$. Let $\bSigma_0^{\dagger}=\bSigma_0+\bu_p\bv_p^T$. Applying the Sherman--Morrison formula twice gives
\[
(\bSigma_0^{\dagger})^{-1}=\bOmega_0-\frac{\bOmega_0\bu_p\bv_p^T\bOmega_0}{1+\bv_p^T\bOmega_0\bu_p} =\bOmega_0-\frac{\rho\bOmega_0\bv_p\bv_p^T\bOmega_0-\bv_p\bv_p^T\bOmega_0}{\rho\bv_p^T\bOmega_0\bv_p}
\]
and
\begin{equation}\label{eq:sherman-morrison}
(\bSigma_c+\rho\bv_p\bv_p^T)^{-1}=(\bSigma_0^{\dagger}+\bv_p\bz_p^T)^{-1}=(\bSigma_0^{\dagger})^{-1}-\frac{(\bSigma_0^{\dagger})^{-1}\bv_p \bz_p^T(\bSigma_0^{\dagger})^{-1}}{1+\bz_p^T(\bSigma_0^{\dagger})^{-1}\bv_p}.
\end{equation}
The terms $(\bSigma_0^{\dagger})^{-1}\bv_p$, $\bz_p^T(\bSigma_0^{\dagger})^{-1}$, and $\bz_p^T(\bSigma_0^{\dagger})^{-1}\bv_p$ are calculated as follows:
\begin{align*}
(\bSigma_0^{\dagger})^{-1}\bv_p&=\bOmega_0\bv_p-\bOmega_0\bv_p+\frac{\bv_p}{\rho}=\frac{\bv_p}{\rho},\\
\bz_p^T(\bSigma_0^{\dagger})^{-1}&=\bv_p^T\bSigma_0\biggl(\frac{\bOmega_0\bv_p\bv_p^T\bOmega_0}{\bv_p^T\bOmega_0\bv_p}-\bOmega_0\biggr)
=\frac{\bv_p^T\bOmega_0}{\bv_p^T\bOmega_0\bv_p}-\bv_p^T,\\
\bz_p^T(\bSigma_0^{\dagger})^{-1}\bv_p&=\frac{1}{\rho}\bz_p^T\bv_p=\frac{1}{\rho}\bv_p^T\bSigma_0(\bv_p\bv_p^T-\bI_p)\bv_p=0.
\end{align*}
Substituting the above expressions into \eqref{eq:sherman-morrison} yields
\begin{align*}
(\bSigma_c+\rho\bv_p\bv_p^T)^{-1}&=\bOmega_0-\frac{\rho\bOmega_0\bv_p\bv_p^T\bOmega_0-\bv_p\bv_p^T\bOmega_0}{\rho\bv_p^T\bOmega_0\bv_p} -\frac{\bv_p}{\rho}\biggl(\frac{\bv_p^T\bOmega_0}{\bv_p^T\bOmega_0\bv_p}-\bv_p^T\biggr)\\
&=\bOmega_0-\frac{\bOmega_0\bone_p\bone_p^T\bOmega_0}{\bone_p^T\bOmega_0\bone_p}+\frac{1}{\rho}\bv_p\bv_p^T.
\end{align*}
Therefore, by \eqref{eq:inverse},
\[
\bOmega_c=(\bSigma_c+\rho\bv_p\bv_p^T)^{-1}-\frac{1}{\rho}\bv_p\bv_p^T=\bOmega_0-\frac{\bOmega_0\bone_p\bone_p^T\bOmega_0} {\bone_p^T\bOmega_0\bone_p},
\]
which completes the proof.

\subsection{Proof of Proposition \ref{prop:omega_c}}
By the spectral decompositions of $\bSigma_c$ and $\bOmega_c$, we have
\[
\bSigma_c\bOmega_c=(\bV,\bv_p)
\Biggl(\begin{matrix}
\bI_{p-1} & \bzero\\
\bzero & 0
\end{matrix}\Biggr)
(\bV,\bv_p)^T=\bV\bV^T=\bI_p-\bv_p\bv_p^T=\bG,
\]
which proves part (a). To show part (b), by Theorem \ref{thm:rel} we have
\begin{align*}
\bG\bOmega_c&=\left(\bI_p-\frac{1}{p}\bone_p\bone_p^T\right)\biggl(\bOmega_0-\frac{\bOmega_0\bone_p\bone_p^T\bOmega_0} {\bone_p^T\bOmega_0\bone_p}\biggr)\\
&=\bOmega_0-\frac{\bOmega_0\bone_p\bone_p^T\bOmega_0}{\bone_p^T\bOmega_0\bone_p}=\bOmega_c
\end{align*}
and
\[
\bOmega_c\bone_p=\biggl(\bOmega_0-\frac{\bOmega_0\bone_p\bone_p^T\bOmega_0}{\bone_p^T\bOmega_0\bone_p}\biggr)\bone_p =\bOmega_0\bone_p-\bOmega_0\bone_p=\bzero.
\]
Part (c) is immediate from Theorem 9.2 of \cite{aitchison2003}. Part (d) follows from the Cauchy--Schwarz inequality, part (a), and the fact that $\bG^2=\bG$:
\begin{align*}
\sigma_{jj}^c\omega_{jj}^c=\|\bSigma_c^{1/2}\be_j\|_2^2\|\bOmega_c^{1/2}\be_j\|_2^2\ge(\be_j^T\bG^{1/2}\be_j)^2=(1-1/p)^2
\end{align*}
for all $j$. This completes the proof.

\subsection{Proof of Proposition \ref{prop:ident}}
By Theorem \ref{thm:rel}, we have
\[
\|\bOmega_0-\bOmega_c\|_{\max}=\frac{\|\bOmega_0\bone_p\bone_p^T\bOmega_0\|_{\max}}{\bone_p^T\bOmega_0\bone_p}.
\]
Since
\[
\lambda_{\min}^2(\bOmega_0)\le\frac{1}{p^2}\bone_p^T(\bOmega_0\bone_p\bone_p^T\bOmega_0)\bone_p \le\|\bOmega_0\bone_p\bone_p^T\bOmega_0\|_{\max}\le\|\bOmega_0\|_{L_1}^2
\]
and
\[
p\lambda_{\min}(\bOmega_0)\le\bone_p^T\bOmega_0\bone_p\le p\lambda_{\max}(\bOmega_0),
\]
it follows that
\[
\frac{R^{-3}}p\le\frac{\lambda_{\min}^2(\bOmega_0)}{p\lambda_{\max}(\bOmega_0)}\le\|\bOmega_0-\bOmega_c\|_{\max}
\le\frac{\|\bOmega_0\|_{L_1}^2}{p\lambda_{\min}(\bOmega_0)}\le\frac{RM_p^2}p.
\]

\subsection{Proof of Proposition \ref{prop:bound}}
By Theorem \ref{thm:rel}, we have
\begin{align*}
\|\bOmega_0-\bOmega_c\|_{L_1}&=\frac{\|\bOmega_0\bone_p\bone_p^T\bOmega_0\|_{L_1}}{\bone_p^T\bOmega_0\bone_p} \le\frac{\|\bOmega_0\bone_p\|_1\|\bOmega_0\bone_p\|_{\infty}}{p\lambda_{\min}(\bOmega_0)}\le\frac{\|\bOmega_0\bone_p\|_2\|\bOmega_0\|_{L_1}}{\sqrt{p}\lambda_{\min}(\bOmega_0)}\\
&\le\frac{\lambda_{\max}(\bOmega_0)\|\bOmega_0\|_{L_1}}{\lambda_{\min}(\bOmega_0)}\le R^2M_p,
\end{align*}
and hence
\[
\|\bOmega_c\|_{L_1}\le\|\bOmega_0\|_{L_1}+\|\bOmega_0-\bOmega_c\|_{L_1}\le(1+R^2)M_p.
\]
The eigenvalue bounds follow from part (c) of Proposition \ref{prop:omega_c} and the assumption that $1/R\le\lambda_{\min}(\bOmega_0)\le\lambda_{\max}(\bOmega_0)\le R$.

\subsection{Proof of Lemma \ref{lem:conc}}
To prepare for the proof of Lemma \ref{lem:conc}, we first show that tail conditions similar to Conditions \ref{cond:subG} and \ref{cond:poly} hold for the centered log-ratio variables. Let $\widetilde\bZ=\bOmega_0\bZ=(\tilde Z_1,\dots,\tilde Z_p)^T$ be the innovated transformation of $\bZ$, and $\widetilde\bSigma_c=\bOmega_0\bSigma_c\bOmega_0=(\tilde\sigma_{ij}^c)$ the covariance matrix of $\widetilde\bZ$.

\begin{lemma}\label{lem:tail}
Suppose that $\bOmega_0\in\cU_q(s_0(p),M_p)$ with $M_p=o(\sqrt p)$.
\begin{compactenum}[(a)]
\item If Condition \ref{cond:subG} holds, then for sufficiently large $p$, there exist some constants $\eta_1>0$ and $K_1>0$ such that
    \[
    E\{\exp(\eta_1Z_j^2/\sigma_{jj}^c)\}\le K_1,\qquad E\{\exp(\eta_1\tilde Z_j^2/\tilde\sigma_{jj}^c)\}\le K_1
    \]
    for all $j=1,\dots,p$.
\item If Condition \ref{cond:poly} holds, then there exists some constant $K_1'>0$ such that
    \[
    E|Z_j|^{4\gamma+4+\ve}\le K_1',\qquad E|\tilde Z_j|^{4\gamma+4+\ve}\le K_1'
    \]
    for all $j=1,\dots,p$.	
\end{compactenum}
\end{lemma}

\begin{proof} To show part (a), first note that for all $j$,
\[
\frac{Z_j^2}{\sigma_{jj}^c}=\frac{(\be_j^T\bG\bY)^2}{\sigma_{jj}^c}=\|\bG\be_j\|_2^2\frac{\Var(\bs_1^T\bY)}{\sigma_{jj}^c}\frac{(\bs_1^T\bY)^2} {\Var(\bs_1^T\bY)}\le\frac{\lambda_{\max}(\bSigma_0)}{\sigma_{jj}^c}\frac{(\bs_1^T\bY)^2}{\Var(\bs_1^T\bY)},
\]
where $\bs_1=\bG\be_j/\|\bG\be_j\|_2$. Using $\bSigma_c=\bG\bSigma_0\bG$, we have
\[
\sigma_{jj}^c=\be_j^T\bG\bSigma_0\bG\be_j\ge\|\bG\be_j\|_2^2\lambda_{\min}(\bSigma_0)=\biggl(1-\frac{1}{p}\biggr)\lambda_{\min}(\bSigma_0).
\]
Hence, for any $p\ge 2$,
\[
\frac{\lambda_{\max}(\bSigma_0)}{\sigma_{jj}^c}\le\frac{\lambda_{\max}(\bSigma_0)}{(1-1/p)\lambda_{\min}(\bSigma_0)}\le 2R^2.
\]
Letting $\eta^*=\eta/(2R^2)$ and using Condition \ref{cond:subG}, we obtain
\[
E\{\exp(\eta^*Z_j^2/\sigma_{jj}^c)\}\le E\biggl\{\frac{\eta(\bs_1^T\bY)^2}{\Var(\bs_1^T\bY)}\biggr\}\le K.
\]

To derive the second inequality, we observe that for all $j$,
\begin{align*}
\frac{\tilde Z_j^2}{\tilde\sigma_{jj}^c}&=\frac{(\be_j^T\bOmega_0\bG\bY)^2}{\tilde\sigma_{jj}^c}=\|\bG\bOmega_0\be_j\|_2^2\frac{\Var(\bs_2^T\bY)} {\tilde\sigma_{jj}^c}\frac{(\bs_2^T\bY)^2}{\Var(\bs_2^T\bY)}\\
&\le\frac{\lambda_{\max}^2(\bOmega_0)\lambda_{\max}(\bSigma_0)}{\tilde\sigma_{jj}^c}\frac{(\bs_2^T\bY)^2}{\Var(\bs_2^T\bY)},
\end{align*}
where $\bs_2=\bG\bOmega_0\be_j/\|\bG\bOmega_0\be_j\|_2$. By the assumption $M_p=o(\sqrt p)$ and for sufficiently large $p$,
\begin{align*}
\tilde\sigma_{jj}^c&=\be_j^T\bOmega_0\bSigma_c\bOmega_0\be_j\ge\|\bG\bOmega_0\be_j\|_2^2\lambda_{\min}(\bSigma_0)\\
&=(\be_j^T\bOmega_0\bOmega_0\be_j-\be_j^T\bOmega_0\bv_p\bv_p^T\bOmega_0\be_j)\lambda_{\min}(\bSigma_0)\\
&\ge\biggl(\frac{1}{R^2}-\frac{RM_p^2}{p}\biggr)\frac{1}{R}\ge\frac{1}{2R^3},
\end{align*}
so that
\[
\frac{\lambda_{\max}^2(\bOmega_0)\lambda_{\max}(\bSigma_0)}{\tilde\sigma_{jj}^c}\le2R^6.
\]
Letting $\eta'=\eta/(2R^6)$ and using Condition \ref{cond:subG}, we conclude that
\[
E\{\exp(\eta'\tilde Z_j^2/\tilde\sigma_{jj}^c)\}\le E\biggl\{\frac{\eta(\bs_2^T\bY)^2}{\Var(\bs_2^T\bY)}\biggr\}\le K.
\]
Taking $\eta_1=\min(\eta^*,\eta')$ and $K_1=K$ proves part (a).

To show part (b), using Condition \ref{cond:poly} and $\lambda_{\max}(\bG)=1$, we have for all $j$,
\[
E|Z_j|^{4\gamma+4+\ve}\le(\|\bG\be_j\|_2)^{4\gamma+4+\ve}E|\bs_1^T\bY|^{4\gamma+4+\ve}\le K'.
\]
Similarly,
\begin{align*}
E|\tilde Z_j|^{4\gamma+4+\ve}&\le(\|\bG\bOmega_0\be_j\|_2)^{4\gamma+4+\ve}E|\bs_2^T\bY|^{4\gamma+4+\ve}\le(\|\bOmega_0\be_j\|_2)^{4\gamma+4+\ve} E|\bs_2^T\bY|^{4\gamma+4+\ve}\\
&\le R^{4\gamma+4+\ve}E|\bs_2^T\bY|^{4\gamma+4+\ve}\le R^{4\gamma+4+\ve}K'.
\end{align*}
Taking $K_1'=R^{4\gamma+4+\ve}K'$, we arrive at part (b) and conclude the proof.
\end{proof}

Next, we show that $\widehat\bSigma_c\bOmega_0$ concentrates around $\bSigma_c\bOmega_0$ with high probability.

\begin{lemma}\label{lem:sigma_c}
Suppose that $\bOmega_0\in\cU_q(s_0(p),M_p)$ with $M_p=o(\sqrt p)$.
\begin{compactenum}[(a)]
\item If Condition \ref{cond:subG} holds, then for any $\xi>0$, there exists some constant $\alpha>0$ depending only on $\xi$, $\eta$, $K$, and $R$ such that
    \[
    P\biggl\{\|(\widehat\bSigma_c-\bSigma_c)\bOmega_0\|_{\max}\le\alpha\sqrt{\frac{\log p}{n}}\biggr\}\ge 1-O(p^{-\xi}).
    \]
\item If Condition \ref{cond:poly} holds, then for any $\xi>0$, there exists some constant $\alpha'>0$ depending only on $\xi$, $\gamma$, $\ve$, $K'$, and $R$ such that
    \[
    P\biggl\{\|(\widehat\bSigma_c-\bSigma_c)\bOmega_0\|_{\max}\le\alpha'\sqrt{\frac{\log p}{n}}\biggr\}\ge 1-O(p^{-\xi/2}+n^{-\ve/8}).
    \]
\end{compactenum}
\end{lemma}

\begin{proof}
To show part (a), we begin with the decomposition
\[
(\widehat\bSigma_c\bOmega_0-\bSigma_c\bOmega_0)_{ij}=\frac{1}{n}\sum_{k=1}^nZ_{ki}\tilde Z_{kj}-E(Z_{ki}\tilde Z_{kj}) -\Biggl(\frac{1}{n}\sum_{k=1}^nZ_{ki}\Biggr)\Biggl(\frac{1}{n}\sum_{k=1}^n\tilde Z_{kj}\Biggr),
\]
which implies
\begin{equation}\label{Sup_align5}
|(\widehat\bSigma_c\bOmega_0-\bSigma_c\bOmega_0)_{ij}|\le\Biggl|\frac{1}{n}\sum_{k=1}^nZ_{ki}\tilde Z_{kj}-E(Z_{ki}\tilde Z_{kj})\Biggr| +\Biggl|\Biggl(\frac{1}{n}\sum_{k=1}^nZ_{ki}\Biggr)\Biggl(\frac{1}{n}\sum_{k=1}^n\tilde Z_{kj}\Biggr)\Biggr|.
\end{equation}
For any $t_1,t_2>0$, we define the following two events:
\begin{align*}
A_{ij}(t_1)&\equiv\Biggl\{\Biggl|\frac{1}{n}\sum_{k=1}^nZ_{ki}\tilde Z_{kj}-E(Z_{ki}\tilde Z_{kj})\Biggr|>t_1\Biggr\},\\
B_{ij}(t_2)&\equiv\Biggl\{\Biggl|\Biggl(\frac{1}{n}\sum_{k=1}^nZ_{ki}\Biggr)\Biggl(\frac{1}{n}\sum_{k=1}^n\tilde Z_{kj}\Biggr)\Biggr|>t_2\Biggr\}.
\end{align*}
For the event $A_{ij}(t_1)$, as in the proof of Lemma A.3 in \cite{bickel2008}, we have
\begin{align*}
P\{A_{ij}(t_1)\}&=P\Biggl\{\Biggl|\sum_{k=1}^n\biggl(\frac{Z_{ki}}{\sqrt{\sigma_{ii}^c}}\frac{\tilde Z_{kj}}{\sqrt{\tilde\sigma_{jj}^c}}-\rho_{ij}^c\biggr) \Biggr|>\frac{nt_1}{\sqrt{\sigma_{ii}^c\tilde\sigma_{jj}^c}}\Biggr\}\\
&\le P\Biggl[\Biggl|\sum_{k=1}^n\biggl\{\biggl(\frac{Z_{ki}}{\sqrt{\sigma_{ii}^c}}+\frac{\tilde Z_{kj}}{\sqrt{\tilde\sigma_{jj}^c}}\biggr)^2-2(1+\rho_{ij}^c) \biggr\}\Biggr|>\frac{2nt_1}{\sqrt{\sigma_{ii}^c\tilde\sigma_{jj}^c}}\Biggr]\\
&\relph{}+P\Biggl[\Biggl|\sum_{k=1}^n\biggl\{\biggl(\frac{Z_{ki}}{\sqrt{\sigma_{ii}^c}}-\frac{\tilde Z_{kj}}{\sqrt{\tilde\sigma_{jj}^c}}\biggr)^2 -2(1-\rho_{ij}^c)\biggr\}\Biggr|>\frac{2nt_1}{\sqrt{\sigma_{ii}^c\tilde\sigma_{jj}^c}}\Biggr],
\end{align*}
where $\rho_{ij}^c=E(Z_{ki}\tilde Z_{kj})/\sqrt{\sigma_{ii}^c\tilde\sigma_{jj}^c}$. Let
\[
U_{kij}=\frac{Z_{ki}}{\sqrt{\sigma_{ii}^c}}+\frac{\tilde Z_{kj}}{\sqrt{\tilde\sigma_{jj}^c}},\qquad
V_{kij}=\frac{Z_{ki}}{\sqrt{\sigma_{ii}^c}}-\frac{\tilde Z_{kj}}{\sqrt{\tilde\sigma_{jj}^c}}.
\]
We first bound the term
\[
P\Biggl[\Biggl|\sum_{k=1}^n\{U_{kij}^2-2(1+\rho_{ij}^c)\}\Biggr|>\frac{2nt_1}{\sqrt{\sigma_{ii}^c\tilde\sigma_{jj}^c}}\Biggr].
\]
Let $\eta_2=\eta_1/4$. By Lemma \ref{lem:tail} and the Cauchy--Schwarz inequality, we have
\begin{align*}
E\{\exp(\eta_2U_{kij}^2)\}&\le E[\exp\{\eta_1Z_{ki}^2/(2\sigma_{ii}^c)\}\exp\{\eta_1\tilde Z_{kj}^2/(2\tilde\sigma_{jj}^c)\}]\\
&\le[E\{\exp(\eta_1Z_{ki}^2/\sigma_{ii}^c)\}]^{1/2}[E\{\exp(\eta_1\tilde Z_{kj}^2/\tilde\sigma_{jj}^c)\}]^{1/2}\le K_1.
\end{align*}
It is easy to show that $E(|U_{kij}|^r)\le K_1\eta_2^{-r/2}\Gamma(r/2+1)$ for $r\ge 1$, where $\Gamma(x)=\int_0^{\infty}u^{x-1}e^{-u}\,du$. Then we obtain
\begin{align*}
&E(\exp[t\{U_{kij}^2-E(U_{kij}^2)\}])=1+\sum_{r=2}^{\infty}\frac{t^r}{r!}E[\{U_{kij}^2-E(U_{kij}^2)\}^r]\\
&\quad\le 1+\sum_{r=2}^{\infty}\frac{|t|^r2^r}{r!}E(U_{kij}^{2r})\le 1+\sum_{r=2}^{\infty}K_1|t|^r2^r\eta_2^{-r} =1+\frac{4K_1\eta_2^{-2}t^2}{1-2|t|\eta_2^{-1}}
\end{align*}
for $|t|<\eta_2/2$, which further implies
\[
E(\exp[t\{U_{kij}^2-E(U_{kij}^2)\}])\le 1+8K_1\eta_2^{-2}t^2\le\exp(8K_1\eta_2^{-2}t^2)\quad\text{for }|t|<\eta_2/4.
\]
Thus, $U_{kij}^2$ is sub-exponential with parameters $a^2=16K_1\eta_2^{-2}$ and $b=4/\eta_2$:
\[
E(\exp[t\{U_{kij}^2-E(U_{kij}^2)\}])\le\exp(a^2t^2/2)\quad\text{for }|t|<1/b.
\]
Now applying Bernstein's inequality yields
\[
P\Biggl[\Biggl|\sum_{k=1}^n\{U_{kij}^2-2(1+\rho_{ij}^c)\}\Biggr|>\frac{2nt_1}{\sqrt{\sigma_{ii}^c\tilde\sigma_{jj}^c}}\Biggr]\le
2\exp\biggl\{-\min\biggl(\frac{2nt_1^2}{a^2\sigma_{ii}^c\tilde\sigma_{jj}^c},\frac{nt_1}{b\sqrt{\sigma_{ii}^c\tilde\sigma_{jj}^c}}\biggr)\biggr\}.
\]
If $t_1\le 2K_1\eta_2^{-1}\sqrt{\sigma_{ii}^c\tilde\sigma_{jj}^c}$, then this becomes
\[
P\Biggl[\Biggl|\sum_{k=1}^n\{U_{kij}^2-2(1+\rho_{ij}^c)\}\Biggr|>\frac{2nt_1}{\sqrt{\sigma_{ii}^c\tilde\sigma_{jj}^c}}\Biggr]\le 2\exp\biggl(-\frac{n\eta_2^2t_1^2}{8K_1\sigma_{ii}^c\tilde\sigma_{jj}^c}\biggr).
\]
The same deviation bound for $V_{kij}^2-2(1-\rho_{ij}^c)$ can be similarly derived. Since $\tilde\sigma_{jj}^c=\be_j^T\bOmega_0\bSigma_c\bOmega_0\be_j\le R^3$, we have, for $t_1\le 2K_1\eta_2^{-1}\sqrt{\sigma_{ii}^c\tilde\sigma_{jj}^c}$,
\[
P\{A_{ij}(t_1)\}\le 4\exp\biggl(-\frac{n\eta_2^2t_1^2}{8K_1\sigma_{ii}^c\tilde\sigma_{jj}^c}\biggr)\le 4\exp\biggl(-\frac{n\eta_2^2t_1^2}{8K_1R^4}\biggr).
\]
For any $\xi>0$, let $t_1=\alpha_1\sqrt{(\log p)/n}$, where $\alpha_1=2\eta_2^{-1}R^2\sqrt{2K_1(\xi+2)}$. Clearly, $t_1\le2K_1\eta_2^{-1}\sqrt{\sigma_{ii}^c\tilde\sigma_{jj}^c}$ for sufficiently large $n$ and $p$. Thus, we obtain
\begin{equation}\label{Sup_align6}
P\Biggl\{\Biggl|\frac{1}{n}\sum_{k=1}^nZ_{ki}\tilde Z_{kj}-E(Z_{ki}\tilde Z_{kj})\Biggr|>\alpha_1\sqrt{\frac{\log p}{n}}\Biggr\}\le 4p^{-\xi-2}.
\end{equation}
For the event $B_{ij}(t_2)$, note that
\begin{align*}
P\{B_{ij}(t_2)\}&=P\Biggl\{\Biggl|\Biggl(\frac{1}{n}\sum_{k=1}^n\frac{Z_{ki}}{\sqrt{\sigma_{ii}^c}}\Biggr)\Biggl(\frac{1}{n}\sum_{k=1}^n \frac{\tilde Z_{kj}}{\sqrt{\tilde\sigma_{jj}^c}}\Biggr)\Biggr|>\frac{t_2}{\sqrt{\sigma_{ii}^c\tilde\sigma_{jj}^c}}\Biggr\}\\
&\le P\Biggl\{\Biggl|\frac{1}{n}\sum_{k=1}^n\frac{Z_{ki}}{\sqrt{\sigma_{ii}^c}}\Biggr|>\frac{t_2^{1/2}}{(\sigma_{ii}^c\tilde\sigma_{jj}^c)^{1/4}} \Biggr\}+P\Biggl\{\Biggl|\frac{1}{n}\sum_{k=1}^n\frac{\tilde Z_{kj}}{\sqrt{\tilde \sigma_{jj}^c}}\Biggr|>\frac{t_2^{1/2}}{(\sigma_{ii}^c \tilde\sigma_{jj}^c)^{1/4}}\Biggr\}\\
&\equiv T_1+T_2.
\end{align*}
We first deal with the term $T_1$. For any $d\in \mathbb{R}$, by Taylor expansion and Lemma \ref{lem:tail},
\begin{align*}
&E\biggl\{\exp\biggl(d\frac{Z_{ki}}{\sqrt{\sigma_{ii}^c}}\biggr)\biggr\}\\
&\quad=1+E\biggl\{\sum_{m=2}^{\infty}\frac{1}{m!}\biggl(\frac{dZ_{ki}}{\sqrt{\sigma_{ii}^c}}\biggr)^m\biggr\}\le 1+\frac{d^2}{2}E\biggl\{ \frac{Z_{ki}^2}{\sigma_{ii}^c}\sum_{m=2}^{\infty}\frac{1}{(m-2)!}\biggl(\frac{dZ_{ki}}{\sqrt{\sigma_{ii}^c}}\biggr)^{m-2}\biggr\}\\
&\quad\le 1+\frac{d^2}{2}E\biggl\{\frac{Z_{ki}^2}{\sigma_{ii}^c}\exp\biggl(\frac{|dZ_{ki}|}{\sqrt{\sigma_{ii}^c}}\biggr)\biggr\}\le 1+\frac{d^2}{2}E\biggl\{\frac{Z_{ki}^2}{\sigma_{ii}^c}\exp\biggl(\frac{d^2}{2\eta_1}+\frac{\eta_1Z_{ki}^2}{2\sigma_{ii}^c}\biggr)\biggr\}\\
&\quad\le 1+\frac{d^2}{\eta_1}\exp\biggl(\frac{d^2}{2\eta_1}\biggr)E\biggl\{\exp\biggl(\frac{\eta_1Z_{ki}^2}{\sigma_{ii}^c}\biggr)\biggr\}
\le 1+\frac{K_1d^2}{\eta_1}\exp\biggl(\frac{d^2}{2\eta_1}\biggr)\\
&\quad\le\exp\biggl(\frac{d^2}{2\eta_1}\biggr)\biggl(1+\frac{K_1d^2}{\eta_1}\biggr)\le\exp\biggl\{\frac{(2K_1+1)d^2}{2\eta_1}\biggr\}.
\end{align*}
Applying Hoeffding's inequality for sub-Gaussian variables yields
\[
T_1\le2\exp\biggl\{-\frac{n\eta_1t_2}{2(2K_1+1)\sqrt{\sigma_{ii}^c\tilde\sigma_{jj}^c}}\biggr\}.
\]
Similarly, we obtain the same bound for the term $T_2$. Therefore,
\[
P\{B_{ij}(t_2)\}\le 4\exp\biggl\{-\frac{n\eta_1t_2}{2(2K_1+1)\sqrt{\sigma_{ii}^c\tilde\sigma_{jj}^c}}\biggr\}\le 4\exp\biggl\{-\frac{n\eta_1t_2}{2(2K_1+1)R^2}\biggr\}.
\]
For any $\xi>0$, let $t_2=\alpha_{2n}\sqrt{(\log p)/n}$, where $\alpha_{2n}=2\eta_1^{-1}R^2(2K_1+1)(\xi+2)\sqrt{(\log p)/n}$. Then it follows that
\begin{equation}\label{Sup_align7}
P\Bigg\{\Biggl|\Biggl(\frac{1}{n}\sum_{k=1}^nZ_{ki}\Biggr)\Biggl(\frac{1}{n}\sum_{k=1}^n\tilde Z_{kj}\Biggr)\Biggr|>\alpha_{2n}\sqrt{\frac{\log p}{n}}\Biggr\}\le 4p^{-\xi-2}.
\end{equation}
Combining \eqref{Sup_align5}--\eqref{Sup_align7}, we conclude that
\[
P\bigg\{|(\widehat\bSigma_c\bOmega_0-\bSigma_c\bOmega_0)_{ij}|>\alpha\sqrt{\frac{\log p}{n}}\text{ for some }1\le i,j\le p\biggr\}\le 8p^{-\xi},
\]
where $\alpha=6\eta_1^{-1}R^2(2K_1+1)(\xi+2)>\alpha_1+\alpha_{2n}$ with $\eta_1=\eta/(2R^6)$ and $K_1=K$. This proves part (a).

To show part (b), let $D_{kij}=Z_{ki}\tilde Z_{kj}-E(Z_{ki}\tilde Z_{kj})$ and $\theta=\max_{1\le i,j\le p}E(D_{kij}^2)$, which, by Lemma \ref{lem:tail}, is bounded by a constant depending only on $\gamma$, $\ve$, and $K_1'$. Also, define
\[
D_{kij}'=D_{kij}I\biggl\{|Z_{ki}\tilde Z_{kj}|\le\sqrt{\frac{n}{(\log p)^3}}\biggr\},\qquad D_{kij}^{*}=D_{kij}-D_{kij}'.
\]
By Bernstein's inequality and some basic calculations, we have
\begin{align}\label{Sup_align8}
&P\Biggl\{\max_{1\le i, j\le p}\Biggl|\sum_{k=1}^nD_{kij}'\Biggr|>\sqrt{(\theta+1)(4+\xi)n\log p}\Biggr\}\notag\\
&\quad\le p^2\max_{1\le i, j\le p}P\Biggl\{\Biggl|\sum_{k=1}^nD_{kij}'\Biggr|>\sqrt{(\theta+1)(4+\xi)n\log p}\Biggr\}\notag\\
&\quad\le 2p^2\exp\Biggl\{-\frac{(\theta+1)(4+\xi)n\log p}{2nE(D_{1ij}')^2+2n\sqrt{(\theta+1)(4+\xi)}/(3\log p)}\Biggr\}\notag\\
&\quad\le 2p^2\exp\Biggl\{-\frac{(\theta+1)(4+\xi)\log p}{2\theta+2\sqrt{(\theta+1)(4+\xi)}/(3\log p)}\Biggr\}=O(p^{-\xi/2}).
\end{align}
Let $\alpha_{3n}=\max_{1\le i, j\le p}E|Z_{ki}\tilde Z_{kj}|I\{|Z_{ki}\tilde Z_{kj}|> \sqrt{n/(\log p)^3}\}$. Then, by the Cauchy--Schwarz inequality and Lemma \ref{lem:tail},
\begin{align}\label{Sup_align9}
\alpha_{3n}&\le \frac{\max_{1\le i, j\le p}E(|Z_{ki}\tilde Z_{kj}|^{2\gamma+2+\ve/2})}{\{n/(\log p)^3\}^{\gamma+1/2+\ve/4}}\notag\\
&\le \frac{\max_{1\le i, j\le p}\{E(|Z_{ki}|^{4\gamma+4+\ve})E(|\tilde Z_{kj}|^{4\gamma+4+\ve})\}^{1/2}}{\{n/(\log p)^3\}^{\gamma+1/2+\ve/4}} \notag\\
&\le K_1'\frac{(\log p)^{3(\gamma+1/2+\ve/4)}}{n^{\ve/4}}\frac{1}{n^{\gamma+1/2}}=o(n^{-\gamma-1/2}).
\end{align}	
Applying Markov's inequality and Lemma \ref{lem:tail} yields
\begin{align}\label{Sup_align10}
&P\Biggl(\max_{1\le i,j\le p}\Biggl|\sum_{k=1}^nD_{kij}^{*}\Biggr|>2n\alpha_{3n}\Biggr)\notag\\
&\quad\le P\Biggl(\max_{1\le i,j\le p}\Biggl|\sum_{k=1}^nZ_{ki}\tilde Z_{kj}I\{|Z_{ki}\tilde Z_{kj}|>\sqrt{n/(\log p)^3}\}\Biggr|>n\alpha_{3n} \Biggr)\notag\\
&\quad\le P\Biggl(\max_{1\le i,j\le p}\sum_{k=1}^n|Z_{ki}\tilde Z_{kj}|I\{Z_{ki}^2+\tilde Z_{kj}^2>2\sqrt{n/(\log p)^3}\}>n\alpha_{3n}\Biggr) \notag\\
&\quad\le npP\biggl(Z_1^2>\sqrt{\frac{n}{(\log p)^3}}\biggr)+npP\biggl(\tilde Z_1^2>\sqrt{\frac{n}{(\log p)^3}}\biggr)\notag\\
&\quad\le \frac{2K_1'np}{n^{\gamma+1}}\frac{(\log p)^{3(\gamma+1+\ve/4)}}{n^{\ve/8}}\frac{1}{n^{\ve/8}}=O(n^{-\ve/8}).
\end{align}
Combining \eqref{Sup_align8} and \eqref{Sup_align10}, it follows that
\begin{equation}\label{Sup_align11}
P\Biggl\{\max_{1\le i, j\le p}\Biggl|\sum_{k=1}^nD_{kij}\Biggr|>\sqrt{(\theta+1)(4+\xi)}\sqrt{\frac{\log p}{n}}+\alpha_{3n}\Biggr\} =O(p^{-\xi/2}+n^{-\ve/8}).
\end{equation}
By the same techniques and Bernstein's inequality, we obtain
\begin{equation}\label{Sup_align12}
P\Biggl\{\max_{1\le i, j\le p} \Biggl|\Biggl(\frac{1}{n}\sum_{k=1}^nZ_{ki}\Biggr)\Biggl(\frac{1}{n}\sum_{k=1}^n\tilde Z_{kj}\Biggr)\Biggr|> \frac{R(4+\xi)\log p}{n}\Biggr\}=O(p^{-\xi/2}+n^{-\ve/8}).
\end{equation}
Combining \eqref{Sup_align9}, \eqref{Sup_align11}, and \eqref{Sup_align12} yields
\[
\|(\widehat\bSigma_c-\bSigma_c)\bOmega_0\|_{\max}\le 3R\sqrt{(\theta+1)(4+\xi)}\sqrt{\frac{\log p}{n}}
\]
with probability at least $1-O(p^{-\xi/2}+n^{-\ve/8})$, which proves part (b) and concludes the proof.
\end{proof}

We now proceed to prove Lemma \ref{lem:conc} with the aid of Lemma \ref{lem:sigma_c}.

\begin{proof}[Proof of Lemma \ref{lem:conc}] Under Condition \ref{cond:subG}, we consider the event
\[
E_1=\biggl\{\|(\widehat\bSigma_c-\bSigma_c)\bOmega_0\|_{\max}\le\alpha\sqrt{\frac{\log p}{n}}\biggr\},
\]
which, by Lemma \ref{lem:sigma_c}, holds with probability at least $1-O(p^{-\xi})$. From part (a) of Proposition \ref{prop:omega_c} we have
\begin{align*}
\|\widehat\bSigma_c\bOmega_0-\bG\|_{\max}&=\|\widehat\bSigma_c\bOmega_0-\bSigma_c\bOmega_c\|_{\max} \\
&\le \|(\widehat\bSigma_c-\bSigma_c)\bOmega_0\|_{\max}+\|\bSigma_c(\bOmega_0-\bOmega_c)\|_{\max} \\
&\equiv T_1+T_2.
\end{align*}	
Note that, on the event $E_1$,
\[
T_1\le\alpha\sqrt{\frac{\log p}{n}}.
\]
To bound the term $T_2$, it follows from $\bSigma_c=\bG\bSigma_0\bG$ and Theorem \ref{thm:rel} that
\begin{align*}
T_2&=\biggl\|(\bG\bSigma_0\bG)\biggl(\frac{\bOmega_0\bone_p\bone_p^T\bOmega_0}{\bone_p^T\bOmega_0\bone_p}\biggr)\biggr\|_{\max} =\biggl\|\frac{1}{p}\bG\bSigma_0\bone_p\bone_p^T\bOmega_0\biggr\|_{\max}\\
&\le\biggl\|\frac{1}{p}\bSigma_0\bone_p\bone_p^T\bOmega_0\biggr\|_{\max}+\biggl\|\frac{1}{p^2}(\bone_p^T\bSigma_0\bone_p)\bone_p\bone_p^T\bOmega_0 \biggr\|_{\max}\\
&\le\frac{1}{p}\|\bSigma_0\|_{L_1}\|\bOmega_0\|_{L_1}+\frac{1}{p}\lambda_{\max}(\bSigma_0)\|\bOmega_0\|_{L_1}\\
&\le\frac{1}{\sqrt p}\lambda_{\max}(\bSigma_0)\|\bOmega_0\|_{L_1}+\frac{1}{p}\lambda_{\max}(\bSigma_0)\|\bOmega_0\|_{L_1}\\
&\le\frac{RM_p}{\sqrt p}+\frac{RM_p}{p}=O\biggl(\frac{M_p}{\sqrt p}\biggr).
\end{align*}	
Combining these two pieces yields, for some constant $C_0>0$,
\[
\|\widehat\bSigma_c\bOmega_0-\bG\|_{\max}\le C_0\biggl(\sqrt{\frac{\log p}{n}}+\frac{M_p}{\sqrt p}\biggr)
\]
with probability at least $1-O(p^{-\xi})$. Under Condition \ref{cond:poly}, we similarly obtain, for some constant $C_0'>0$,
\[
\|\widehat\bSigma_c\bOmega_0-\bG\|_{\max}\le C_0'\biggl(\sqrt{\frac{\log p}{n}}+\frac{M_p}{\sqrt p}\biggr)
\]
with probability at least $1-O(p^{-\xi/2}+n^{-\ve/8})$. This completes the proof of Lemma \ref{lem:conc}.
\end{proof}

\subsection{Proof of Theorem \ref{thm:omega}}
By the definition of $\widetilde\bomega_j$, we have, for all $j$,
\[
\|\widetilde\bomega_j\|_1\le \|\bomega_j^0\|_1\le \|\bOmega_0\|_{L_1},
\]
where $\bomega_j^0$ is the $j$th column of $\bOmega_0$. Thus, $\|\widetilde\bOmega\|_{L_1}\le\|\bOmega_0\|_{L_1}\le M_p$. In view of Lemma \ref{lem:conc}, we condition on the event
\[
E_2=\biggl\{\|\widehat\bSigma_c\bOmega_0-\bG\|_{\max}\le C_0\biggl(\sqrt{\frac{\log p}{n}}+\frac{M_p}{\sqrt p}\biggr)\biggr\},
\]
which holds with probability at least $1-O(p^{-\xi})$ under Condition \ref{cond:subG}. Choosing $\lambda_j\asymp\sqrt{(\log p)/n}+M_p/\sqrt p$ for all $j$, we have
\begin{align*}
&\|\widetilde\bOmega-\bOmega_0\|_{\max}\\
&\quad\le\|(\bG-\bOmega_0\widehat\bSigma_c)\widetilde\bOmega\|_{\max}+\|\bOmega_0(\widehat\bSigma_c\widetilde\bOmega-\bG)\|_{\max} +\biggl\|\frac{1}{p}\bone_p\bone_p^T\widetilde\bOmega\biggr\|_{\max}+\biggl\|\frac{1}{p}\bOmega_0\bone_p\bone_p^T\biggr\|_{\max}\\
&\quad\le\|\widehat\bSigma_c\bOmega_0-\bG\|_{\max}\|\widetilde\bOmega\|_{L_1}+\|\widehat\bSigma_c\widetilde\bOmega-\bG\|_{\max}\|\bOmega_0\|_{L_1} +\frac{1}{p}\|\widetilde\bOmega\|_{L_1}+\frac{1}{p}\|\bOmega_0\|_{L_1}\\
&\quad\le C_0\biggl(M_p\sqrt{\frac{\log p}{n}}+\frac{M_p^2}{\sqrt p}\biggr)+\Bigl(\max_{1\le j\le p}\lambda_j\Bigr)M_p+\frac{2M_p}{p}\\
&\quad=O\biggl(M_p\sqrt{\frac{\log p}{n}}+\frac{M_p^2}{\sqrt p}\biggr).
\end{align*}
Clearly, the same bound holds for the symmetrized version $\widehat\bOmega$, that is, for some constant $C_1>0$,
\begin{equation}\label{Sup_align13}
\|\widehat\bOmega-\bOmega_0\|_{\max}\le C_1\biggl(M_p\sqrt{\frac{\log p}{n}}+\frac{M_p^2}{\sqrt p}\biggr)
\end{equation}
with probability at least $1-O(p^{-\xi})$. Let $c_{np}=\|\widehat\bOmega-\bOmega_0\|_{\max}$ and $\bd_j=\widehat\bomega_j-\bomega_j^0=\bd_j^{(1)}+\bd_j^{(2)}$, where $\widehat\bomega_j$ is the $j$th column of $\widehat\bOmega$, and $\bd_j^{(k)}=(d_{1j}^{(k)},\dots,d_{pj}^{(k)})^T$, $k=1,2$, with
\begin{align*}
d_{ij}^{(1)}&=\hat\omega_{ij}I(|\hat\omega_{ij}|>2c_{np})-\omega_{ij}^0,\\
d_{ij}^{(2)}&=\hat\omega_{ij}I(|\hat\omega_{ij}|\le2c_{np}).
\end{align*}
It follows from the definition of $\widehat\bOmega$ that
\begin{align*}
\|\bd_j^{(2)}\|_1-\|\bd_j^{(1)}\|_1&\le\sum_{i=1}^p|\hat\omega_{ij}|I(|\hat\omega_{ij}|\le2c_{np})+\sum_{i=1}^p |\hat\omega_{ij}|I(|\hat\omega_{ij}|>2c_{np})-\|\bomega_j^0\|_1\\
&=\|\widehat\bomega_j\|_1-\|\bomega_j^0\|_1\le\|\widetilde\bomega_j\|_1-\|\bomega_j^0\|_1\le 0,
\end{align*}
and hence $\|\bd_j\|_1\le 2\|\bd_j^{(1)}\|_1$. Furthermore,
\begin{align*}
\|\bd_j^{(1)}\|_1&\le\sum_{i=1}^p|\hat\omega_{ij}-\omega_{ij}^0|I(|\hat\omega_{ij}|>2c_{np}) +\sum_{i=1}^p|\omega_{ij}^0|I(|\hat\omega_{ij}|\le2c_{np})\\
&\le\sum_{i=1}^pc_{np}I(|\omega_{ij}^0|>c_{np})+\sum_{i=1}^p|\omega_{ij}^0|I(|\omega_{ij}^0|\le3c_{np})\\
&\le c_{np}^{1-q}s_0(p)+(3c_{np})^{1-q}s_0(p)<4c_{np}^{1-q}s_0(p).
\end{align*}
Combining with the bound \eqref{Sup_align13} for $c_{np}$ yields, for some constant $C_2>0$,
\begin{equation}\label{Sup_align14}
\|\widehat\bOmega-\bOmega_0\|_{L_1}\le C_2s_0(p)\biggl(M_p\sqrt{\frac{\log p}{n}}+\frac{M_p^2}{\sqrt p}\biggr)^{1-q}
\end{equation}
with probability at least $1-O(p^{-\xi})$. Finally, \eqref{Sup_align13}, \eqref{Sup_align14}, and the inequality $\|\bA\|_F^2\le p\|\bA\|_{L_1}\|\bA\|_{\max}$ together imply
\[
\frac{1}{p}\|\widehat\bOmega-\bOmega_0\|_F^2\le C_1C_2s_0(p)\biggl(M_p\sqrt{\frac{\log p}{n}}+\frac{M_p^2}{\sqrt p}\biggr)^{2-q}.
\]
The case under Condition \ref{cond:poly} is similar. This completes the proof of Theorem \ref{thm:omega}.

\subsection{Proof of Corollary \ref{cor:minimax}}
By the assumption $M_p=o(\sqrt{p(\log p)/n})$, the second terms in the bounds of Theorem \ref{thm:omega} are negligible since
\[
\frac{M_p^2/\sqrt p}{M_p\sqrt{(\log p)/n}}=M_p\sqrt{\frac{n}{p\log p}}=o(1).
\]
This proves Corollary \ref{cor:minimax}.

\subsection{Proof of Theorem \ref{thm:supp}}
Note that $\sgn(\hat\omega_{ij}^t)\ne\sgn(\omega_{ij}^0)$ implies either (a) $\omega_{ij}^0=0$ and $\hat\omega_{ij}^t\ne0$, or (b) $\omega_{ij}^0\ne0$ and $\hat\omega_{ij}^t\omega_{ij}^0\le0$. In case (a), by the definition of $\hat\omega_{ij}^t$,
\[
|\hat\omega_{ij}^t-\omega_{ij}^0|=|\hat\omega_{ij}^t|>\tau_{np}.
\]
In case (b), by the assumption $\min_{(i,j):\omega_{ij}^0\ne0}|\omega_{ij}^0|>2\tau_{np}$,
\[
|\hat\omega_{ij}^t-\omega_{ij}^0|>2\tau_{np}-\tau_{np}=\tau_{np}.
\]
Then, by Theorem \ref{thm:omega}, it follows that
\[
P\{\sgn(\hat\omega_{ij}^t)\ne\sgn(\omega_{ij}^0)\text{ for some }1\le i,j\le p\}\le P\Bigl(\max_{1\le i,j\le p}|\hat\omega_{ij}^t-\omega_{ij}^0| >\tau_{np}\Bigr)=O(p^{-\xi}),
\]
which proves Theorem \ref{thm:supp}.

\subsection{Proof of Theorem \ref{thm:cv}}
We first introduce a technical lemma to be used in the proof of Theorem \ref{thm:cv}.

\begin{lemma}\label{lem:inner}
If Condition \ref{cond:subG} holds, then for any fixed $\bU_i\in\mathbb{R}^{p\times p}$ with $\|\bU_i\|_F^2=1$,
\begin{align*}
\max_{1\le i\le N}|\langle\bOmega_0^{1/2}(\widehat\bSigma_c^{(2)}-\bSigma_c)\bOmega_0^{1/2}\bU_i,\bU_i\rangle|&=O_p\biggl(\sqrt{\frac{\log N +\log p}{n_2}}\biggr),\\
\max_{1\le i\le N}|\langle(\widehat\bSigma_c^{(2)}-\bSigma_c)\bOmega_0, \bU_i\rangle|&=O_p\biggl(\sqrt{\frac{\log N +\log p}{n_2}}\biggr),
\end{align*}
where $\langle\cdot,\cdot\rangle$ denotes the Frobenius inner product.
\end{lemma}

\begin{proof}
Let $\bu_{ij}$ denote the $j$th column of $\bU_i$, so that $\sum_{j=1}^p\|\bu_{ij}\|_2^2=1$. Further, let $\widetilde\bu_{ij}=\bu_{ij}/\|\bu_{ij}\|_2$ and $\check\bZ_k=\bOmega_0^{1/2}\bZ_k$. Then we write
\begin{align*}
&\langle\bOmega_0^{1/2}(\widehat\bSigma_c^{(2)}-\bSigma_c)\bOmega_0^{1/2}\bU_i,\bU_i\rangle\\
&\quad=\sum_{j=1}^p\Biggl\{\frac{1}{n_2}\sum_{k=1}^{n_2}(\bu_{ij}^T\check\bZ_k)^2-E(\bu_{ij}^T\check\bZ_k)^2\Biggr\}
+\sum_{j=1}^p\Biggl(\frac{1}{n_2}\sum_{k=1}^{n_2}\bu_{ij}^T\check\bZ_k\Biggr)^2,\\
&\quad=\sum_{j=1}^p\|\bu_{ij}\|_2^2\Biggl\{\frac{1}{n_2}\sum_{k=1}^{n_2}(\widetilde\bu_{ij}^T\check\bZ_k)^2-E(\widetilde\bu_{ij}^T\check\bZ_k)^2\Biggr\} +\sum_{j=1}^p\|\bu_{ij}\|_2^2\Biggl(\frac{1}{n_2}\sum_{k=1}^{n_2}\widetilde\bu_{ij}^T\check\bZ_k\Biggr)^2\\
&\quad\equiv T_1+T_2.
\end{align*}
By the same techniques as in the proof of Lemma \ref{lem:tail}, we can show that $(\widetilde\bu_{i, j}^T\check\bZ_k)^2$ is sub-exponential. Define $H_{ij}=n_2^{-1}\sum_{k=1}^{n_2}(\widetilde\bu_{ij}^T\check\bZ_k)^2-E(\widetilde\bu_{ij}^T\check\bZ_k)^2$. Then, by Bernstein's inequality and the union bound, there exists some constant $C_3>0$ such that
\[
P\Bigl(\max_{1\le i\le N}\max_{1\le j\le p}|H_{ij}|>t_3\Bigr)\le 2Np\exp(-C_3n_2t_3^2)
\]
for all sufficiently small $t_3>0$. For any $\xi>0$, letting $t_3=C_3'\sqrt{(\log N+\log p)/n_2}$ with $C_3'=\sqrt{(\xi+1)/C_3}$, we obtain
\[
P\biggl(\max_{1\le i\le N}\max_{1\le j\le p}|H_{ij}|\le C_3'\sqrt{\frac{\log N+\log p}{n_2}}\biggr)\ge 1-2(Np)^{-\xi},
\]
and hence
\[
T_1=O_p\biggl(\sqrt{\frac{\log N+\log p}{n_2}}\biggr).
\]
Also, we apply Hoeffding's inequality to obtain
\[
\max_{1\le i\le N}\max_{1\le j\le p}\Biggl|\frac{1}{n_2}\sum_{k=1}^{n_2}\widetilde\bu_{ij}^T\check\bZ_k\Biggr| =O_p\biggl(\sqrt{\frac{\log N+\log p}{n_2}}\biggr),
\]
and hence
\[
T_2=O_p\biggl(\frac{\log N+\log p}{n_2}\biggr).
\]
Combining these two pieces yields the first inequality. The second inequality follows essentially from the above arguments and those used for Lemma \ref{lem:conc}.
\end{proof}

\begin{proof}[Proof of Theorem \ref{thm:cv}]
Let $\mu_{j\ell}=\delta_j\ell/N$ for $j=1,\dots,p$ and $\ell=1,\dots,N$. Define the oracle estimator $\widehat\bOmega^o=(\widehat\bomega_1^o, \dots,\widehat\bomega_p^o)$, where
\[
\widehat\bomega_j^o=\widehat\bomega_j^{(1)}(\mu_j^o),\quad\mu_j^o=\argmin_{\mu_{j\ell}}L\{\widehat\bomega_j^{(1)}(\mu_{j\ell}),\bSigma_c\},\quad j=1,\dots,p,
\]
and the optimal data-driven estimator $\widehat\bOmega^d=(\widehat\bomega_1^d,\dots,\widehat\bomega_p^d)$, where
\[
\widehat\bomega_j^d=\widehat\bomega_j^{(1)}(\mu_j^d),\quad\mu_j^d=\argmin_{\mu_{j\ell}}L\{\widehat\bomega_j^{(1)}(\mu_{j\ell}), \widehat\bSigma_c^{(2)}\},\quad j=1,\dots,p.
\]
Using $\bSigma_c=\bG\bSigma_0\bG$, we write
\begin{align*}
L(\bomega_j,\bSigma_c)&=\frac{1}{2}\bomega_j^T\bSigma_c\bomega_j-\be_j^T\bG\bomega_j=\frac{1}{2}(\bG\bomega_j)^T\bSigma_0\bG\bomega_j -\be_j^T\bG\bomega_j\\
&=\frac{1}{2}(\bG\bomega_j-\bomega_j^0+\bomega_j^0)^T\bSigma_0(\bG\bomega_j-\bomega_j^0+\bomega_j^0)-\be_j^T(\bG\bomega_j-\bomega_j^0+\bomega_j^0)\\
&=\frac{1}{2}\|\bSigma_0^{1/2}(\bG\bomega_j-\bomega_j^0)\|_2^2-\frac{1}{2}\be_j^T\bomega_j^0.
\end{align*}
Define $a_{np}=\|\bSigma_0^{1/2}(\bG\widehat\bOmega^o-\bOmega_0)\|_F^2/p$ and $b_{np}=\|\bSigma_0^{1/2}(\bG\widehat\bOmega^d-\bOmega_0)\|_F^2/p$. By Theorem \ref{thm:omega}, for some sufficiently large $\delta_j$, there exists a vector of tuning parameters $\bmu=(\mu_{1\ell_1},\dots,\mu_{p\ell_p})^T$ such that $\mu_{j\ell_j}\asymp\sqrt{(\log p)/n}+M_p/\sqrt p$ for all $j$, and
\[
\frac{1}{p}\|\widehat\bOmega^{(1)}(\bmu)-\bOmega_0\|_F^2=O_p\biggl\{s_0(p)\biggl(M_p^2\sqrt{\frac{\log p}{n}}+\frac{M_p^2}{\sqrt p}\biggr)^{2-q}\biggr\},
\]
where $\widehat\bOmega^{(1)}(\bmu)=(\widehat\bomega_1^{(1)}(\mu_{1\ell_1}),\dots,\widehat\bomega_p^{(1)}(\mu_{p\ell_p}))$. It follows from the assumption $\min_{1\le j\le p}\delta_j\ge NC_0(\sqrt{(\log p)/n}+M_p/\sqrt p)$ that $\|\widehat\bOmega^{(1)}(\bmu)\|_{L_1}\le \|\bOmega_0\|_{L_1}\le M_p$. Then, by the definition of $\widehat\bomega_j^o$ and the fact that $\lambda_{\max}(\bSigma_0)\le R$, we have
\begin{align}\label{Sup_align15}
a_{np}&=\frac{1}{p}\sum_{j=1}^p\|\bSigma_0^{1/2}(\bG\widehat\bomega_j^o-\bomega_j^0)\|_2^2\le \frac{1}{p}\sum_{j=1}^p\|\bSigma_0^{1/2}\{\bG\widehat\bomega_j^{(1)}(\mu_{j\ell_j})-\bomega_j^0\}\|_2^2\notag\\
&\le\frac{R}{p}\sum_{j=1}^p\|\bG\widehat\bomega_j^{(1)}(\mu_{j\ell_j})-\bomega_j^0\|_2^2 =\frac{R}{p}\|\bG\widehat\bOmega^{(1)}(\bmu)-\bOmega_0\|_F^2\notag\\
&\le\frac{2R}{p}\biggl\{\|\widehat\bOmega^{(1)}(\bmu)-\bOmega_0\|_F^2+
\biggl\|\frac{1}{p}\bone_p\bone_p^T\widehat\bOmega^{(1)}(\bmu)\biggr\|_F^2\biggr\}\notag\\
&\le\frac{2R}{p}\{\|\widehat\bOmega^{(1)}(\bmu)-\bOmega_0\|_F^2+\|\widehat\bOmega^{(1)}(\bmu)\|_{L_1}^2\}\notag\\
&=O_p\biggl\{s_0(p)\biggl(M_p\sqrt{\frac{\log p}{n}}+\frac{M_p^2}{\sqrt p}\biggr)^{2-q}+\frac{M_p^2}{p}\biggr\}\notag\\
&=O_p\biggl\{s_0(p)\biggl(M_p\sqrt{\frac{\log p}{n}}+\frac{M_p^2}{\sqrt p}\biggr)^{2-q}\biggr\}.
\end{align}
Clearly, $a_{np}\le b_{np}$ by the definition of $\widehat\bomega_j^o$. We will show that $b_{np}=O_p(a_{np}+M_p^2/p)$ and hence the result follows. Noting that $\widehat\bSigma_c^{(2)}=\bG\widehat\bSigma_c^{(2)}\bG$, $\bSigma_c=\bG\bSigma_0\bG$, and $\bG^2=\bG$, we write
\begin{align}\label{Sup_align16}
L(\bomega_j,\widehat\bSigma_c^{(2)})&=\frac{1}{2}\bomega_j^T\widehat\bSigma_c^{(2)}\bomega_j-\be_j^T\bG\bomega_j\notag\\
&=\frac{1}{2}(\bG\bomega_j-\bomega_j^0+\bomega_j^0)^T\widehat\bSigma_c^{(2)}(\bG\bomega_j -\bomega_j^0+\bomega_j^0)-\be_j^T\bG(\bG\bomega_j-\bomega_j^0+\bomega_j^0)\notag\\
&=\frac{1}{2}(\bG\bomega_j-\bomega_j^0)^T(\widehat\bSigma_c^{(2)}-\bSigma_c)(\bG\bomega_j-\bomega_j^0)
+\frac{1}{2}(\bG\bomega_j-\bomega_j^0)^T\bSigma_0(\bG\bomega_j-\bomega_j^0)\notag\\
&\relph{}+(\bG\bomega_j-\bomega_j^0)^T\bSigma_0(\bomega_j^0-\bG\bomega_j^0)+(\bG\bomega_j-\bomega_j^0)^T(\widehat\bSigma_c^{(2)}\bomega_j^0 -\bG\be_j)\notag\\
&\relph{}+\frac{1}{2}(\bomega_j^0-\bG\bomega_j^0)^T\bSigma_0(\bomega_j^0-\bG\bomega_j^0)+\frac{1}{2}\bomega_j^{0T}\widehat\bSigma_c^{(2)} \bomega_j^0-\be_j^T\bG\bomega_j^0.
\end{align}
Then, by the definition of $\widehat\bomega_j^d$,
\[
\frac{1}{p}\sum_{j=1}^pL(\widehat\bomega_j^d,\widehat\bSigma_c^{(2)})\le \frac{1}{p}\sum_{j=1}^pL(\widehat\bomega_j^o,\widehat\bSigma_c^{(2)}).
\]
Expanding using \eqref{Sup_align16} and rearranging gives
\begin{align*}
&\frac{1}{p}\sum_{j=1}^p(\bG\widehat\bomega_j^d-\bomega_j^0)^T(\widehat\bSigma_c^{(2)}-\bSigma_c)(\bG\widehat\bomega_j^d-\bomega_j^0) +b_{np}\notag\\
&\quad\le\frac{1}{p}\sum_{j=1}^p(\bG\widehat\bomega_j^o-\bomega_j^0)^T(\widehat\bSigma_c^{(2)}-\bSigma_c)(\bG\widehat\bomega_j^o-\bomega_j^0) +a_{np}\notag\\
&\quad\relph{}-\frac{2}{p}\sum_{j=1}^p(\bG\widehat\bomega_j^d-\bomega_j^0)^T\bSigma_0(\bomega_j^0-\bG\bomega_j^0)+\frac{2}{p}\sum_{j=1}^p (\bG\widehat\bomega_j^o-\bomega_j^0)^T\bSigma_0(\bomega_j^0-\bG\bomega_j^0)\notag\\
&\quad\relph{}-\frac{2}{p}\sum_{j=1}^p(\bG\widehat\bomega_j^d -\bomega_j^0)^T(\widehat\bSigma_c^{(2)}\bomega_j^0-\bG\be_j)
+\frac{2}{p}\sum_{j=1}^p(\bG\widehat\bomega_j^o -\bomega_j^0)^T(\widehat\bSigma_c^{(2)}\bomega_j^0-\bG\be_j),
\end{align*}
or
\begin{align}\label{Sup_align17}
&\frac{1}{p}\langle(\widehat\bSigma_c^{(2)}-\bSigma_c)(\bG\widehat\bOmega^d-\bOmega_0),\bG\widehat\bOmega^d-\bOmega_0\rangle+b_{np}\notag\\
&\quad\le\frac{1}{p}\langle(\widehat\bSigma_c^{(2)}-\bSigma_c)(\bG\widehat\bOmega^o-\bOmega_0),\bG\widehat\bOmega^o-\bOmega_0\rangle+a_{np}\notag\\
&\quad\relph{}-\frac{2}{p}\langle\bSigma_0(\bG\widehat\bOmega^d-\bOmega_0),\bOmega_0-\bG\bOmega_0\rangle+\frac{2}{p}\langle\bSigma_0
(\bG\widehat\bOmega^o-\bOmega_0),\bOmega_0-\bG\bOmega_0\rangle\notag\\
&\quad\relph{}-\frac{2}{p}\langle\widehat\bSigma_c^{(2)}\bOmega_0-\bG,\bG\widehat\bOmega^d-\bOmega_0\rangle+\frac{2}{p}\langle \widehat\bSigma_c^{(2)}\bOmega_0-\bG,\bG\widehat\bOmega^o-\bOmega_0\rangle.
\end{align}
Now, we write
\begin{align*}
&\frac{1}{p}\langle(\widehat\bSigma_c^{(2)}-\bSigma_c)(\bG\widehat\bOmega^d-\bOmega_0),\bG\widehat\bOmega^d-\bOmega_0\rangle\\
&\quad=\frac{1}{p}\langle(\widehat\bSigma_c^{(2)}-\bSigma_c)\bOmega_0^{1/2}\bSigma_0^{1/2}(\bG\widehat\bOmega^d-\bOmega_0),\bOmega_0^{1/2} \bSigma_0^{1/2}(\bG\widehat\bOmega^d-\bOmega_0)\rangle\\
&\quad=\langle\bOmega_0^{1/2}(\widehat\bSigma_c^{(2)}-\bSigma_c)\bOmega_0^{1/2}\bU_d,\bU_d\rangle b_{np},
\end{align*}
where $\bU_d=\bSigma_0^{1/2}(\bG\widehat\bOmega^d-\bOmega_0)/\|\bSigma_0^{1/2}(\bG\widehat\bOmega^d-\bOmega_0)\|_F$. Thus, by Lemma \ref{lem:inner} and the assumption $\log N=o(n)$, we have
\[
\frac{1}{p}\langle(\widehat\bSigma_c^{(2)}-\bSigma_c)(\bG\widehat\bOmega^d-\bOmega_0),\bG\widehat\bOmega^d-\bOmega_0\rangle=o_p(b_{np}).
\]
Similarly,
\[
\frac{1}{p}\langle(\widehat\bSigma_c^{(2)}-\bSigma_c)(\bG\widehat\bOmega^o-\bOmega_0),\bG\widehat\bOmega^o-\bOmega_0\rangle=o_p(a_{np}).
\]
By the Cauchy--Schwarz inequality, we have
\[
\frac{1}{p}|\langle\bSigma_0(\bG\widehat\bOmega^d-\bOmega_0),\bOmega_0-\bG\bOmega_0\rangle|=\frac{1}{p}|\langle\bSigma_0^{1/2}
(\bG\widehat\bOmega^d-\bOmega_0),\bSigma_0^{1/2}(\bOmega_0-\bG\bOmega_0)\rangle|\le\frac{\sqrt RM_p}{\sqrt p}b_{np}^{1/2}.
\]
Thus,
\[
\frac{1}{p}\langle\bSigma_0(\bG\widehat\bOmega^d-\bOmega_0),\bOmega_0-\bG\bOmega_0\rangle=O\biggl(\frac{M_p}{\sqrt p}b_{np}^{1/2}\biggr).
\]
Similarly,
\[
\frac{1}{p}\langle\bSigma_0(\bG\widehat\bOmega^o-\bOmega_0),\bOmega_0-\bG\bOmega_0\rangle=O\biggl(\frac{M_p}{\sqrt p}a_{np}^{1/2}\biggr).
\]
To bound the last two terms in \eqref{Sup_align17}, note that
\begin{align*}
&\frac{1}{p}|\langle\widehat\bSigma_c^{(2)}\bOmega_0-\bG, \bG\widehat\bOmega^d-\bOmega_0\rangle|\\
&\quad\le\frac{1}{p}|\langle(\widehat\bSigma_c^{(2)}-\bSigma_c)\bOmega_0,\bG\widehat\bOmega^d-\bOmega_0\rangle|
+\frac{1}{p}|\langle\bSigma_c(\bOmega_0-\bOmega_c), \bG\widehat\bOmega^d-\bOmega_0\rangle|\\
&\quad\equiv T_1+T_2.
\end{align*}
For the term $T_1$, we have
\begin{align*}
T_1&=\frac{1}{p}|\langle(\widehat\bSigma_c^{(2)}-\bSigma_c)\bOmega_0,\bV_d\rangle|\|\bG\widehat\bOmega^d-\bOmega_0\|_F\\
&\le\frac{\sqrt{R}}{p}|\langle(\widehat\bSigma_c^{(2)}-\bSigma_c)\bOmega_0,\bV_d\rangle|\|\bSigma_0^{1/2}(\bG\widehat\bOmega^d-\bOmega_0)\|_F\\
&=\sqrt{\frac{R}{p}}|\langle(\widehat\bSigma_c^{(2)}-\bSigma_c)\bOmega_0,\bV_d\rangle|b_{np}^{1/2},
\end{align*}
where $\bV_d=(\bG\widehat\bOmega^d-\bOmega_0)/\|\bG\widehat\bOmega^d-\bOmega_0\|_F$. Applying Lemma \ref{lem:inner}, we obtain
\[
T_1=o_p(b_{np}^{1/2}p^{-1/2}).
\]
Also, it follows from $R^{-1}\le \lambda_{\min}(\bSigma_0)\le\lambda_{\max}(\bSigma_0)\le R$ and $M_p=o(\sqrt p)$ that
\begin{align*}
T_2&\le\frac{1}{p}\|\bSigma_c(\bOmega_0-\bOmega_c)\|_F\|\bG\widehat\bOmega^d-\bOmega_0\|_F\\
&\le\frac{1}{p}\biggl(\biggl\|\frac{1}{p}\bSigma_0\bone_p\bone_p^T\bOmega_0\biggr\|_{F}+\biggl\|\frac{1}{p^2}(\bone_p^T\bSigma_0\bone_p) \bone_p\bone_p^T\bOmega_0\biggr\|_{F}\biggr)\|\bG\widehat\bOmega^d-\bOmega_0\|_F\\
&=O\biggl(\frac{M_pb_{np}^{1/2}}{p}\biggr)=o(b_{np}^{1/2}p^{-1/2}).
\end{align*}
Combining the bounds for $T_1$ and $T_2$ yields
\[
\frac{1}{p}\langle\widehat\bSigma_c^{(2)}\bOmega_0-\bG, \bG\widehat\bOmega^d-\bOmega_0\rangle=o_p(b_{np}^{1/2}p^{-1/2}),
\]
and similarly,
\[
\frac{1}{p}\langle\widehat\bSigma_c^{(2)}\bOmega_0-\bG, \bG\widehat\bOmega^o-\bOmega_0\rangle=o_p(a_{np}^{1/2}p^{-1/2}).
\]
Putting these pieces together, \eqref{Sup_align17} becomes
\[
b_{np}\{1+o_{p}(1)\}\le a_{np}\{1+o_{p}(1)\}+O_p\biggl(\frac{M_p}{\sqrt p}b_{np}^{1/2}+\frac{M_p}{\sqrt p}a_{np}^{1/2}\biggr).
\]
Dividing both sides by $b_{np}^{1/2}$ and using the fact that $a_{np}\le b_{np}$, we find that $b_{np}=O_p(a_{np}+M_p^2/p)$. Also, the assumption $\min_{1\le j\le p}\delta_j\ge NC_0(\sqrt{\log p/n}+M_p/\sqrt p)$ implies that $\|\widehat\bOmega^{d}\|_{L_1}\le M_p$. Therefore, it follows from \eqref{Sup_align15} and $\lambda_{\min}(\bSigma_0)\ge R^{-1}$ that
\begin{align*}
\frac{1}{p}\|\widehat\bOmega^d-\bOmega_0\|_F^2&\le\frac{2}{p}\|\bG\widehat\bOmega^d-\bOmega_0\|_F^2+\frac{2}{p}\biggl\|\frac{1}{p}\bone_p\bone_p^T \widehat\bOmega^d\biggr\|_F^2\\
&\le 2Rb_{np}+\frac{2}{p}\|\widehat\bOmega^d\|_{L_1}^2\\
&=O_p\biggl\{s_0(p)\biggl(M_p\sqrt{\frac{\log p}{n}}+\frac{M_p^2}{\sqrt p}\biggr)^{2-q}\biggr\},
\end{align*}
which completes the proof of Theorem \ref{thm:cv}.
\end{proof}

\subsection{Proof of Theorem \ref{thm:perturb}}
Note that
\begin{align*}
\|\check\bSigma_c\bOmega_0-\bG\|_{\max}&\le\|\widehat\bSigma_c\bOmega_0-\bG\|_{\max}+\|(\check\bSigma_c-\widehat\bSigma_c)\bOmega_0\|_{\max}\\
&\le\lambda_0+\|\check\bSigma_c-\widehat\bSigma_c\|_{\max}\|\bOmega_0\|_{L_1}\le\lambda_0+M_p\|\bE\|_{\max}.
\end{align*}
If we choose $\lambda\ge\lambda_0+M_p\|\bE\|_{\max}$, then $\bOmega_0$ belongs to the feasible sets of both problem \eqref{eq:opt} and its perturbed version, and hence $\|\widehat\bOmega\|_{L_1}\le\|\bOmega_0\|_{L_1}\le M_p$ and $\|\check\bOmega\|_{L_1}\le\|\bOmega_0\|_{L_1}\le M_p$. From the decomposition
\begin{align*}
\check\bOmega-\widehat\bOmega&=\bG\check\bOmega+\frac{1}{p}\bone_p\bone_p^T\check\bOmega-\widehat\bOmega\bG-\frac{1}{p}\widehat\bOmega\bone_p \bone_p^T\\
&=(\bG-\widehat\bOmega\widehat\bSigma_c)\check\bOmega+\widehat\bOmega(\widehat\bSigma_c-\check\bSigma_c)\check\bOmega+\widehat\bOmega (\check\bSigma_c\check\bOmega-\bG)+\frac{1}{p}\bone_p\bone_p^T\check\bOmega-\frac{1}{p}\widehat\bOmega\bone_p\bone_p^T,
\end{align*}
we have
\begin{align*}
\|\check\bOmega-\widehat\bOmega\|_{\max}&\le\|\bG-\widehat\bOmega\widehat\bSigma_c\|_{\max}\|\check\bOmega\|_{L_1}+\|\widehat\bOmega\|_{L_1} \|\widehat\bSigma_c-\check\bSigma_c\|_{\max}\|\check\bOmega\|_{L_1}\\
&\relph{}+\|\widehat\bOmega\|_{L_1}\|\check\bSigma_c\check\bOmega-\bG\|_{\max}+\frac{1}{p}\|\check\bOmega\|_{L_1} +\frac{1}{p}\|\widehat\bOmega\|_{L_1}\\
&\le 2\lambda M_p+M_p^2\|\bE\|_{\max}+2M_p/p.
\end{align*}
Choosing $\lambda=\lambda_0+M_p\|\bE\|_{\max}$ yields the desired bound.

\section{An Alternative Method}\label{sec:aclime}
In an earlier version of our article, we considered an alternative estimation strategy that follows the ACLIME approach more closely, which we briefly discuss here. That strategy consists of a
two-stage adaptive constrained $\ell_1$-minimization procedure for estimating $\bOmega_c$ and an additional thresholding step for obtaining a sparse estimator of $\bOmega_0$. More specifically, the first-stage estimator $\widetilde\bOmega_c= (\widetilde\bomega_1,\dots,\widetilde\bomega_p)$ is defined through the solutions $\widetilde\bomega_j$ to the optimization problems
\[
\text{minimize }\|\bomega_j\|_1\quad\text{subject to}\quad\left\|\widehat\bSigma_c\bomega_j-\left(\be_j-\frac{\bone_p}{p}\right)\right\|_{\infty}\le\lambda_1\left(\max_{1\le i\le p}\hat\sigma_{ii}^c\right)\omega_{jj}\sqrt{\frac{\log p}{n}},
\]
where $\omega_{jj}$ is the $j$th component of $\bomega_j$, and $\lambda_1>0$ is a tuning parameter. After obtaining the estimates $\tilde\omega_{jj}^c$ of $\omega_{jj}^c$, we define the estimator $\widetilde\bOmega_c^*=(\widetilde\bomega_1^*,\dots,\widetilde\bomega_p^*)$ with $\widetilde\bomega_j^*$ being the solutions to the optimization problems
\begin{align*}
&\text{minimize }\|\bomega_j\|_1\quad\text{subject to}\notag\\
&\left|\left\{\hat\bSigma_c\bomega_j-\left(\be_j-\frac{\bone_p}{p}\right)\right\}_i\right|\le\lambda_2\sqrt{\frac{\hat\sigma_{ii}^c \tilde\omega_{jj}^c\log p}{n}},\quad i=1,\dots,p,
\end{align*}
where $\lambda_2>0$ is a tuning parameter. The second-stage estimator $\widehat\bOmega_c=(\hat\omega_{ij}^c)$ of $\bOmega_c$ is then formed by symmetrizing $\widetilde\bOmega_c^*$:
\[
\hat\omega_{ij}^c=\hat\omega_{ji}^c=\tilde\omega_{ij}^*I(|\tilde\omega_{ij}^*|\le|\tilde\omega_{ji}^*|)+\tilde\omega_{ji}^* I(|\tilde\omega_{ij}^*|>|\tilde\omega_{ji}^*|).
\]
Finally, we define our estimator as the hard-thresholded estimator $\widehat\bOmega^*=(\hat\omega_{ij}^*)$ with
\[
\hat\omega_{ij}^*=\hat\omega_{ij}^cI(|\hat\omega_{ij}^c|>2\tau_{np}),
\]
where $\tau_{np}=C_0(M_p\sqrt{(\log p)/n}+M_p^2/p)$ is the ideal threshold level. The most notable difference of this three-step method from our single-step method presented in Section \ref{sec:method} is that a scaling factor of $\sqrt{(\log p)/n}$ is still used in the adaptive estimation procedure and the identification error is not taken into account until the thresholding step.

At the cost of a required thresholding step and the more stringent assumption that $M_p=o(\sqrt{n/\log p}\wedge\sqrt p)$, we can show that the three-step method achieves a faster rate of convergence ($M_p^2/p$) for the identification error than the single-step method ($M_p^2/\sqrt{p}$). On the other hand, the single-step method has the following noteworthy advantages: (1) it is methodologically much simpler and incurs less computational burden; (2) the tuning parameters involved in the procedure are easier to tune via cross-validation; (3) stronger theoretical guarantees can be provided for the data-driven version; (4) it has better practical performance on simulated and real data. Therefore, we regard the single-step method as superior to the three-step method and have chosen to present it in the main article.

\section{Additional Numerical Results}

\subsection{Timing Results}
We compare the computational efficiency of all methods based on model (a) in our simulation settings. The timing results for computing the entire solution path and for tuning parameter selection via fivefold cross-validation with varying dimensionality are shown in Figure \ref{fig:runtime}. As seen from Figure \ref{fig:runtime}(a), our method is the fastest among the four competitors. Owing to high computational burdens, CD-trace and gCoda use BIC for tuning parameter selection, and are not included in the comparisons in Figure \ref{fig:runtime}(b). We observe that the computational cost of our method grows much more slowly than that of SPIEC-EASI and is only about the same as that of gCoda even with the computational overhead of cross-validation.

\begin{figure}
\begin{subfigure}{.5\textwidth}
\centering
\includegraphics[width=\textwidth]{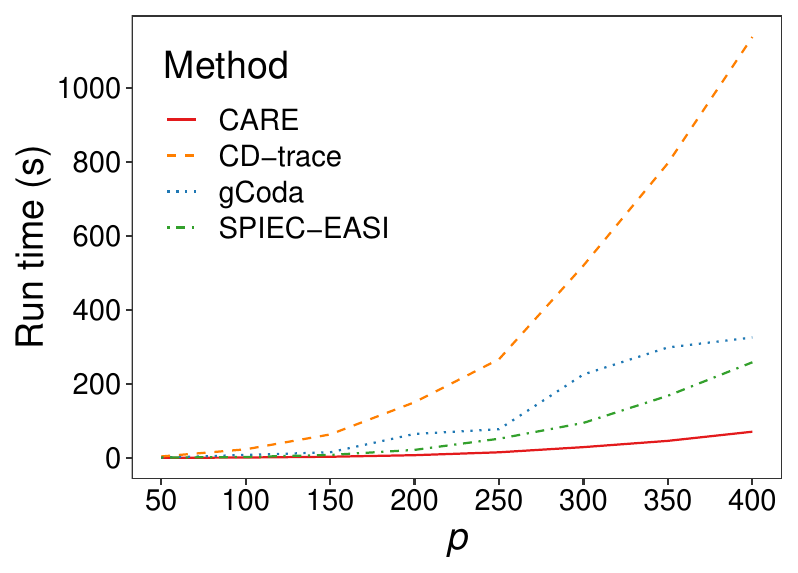}
\caption{}
\end{subfigure}%
\begin{subfigure}{.5\textwidth}
\centering
\includegraphics[width=\textwidth]{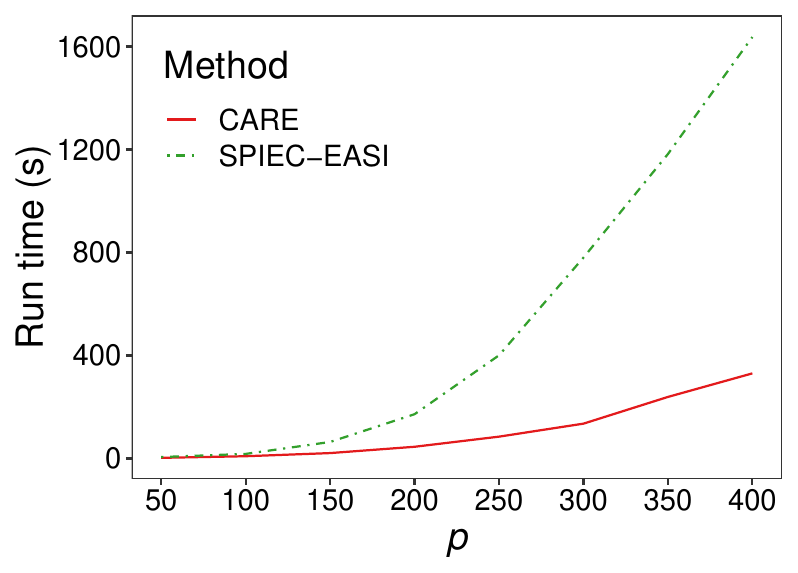}
\caption{}
\end{subfigure}
\caption{Timing results for (a) computing the entire solution path and (b) fivefold cross-validation using different methods with varying dimensionality.}\label{fig:runtime}
\end{figure}

\subsection{Simulation Results Based on Settings of \citet{yuan2019}}
We closely follow the simulation settings of \citet{yuan2019} and reexamine the performance of all methods. The main differences from our settings in the main article are as follows: (1) they fixed $p=50$ and varied the sample size from $n=200$ to 500, while we fixed $n=200$ and varied the dimensionality from $p=50$ to 400; (2) they fixed the number of edges at $e=150$, while the number of edges in our settings is determined by the network models and varies with $p$; (3) they included the scale-free graph, which is not covered in our settings. The simulation results under these settings are reported in Tables \ref{tab:band}--\ref{tab:scale_free}, and the ROC curves are shown in Figure \ref{fig:roc_yuan}. We observe similar trends to our previous results, suggesting that our method consistently outperforms the competitors.

\begin{table}
\caption{Means and standard errors (in parentheses) of performance measures for different methods in the band graph over 100 replications.}\label{tab:band}
\def~{\phantom{0}}
\begin{tabular*}{\textwidth}{@{}c*{5}{@{\extracolsep{\fill}}c}@{}}
\hline
& \multicolumn{5}{c}{Method}\\
\cline{2-6}
$n$ & CARE & Oracle & CD-trace & gCoda & SPIEC-EASI\\
\hline
\multicolumn{6}{c}{Spectral norm loss}\\
200 & 13.55 (0.54) & 13.34 (0.56) & 16.48 (0.17) & 16.35 (0.37) & 23.27 (0.02)\\
300 & 11.72 (0.56) & 11.31 (0.68) & 16.13 (0.16) & 16.30 (0.28) & 23.23 (0.04)\\
500 & ~9.50 (0.74) & ~8.68 (0.72) & 15.67 (0.15) & 16.31 (0.23) & 23.18 (0.04)\\
\multicolumn{6}{c}{Matrix $\ell_1$-norm loss}\\
200 & 21.09 (1.16) & 20.46 (1.12) & 24.55 (0.44) & 23.56 (0.73) & 31.29 (0.07)\\
300 & 18.59 (1.11) & 17.64 (1.44) & 23.65 (0.43) & 23.26 (0.61) & 31.14 (0.12)\\
500 & 15.47 (1.41) & 13.68 (1.21) & 22.50 (0.37) & 23.11 (0.48) & 30.99 (0.15)\\
\multicolumn{6}{c}{Frobenius norm loss}\\
200 & 40.27 (1.09) & 39.76 (1.26) & 56.56 (0.49) & 55.37 (1.62) & 94.24 (0.15)\\
300 & 33.29 (1.06) & 32.32 (1.18) & 54.78 (0.41) & 55.06 (1.36) & 93.78 (0.36)\\
500 & 25.34 (1.01) & 24.05 (1.00) & 52.26 (0.41) & 55.12 (1.03) & 93.35 (0.36)\\
\multicolumn{6}{c}{True positive rate (\%)}\\
200 &  71.1 (3.4)  &  71.2 (3.6)  &  42.8 (2.2)  &  57.4 (4.9)  &  34.2 (3.2)\\
300 &  86.9 (2.8)  &  87.6 (2.7)  &  53.4 (2.4)  &  59.0 (4.3)  &  42.9 (6.9)\\
500 &  97.1 (1.2)  &  97.8 (1.1)  &  68.0 (2.2)  &  59.3 (3.1)  &  49.8 (5.3)\\
\multicolumn{6}{c}{False positive rate (\%)}\\
200 &  ~3.0 (0.6)  &  ~2.3 (0.5)  &  ~5.1 (0.6)  &  ~7.3 (0.8)  &  ~5.7 (0.6)\\
300 &  ~4.8 (0.7)  &  ~3.8 (0.7)  &  ~5.5 (0.6)  &  ~7.2 (0.7)  &  ~6.7 (1.0)\\
500 &  ~7.4 (0.9)  &  ~5.8 (0.9)  &  ~6.1 (0.5)  &  ~7.0 (0.6)  &  ~7.6 (0.8)\\
\hline
\end{tabular*}
\end{table}

\begin{table}
\caption{Means and standard errors (in parentheses) of performance measures for different methods in the block graph over 100 replications.}\label{tab:block}
\def~{\phantom{0}}
\begin{tabular*}{\textwidth}{@{}c*{5}{@{\extracolsep{\fill}}c}@{}}
\hline
& \multicolumn{5}{c}{Method}\\
\cline{2-6}
$n$ & CARE & Oracle & CD-trace & gCoda & SPIEC-EASI\\
\hline
\multicolumn{6}{c}{Spectral norm loss}\\
200 & 14.02 (0.63) & 14.25 (0.72) & 16.34 (0.24) & 17.35 (0.30) & 22.98 (0.01)\\
300 & 12.37 (0.70) & 12.32 (0.78) & 15.86 (0.20) & 17.38 (0.26) & 22.98 (0.01)\\
500 & ~9.73 (0.81) & ~9.23 (0.84) & 15.38 (0.17) & 17.37 (0.23) & 22.96 (0.03)\\
\multicolumn{6}{c}{Matrix $\ell_1$-norm loss}\\
200 & 22.69 (1.32) & 22.65 (1.21) & 27.64 (0.53) & 28.83 (0.71) & 35.33 (0.06)\\
300 & 20.11 (1.23) & 19.56 (1.41) & 26.82 (0.52) & 28.92 (0.58) & 35.34 (0.06)\\
500 & 16.03 (1.22) & 15.12 (1.33) & 25.62 (0.48) & 29.02 (0.47) & 35.27 (0.13)\\
\multicolumn{6}{c}{Frobenius norm loss}\\
200 & 41.61 (1.07) & 41.32 (1.17) & 55.87 (0.50) & 61.78 (1.50) & 96.21 (0.06)\\
300 & 35.08 (1.08) & 34.10 (1.33) & 54.49 (0.46) & 62.03 (1.23) & 96.20 (0.07)\\
500 & 26.65 (1.00) & 25.34 (0.93) & 52.51 (0.40) & 62.10 (0.99) & 95.99 (0.33)\\
\multicolumn{6}{c}{True positive rate (\%)}\\
200 &  61.3 (3.7)  &  61.5 (3.6)  &  35.7 (2.5)  &  32.8 (4.3)  &  37.3 (2.6)\\
300 &  79.7 (3.0)  &  81.9 (3.9)  &  45.9 (2.4)  &  31.9 (3.4)  &  36.6 (2.6)\\
500 &  95.3 (1.6)  &  96.7 (1.4)  &  59.3 (2.1)  &  31.2 (2.9)  &  40.4 (7.5)\\
\multicolumn{6}{c}{False positive rate (\%)}\\
200 &  ~1.9 (0.5)  &  ~1.6 (0.5)  &  ~0.8 (0.3)  &  ~1.8 (0.4)  &  ~0.4 (0.2)\\
300 &  ~3.4 (0.6)  &  ~3.0 (0.8)  &  ~0.9 (0.3)  &  ~1.4 (0.3)  &  ~0.3 (0.2)\\
500 &  ~6.3 (0.9)  &  ~5.4 (0.9)  &  ~1.1 (0.3)  &  ~0.9 (0.3)  &  ~0.4 (0.5)\\
\hline
\end{tabular*}
\end{table}

\begin{table}
\caption{Means and standard errors (in parentheses) of performance measures for different methods in the scale-free graph over 100 replications.}\label{tab:scale_free}
\def~{\phantom{0}}
\begin{tabular*}{\textwidth}{@{}c*{5}{@{\extracolsep{\fill}}c}@{}}
\hline
& \multicolumn{5}{c}{Method}\\
\cline{2-6}
$n$ & CARE & Oracle & CD-trace & gCoda & SPIEC-EASI\\
\hline
\multicolumn{6}{c}{Spectral norm loss}\\
200 & 15.62 (1.08) & 15.23 (1.14) & 19.12 (0.22) & 20.16 (0.50) & ~27.84 (0.03)\\
300 & 12.72 (1.07) & 12.13 (1.00) & 18.63 (0.21) & 20.16 (0.42) & ~27.75 (0.09)\\
500 & ~9.77 (0.90) & ~8.98 (0.84) & 18.06 (0.19) & 20.09 (0.31) & ~27.59 (0.10)\\
\multicolumn{6}{c}{Matrix $\ell_1$-norm loss}\\
200 & 38.19 (5.13) & 37.29 (5.16) & 42.34 (0.77) & 44.41 (1.36) & ~63.39 (0.08)\\
300 & 29.87 (4.43) & 28.76 (4.31) & 40.71 (0.70) & 44.32 (1.19) & ~63.19 (0.20)\\
500 & 22.79 (3.63) & 20.84 (3.22) & 38.76 (0.63) & 44.03 (0.97) & ~62.83 (0.25)\\
\multicolumn{6}{c}{Frobenius norm loss}\\
200 & 40.59 (1.40) & 39.99 (1.46) & 57.32 (0.58) & 62.95 (2.18) & 100.57 (0.16)\\
300 & 33.14 (1.29) & 31.99 (1.34) & 55.29 (0.47) & 63.19 (1.84) & 100.11 (0.45)\\
500 & 25.11 (1.14) & 23.81 (1.11) & 52.74 (0.47) & 63.03 (1.39) & ~99.33 (0.48)\\
\multicolumn{6}{c}{True positive rate (\%)}\\
200 &  73.2 (4.2)  &  74.2 (4.5)  &  56.0 (2.1)  &  55.4 (3.9)  &  ~50.2 (3.2)\\
300 &  89.6 (3.2)  &  91.2 (2.7)  &  64.9 (1.9)  &  55.1 (3.5)  &  ~54.7 (4.4)\\
500 &  98.6 (1.1)  &  99.1 (0.8)  &  75.2 (1.7)  &  55.6 (3.1)  &  ~60.9 (2.9)\\
\multicolumn{6}{c}{False positive rate (\%)}\\
200 &  ~3.0 (0.6)  &  ~2.7 (0.5)  &  ~8.0 (0.1)  &  ~8.4 (0.8)  &  ~29.9 (3.2)\\
300 &  ~4.4 (0.7)  &  ~4.0 (0.7)  &  ~9.0 (0.7)  &  ~8.1 (0.6)  &  ~33.1 (2.7)\\
500 &  ~6.4 (0.8)  &  ~5.6 (0.8)  &  10.1 (0.7)  &  ~7.9 (0.6)  &  ~36.6 (1.4)\\
\hline
\end{tabular*}
\end{table}

\begin{figure}
\centering
\includegraphics[width=.33\textwidth]{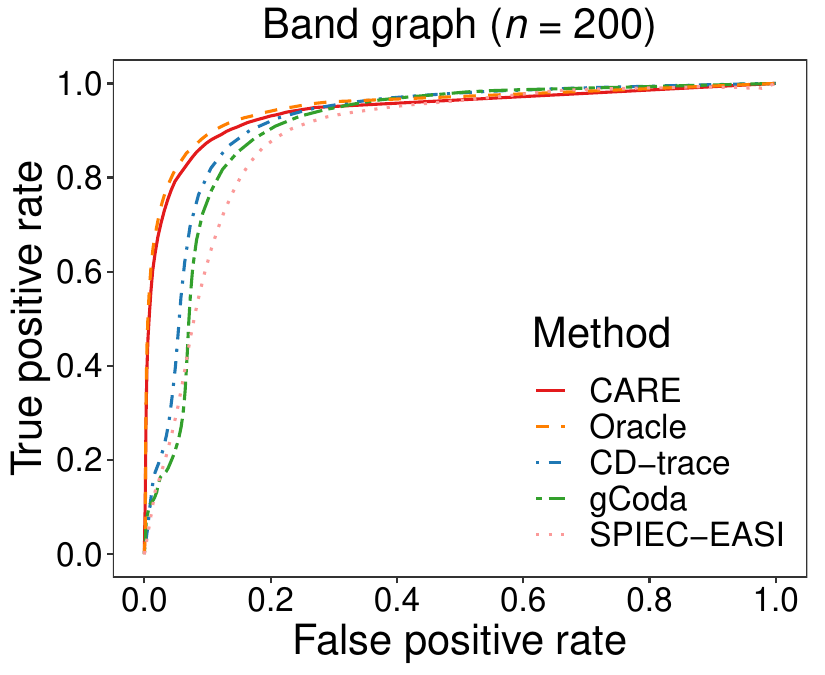}%
\includegraphics[width=.33\textwidth]{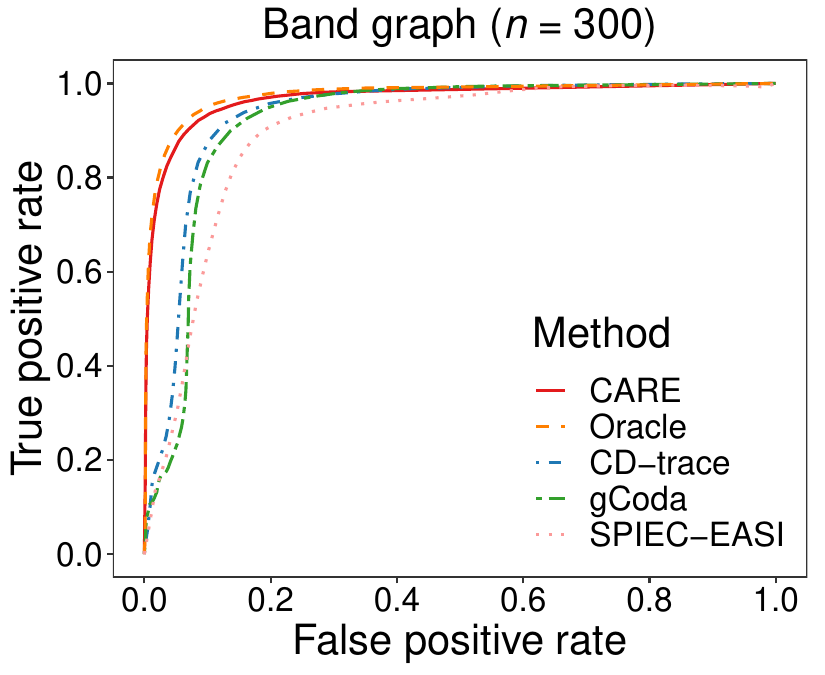}%
\includegraphics[width=.33\textwidth]{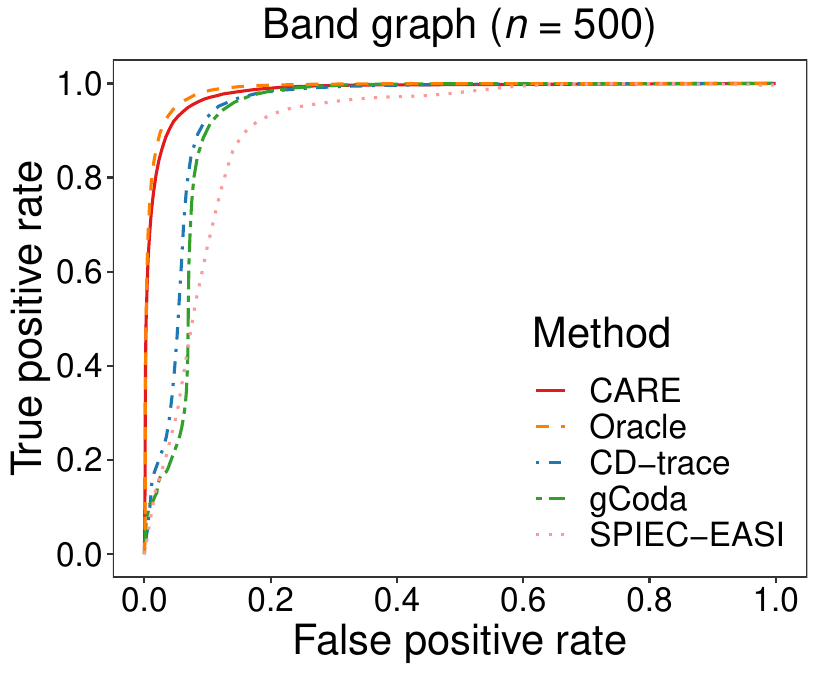}\\
\includegraphics[width=.33\textwidth]{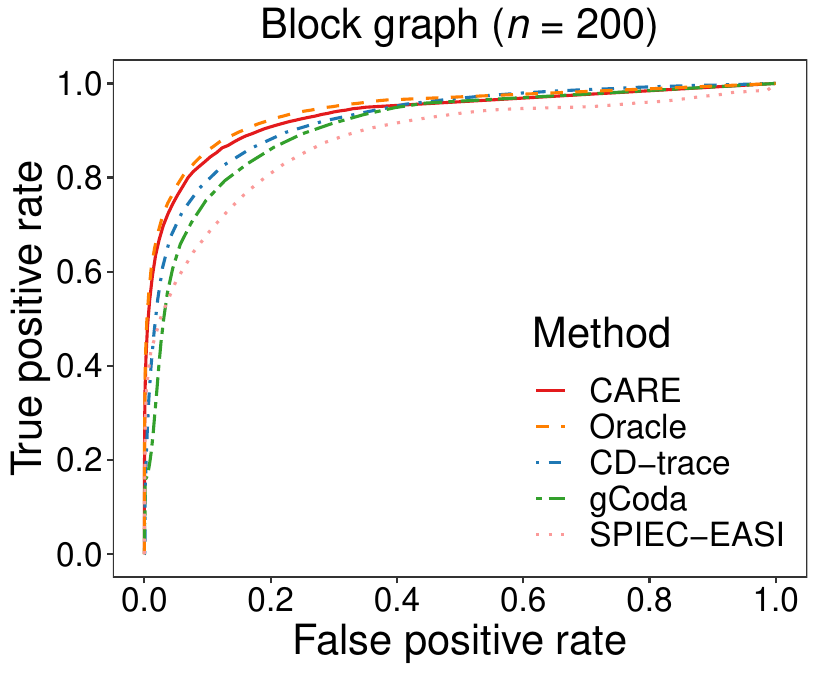}%
\includegraphics[width=.33\textwidth]{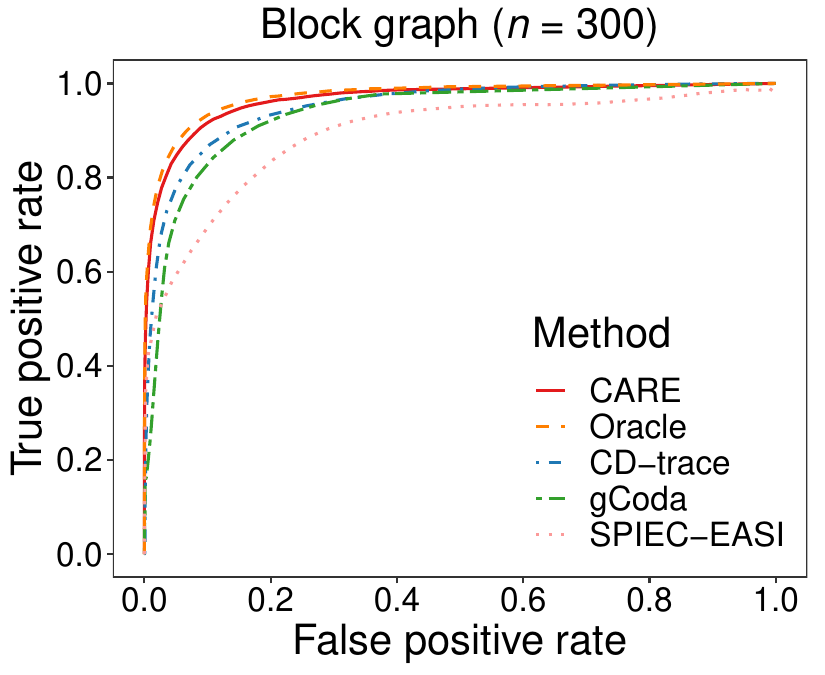}%
\includegraphics[width=.33\textwidth]{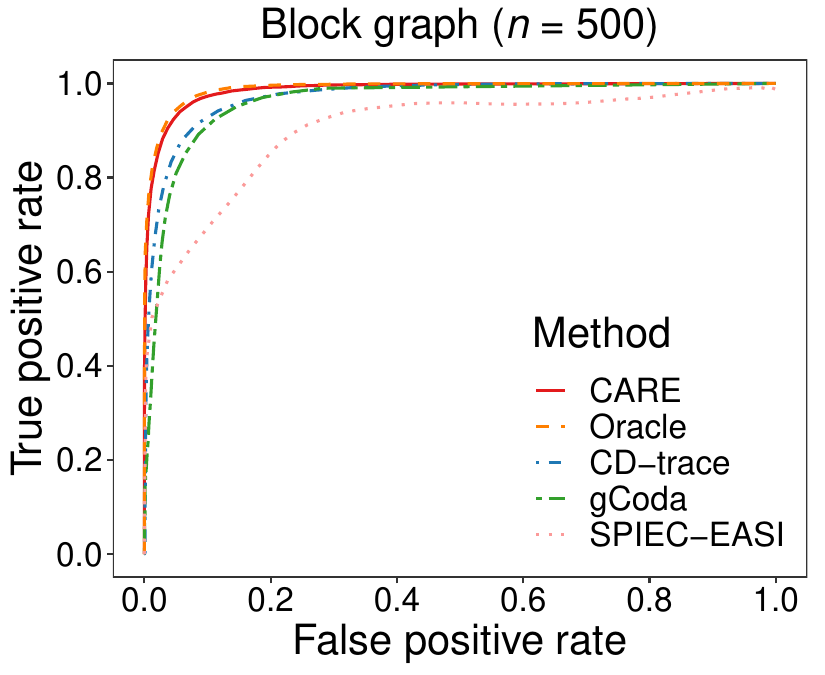}\\
\includegraphics[width=.33\textwidth]{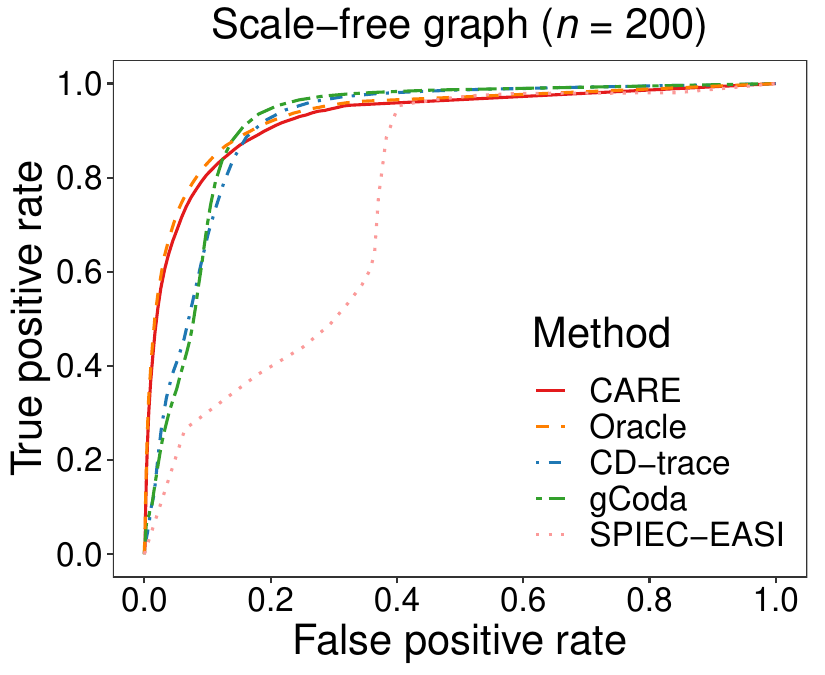}%
\includegraphics[width=.33\textwidth]{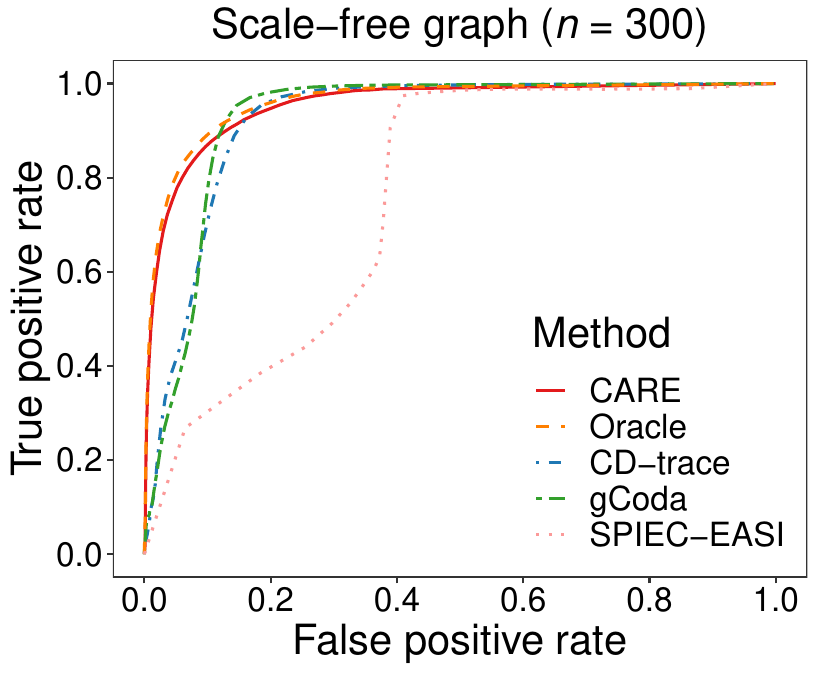}%
\includegraphics[width=.33\textwidth]{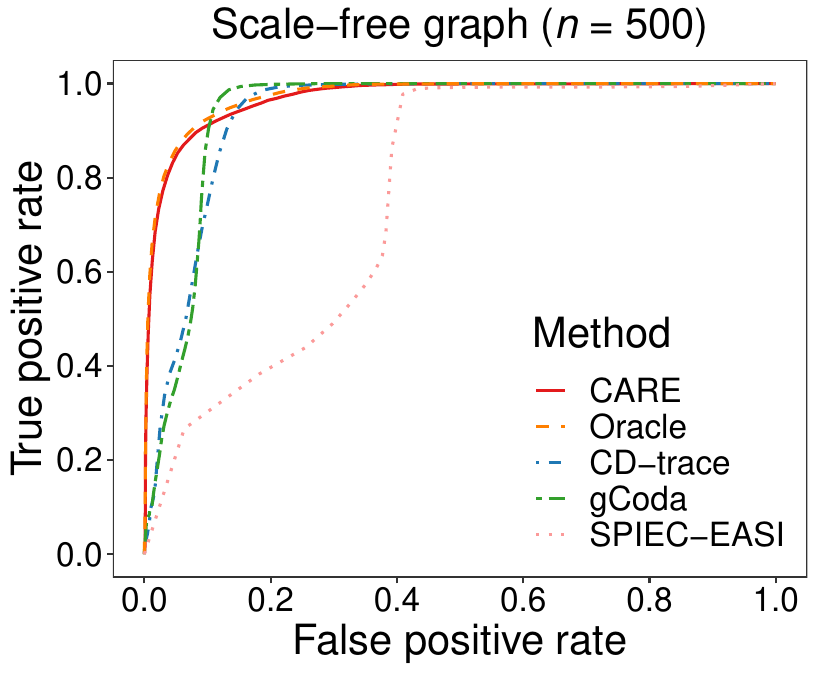}
\caption{The ROC curves for different methods in the band, block, and scale-free graphs with $n=200,300,500$.}\label{fig:roc_yuan}
\end{figure}

\subsection{More Discussion on Microbiome Data}\label{ssec:genera}
A further comparison of the lean and obese networks identifies the interactions between \emph{Barnesiella} and \emph{Butyricimonas} and between \emph{Dialister} and \emph{Phascolarctobacterium} as the shared links between the two groups. Among these genera, \emph{Barnesiella} and \emph{Butyricimonas} have been found to drive the efficacy of prebiotic intervention in obese patients \citep{Rodriguez.etal2020}. Also, the genera \emph{Dialister} and \emph{Phascolarctobacterium} have been associated with insulin sensitivity and metabolic benefits in obesity and metabolic syndrome \citep{mocanu2021}. Moreover, we observe that the interaction between \emph{Prevotella} and \emph{Paraprevotella} has the highest strength in the obese networks constructed by CARE, gCoda, and SPIEC-EASI. The genus \emph{Prevotella} is known to characterize one of three human enterotypes \citep{arumugam2011} and has been suggested to play a role in dietary fiber-induced improvement in glucose metabolism in mice and humans \citep{kovatcheva2015}. The genus \emph{Paraprevotella} is a member of the family \emph{Prevotellaceae} closely related to \emph{Prevotella} and has been implicated in type 2 diabetes \citep{forslund2015}. More microbial interactions in the identified networks are yet to be elucidated and could be exploited to develop therapeutic strategies for obesity.

\subsection{Supplementary Tables and Figures}
The simulation results for models (c) and (d) over 100 replications are reported in Tables \ref{tab:model_c} and \ref{tab:model_d}, respectively. The microbial interaction networks identified by the CD-trace, gCoda, and SPIEC-EASI methods are shown in Figures \ref{fig:CD_net}--\ref{fig:SPIEC_net}.

\begin{table}
\caption{Means and standard errors (in parentheses) of performance measures for different methods in model (c) over 100 replications.}\label{tab:model_c}
\def~{\phantom{0}}
\begin{tabular*}{\textwidth}{@{}c*{5}{@{\extracolsep{\fill}}c}@{}}
\hline
& \multicolumn{5}{c}{Method}\\
\cline{2-6}
$p$ & CARE & Oracle & CD-trace & gCoda & SPIEC-EASI\\
\hline
\multicolumn{6}{c}{Spectral norm loss}\\
~50 & ~3.35 (0.20) & ~3.00 (0.24) & ~3.81 (0.08) & ~4.54 (0.06) & ~4.34 (0.03)\\
100 & ~2.86 (0.16) & ~2.72 (0.16) & ~3.35 (0.08) & ~4.06 (0.08) & ~4.05 (0.03)\\
200 & ~3.80 (0.16) & ~3.71 (0.16) & ~4.28 (0.08) & ~4.94 (0.10) & ~5.14 (0.02)\\
400 & ~3.69 (0.12) & ~3.64 (0.12) & ~4.39 (0.07) & ~4.98 (0.07) & ~5.33 (0.02)\\
\multicolumn{6}{c}{Matrix $\ell_1$-norm loss}\\
~50 & ~4.48 (0.33) & ~4.04 (0.31) & ~4.94 (0.13) & ~5.80 (0.07) & ~5.41 (0.08)\\
100 & ~4.39 (0.36) & ~4.15 (0.34) & ~5.31 (0.20) & ~5.84 (0.14) & ~5.75 (0.10)\\
200 & ~7.55 (0.54) & ~7.26 (0.59) & ~7.82 (0.33) & ~9.29 (0.23) & ~8.92 (0.07)\\
400 & ~7.21 (0.66) & ~7.05 (0.62) & ~7.53 (0.28) & 11.70 (0.31) & ~9.88 (0.08)\\
\multicolumn{6}{c}{Frobenius norm loss}\\
~50 & ~7.02 (0.24) & ~6.53 (0.24) & ~8.18 (0.13) & 11.86 (0.28) & 11.95 (0.04)\\
100 & ~9.41 (0.24) & ~9.09 (0.24) & 10.74 (0.15) & 15.63 (0.51) & 17.01 (0.13)\\
200 & 16.92 (0.27) & 16.66 (0.28) & 20.59 (0.14) & 25.92 (0.64) & 32.56 (0.05)\\
400 & 24.28 (0.24) & 24.13 (0.25) & 29.16 (0.13) & 38.66 (0.73) & 47.09 (0.05)\\
\multicolumn{6}{c}{True positive rate (\%)}\\
~50 &  76.8 (4.2)  &  82.7 (3.8)  &  60.0 (2.9)  &  37.9 (5.5)  &  63.8 (2.4)\\
100 &  79.3 (2.8)  &  81.5 (2.4)  &  66.4 (2.1)  &  60.4 (3.6)  &  72.8 (1.9)\\
200 &  73.0 (1.8)  &  74.3 (1.8)  &  47.6 (1.5)  &  48.6 (4.5)  &  63.9 (1.6)\\
400 &  73.1 (1.4)  &  73.6 (1.3)  &  44.9 (1.1)  &  40.9 (3.9)  &  67.2 (1.2)\\
\multicolumn{6}{c}{False positive rate (\%)}\\
~50 &  ~2.8 (0.6)  &  ~2.1 (0.5)  &  ~1.0 (0.2)  &  ~2.7 (0.4)  &  ~2.1 (0.2)\\
100 &  ~1.5 (0.2)  &  ~1.4 (0.2)  &  ~1.1 (0.1)  &  ~2.9 (0.2)  &  ~4.3 (0.2)\\
200 &  ~1.0 (0.1)  &  ~0.9 (0.1)  &  ~0.4 (0.0)  &  ~1.3 (0.1)  &  ~1.8 (0.1)\\
400 &  ~0.4 (0.0)  &  ~0.4 (0.0)  &  ~0.1 (0.0)  &  ~0.5 (0.0)  &  ~1.0 (0.0)\\
\hline
\end{tabular*}
\end{table}

\begin{table}
\caption{Means and standard errors (in parentheses) of performance measures for different methods in model (d) over 100 replications.}\label{tab:model_d}
\def~{\phantom{0}}
\begin{tabular*}{\textwidth}{@{}c*{5}{@{\extracolsep{\fill}}c}@{}}
\hline
& \multicolumn{5}{c}{Method}\\
\cline{2-6}
$p$ & CARE & Oracle & CD-trace & gCoda & SPIEC-EASI\\
\hline
\multicolumn{6}{c}{Spectral norm loss}\\
~50 & ~2.86 (0.17) & ~2.53 (0.17) & ~3.55 (0.08) & ~3.96 (0.11) & ~4.92 (0.06)\\
100 & ~3.24 (0.15) & ~3.10 (0.16) & ~3.76 (0.06) & ~4.63 (0.13) & ~5.69 (0.07)\\
200 & ~3.29 (0.07) & ~3.20 (0.07) & ~4.16 (0.04) & ~4.87 (0.12) & ~5.65 (0.05)\\
400 & ~3.37 (0.05) & ~3.32 (0.06) & ~4.10 (0.03) & ~4.77 (0.09) & ~5.11 (0.01)\\
\multicolumn{6}{c}{Matrix $\ell_1$-norm loss}\\
~50 & ~4.70 (0.41) & ~4.03 (0.40) & ~5.29 (0.17) & ~6.15 (0.26) & ~7.71 (0.03)\\
100 & ~5.34 (0.44) & ~5.05 (0.41) & ~6.28 (0.38) & ~8.53 (0.48) & ~9.42 (0.06)\\
200 & ~6.50 (0.62) & ~6.07 (0.49) & ~7.07 (0.16) & 10.38 (0.62) & 10.11 (0.18)\\
400 & ~6.49 (0.40) & ~6.40 (0.39) & ~7.43 (0.22) & ~8.66 (0.39) & ~8.28 (0.13)\\
\multicolumn{6}{c}{Frobenius norm loss}\\
~50 & ~6.68 (0.27) & ~6.24 (0.28) & ~8.04 (0.16) & ~9.94 (0.37) & 14.75 (0.26)\\
100 & 11.41 (0.26) & 11.09 (0.26) & 13.68 (0.12) & 18.40 (0.61) & 26.19 (0.40)\\
200 & 17.02 (0.26) & 16.75 (0.26) & 21.16 (0.14) & 27.77 (0.81) & 36.72 (0.33)\\
400 & 24.92 (0.21) & 24.74 (0.22) & 30.22 (0.12) & 39.64 (0.90) & 47.82 (0.06)\\
\multicolumn{6}{c}{True positive rate (\%)}\\
~50 &  85.8 (4.0)  &  88.9 (3.0)  &  76.4 (3.3)  &  84.0 (4.0)  &  35.4 (6.9)\\
100 &  76.8 (2.9)  &  78.6 (3.0)  &  56.0 (2.4)  &  51.0 (5.4)  &  33.3 (6.1)\\
200 &  75.6 (1.9)  &  76.8 (1.9)  &  50.3 (1.7)  &  46.5 (5.3)  &  39.7 (4.4)\\
400 &  74.6 (1.2)  &  75.1 (1.3)  &  45.8 (1.3)  &  44.1 (5.2)  &  72.8 (1.2)\\
\multicolumn{6}{c}{False positive rate (\%)}\\
~50 &  ~3.1 (0.5)  &  ~2.2 (0.5)  &  ~4.2 (0.5)  &  ~9.6 (0.7)  &  ~7.8 (3.2)\\
100 &  ~1.5 (0.2)  &  ~1.4 (0.2)  &  ~2.0 (0.2)  &  ~3.9 (0.2)  &  ~6.4 (2.5)\\
200 &  ~0.9 (0.1)  &  ~0.8 (0.1)  &  ~0.5 (0.0)  &  ~1.7 (0.1)  &  ~3.2 (0.3)\\
400 &  ~0.4 (0.0)  &  ~0.4 (0.0)  &  ~0.1 (0.0)  &  ~0.8 (0.1)  &  ~1.3 (0.1)\\
\hline
\end{tabular*}
\end{table}

\begin{figure}
\begin{subfigure}{\textwidth}
\centering
\includegraphics[width=.75\textwidth]{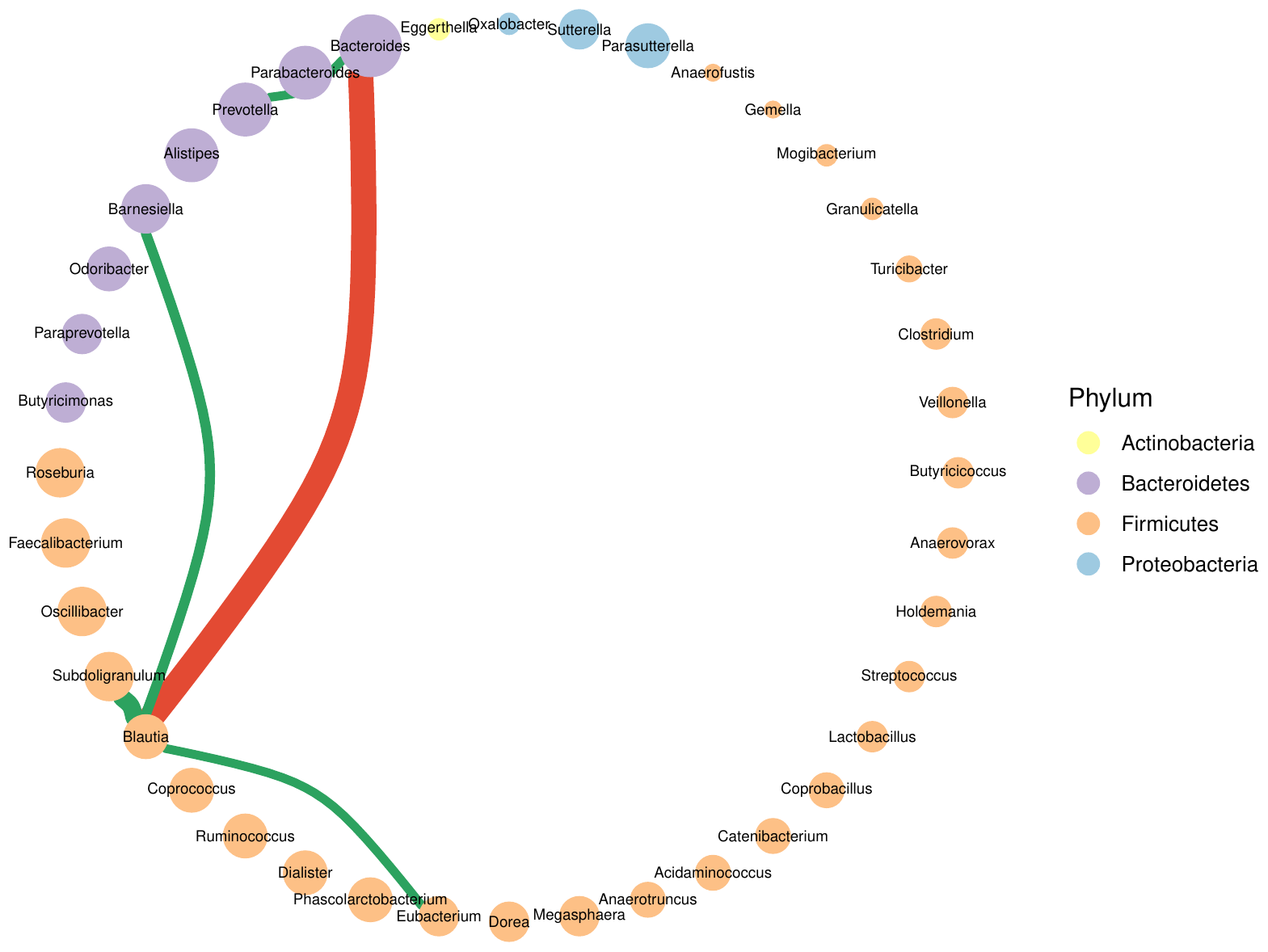}
\caption{Lean}
\end{subfigure}
\begin{subfigure}{\textwidth}
\centering
\includegraphics[width=.75\textwidth]{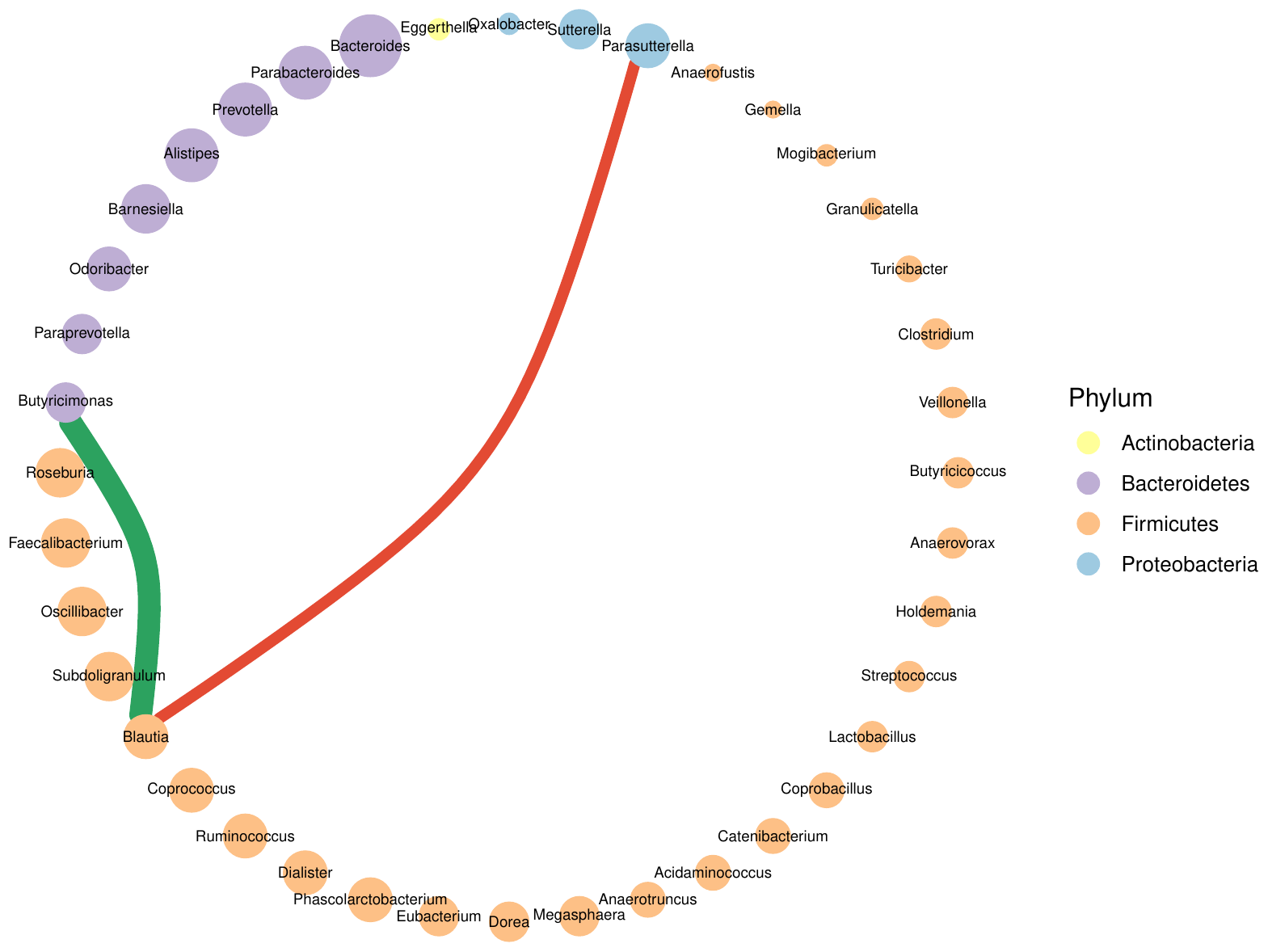}
\caption{Obese}
\end{subfigure}
\caption{Microbial interaction networks identified by the CD-trace method for the (a) lean and (b) obese groups in the gut microbiome data. Positive and negative edges are displayed in green and red, respectively, with thicknesses proportional to their strengths. Node sizes are proportional to the relative abundances of genera among all samples.}\label{fig:CD_net}
\end{figure}

\begin{figure}
\begin{subfigure}{\textwidth}
\centering
\includegraphics[width=.75\textwidth]{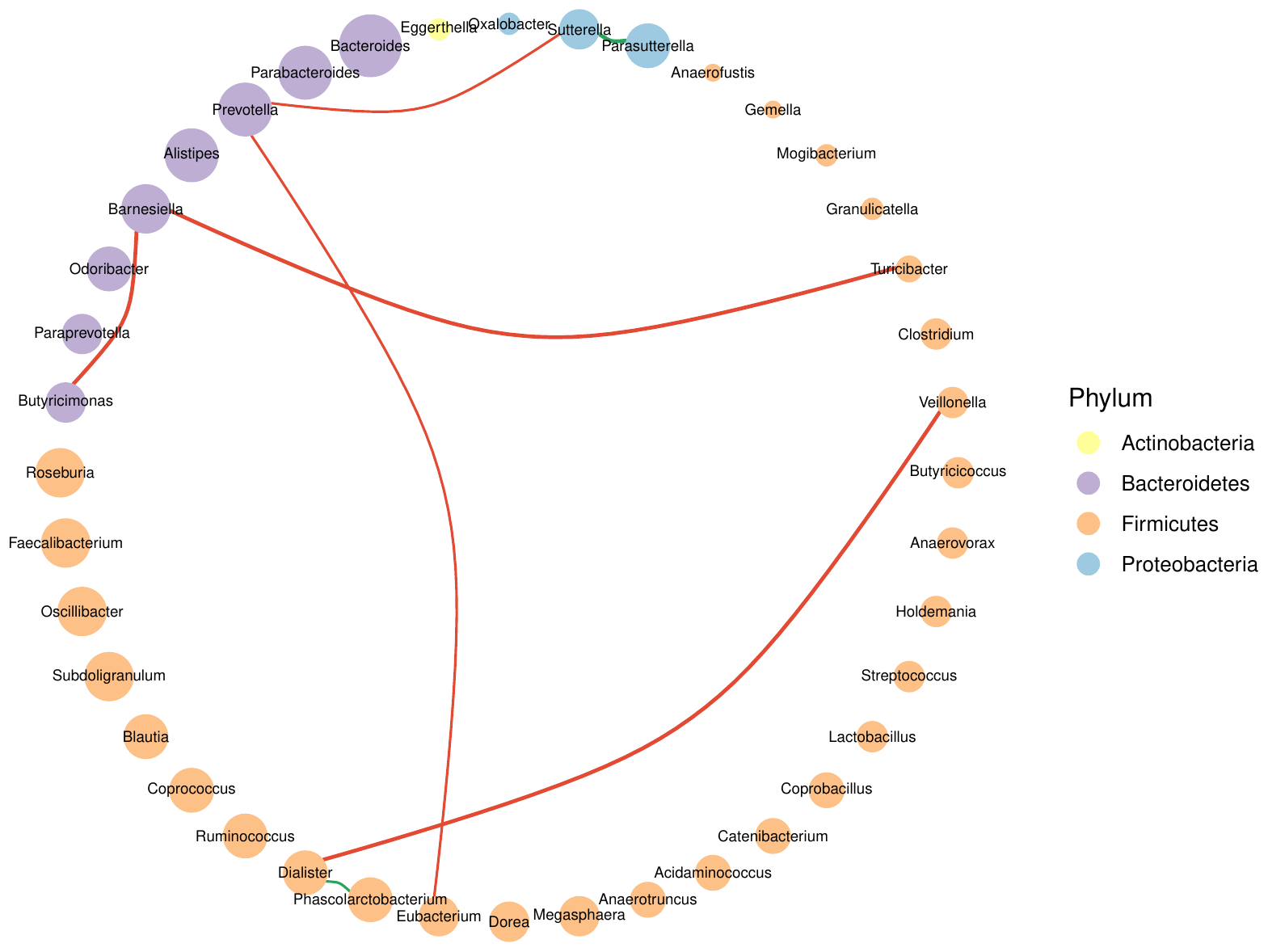}
\caption{Lean}
\end{subfigure}
\begin{subfigure}{\textwidth}
\centering
\includegraphics[width=.75\textwidth]{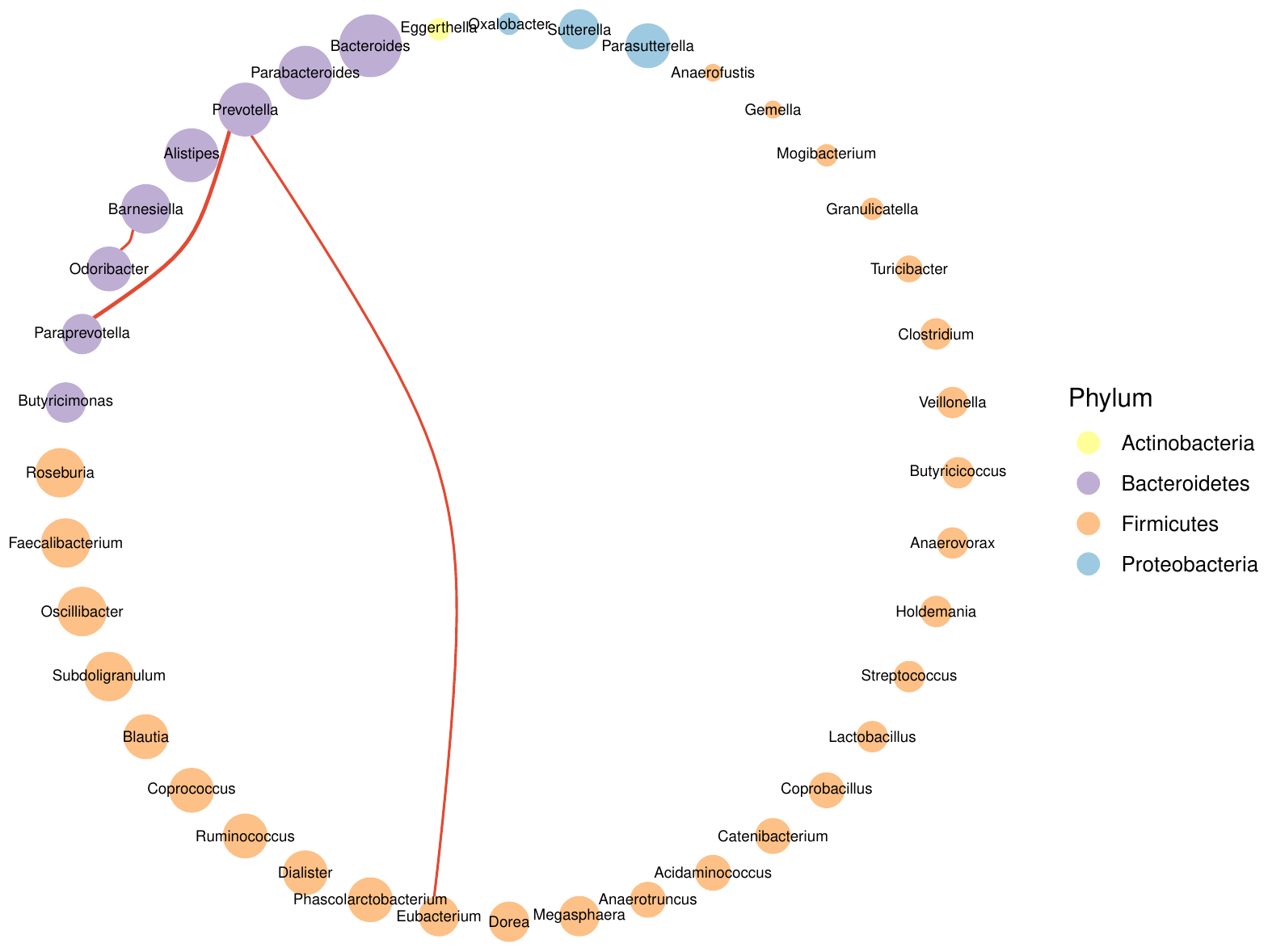}
\caption{Obese}
\end{subfigure}
\caption{Microbial interaction networks identified by the gCoda method for the (a) lean and (b) obese groups in the gut microbiome data. Positive and negative edges are displayed in green and red, respectively, with thicknesses proportional to their strengths. Node sizes are proportional to the relative abundances of genera among all samples.}\label{fig:gCoda_net}
\end{figure}

\begin{figure}
\begin{subfigure}{\textwidth}
\centering
\includegraphics[width=.75\textwidth]{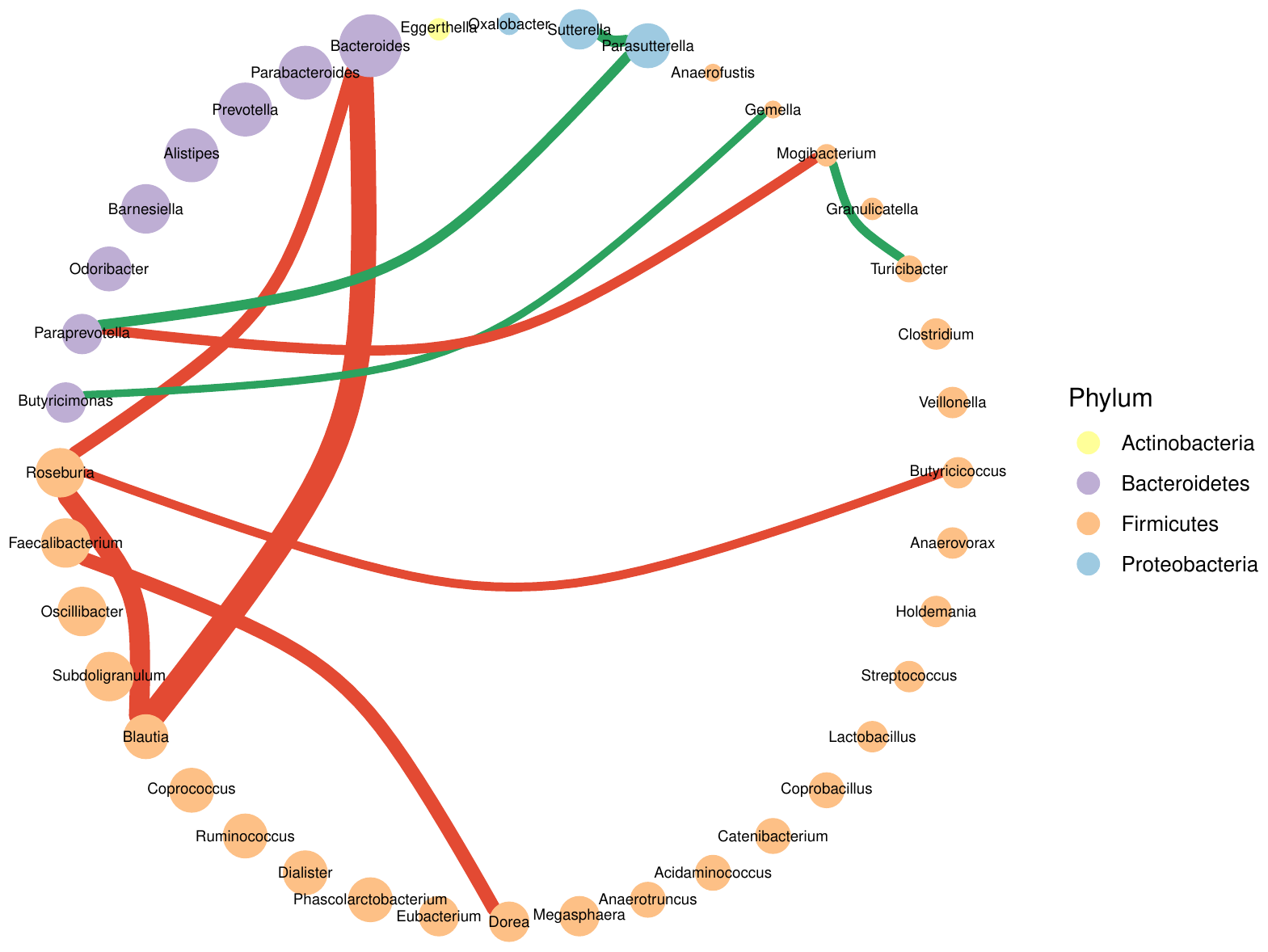}
\caption{Lean}
\end{subfigure}
\begin{subfigure}{\textwidth}
\centering
\includegraphics[width=.75\textwidth]{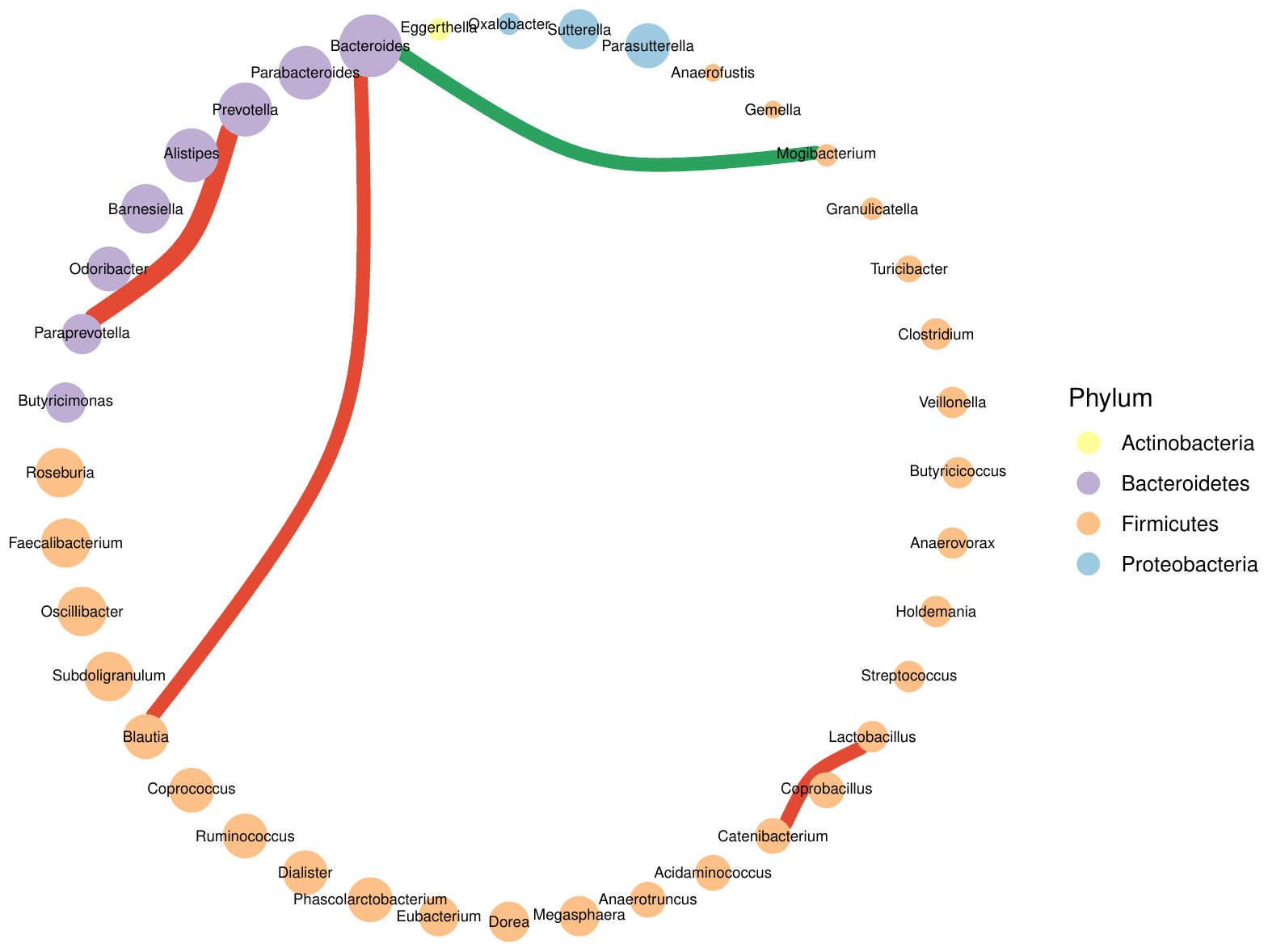}
\caption{Obese}
\end{subfigure}
\caption{Microbial interaction networks identified by the SPIEC-EASI method for the (a) lean and (b) obese groups in the gut microbiome data. Positive and negative edges are displayed in green and red, respectively, with thicknesses proportional to their strengths. Node sizes are proportional to the relative abundances of genera among all samples.}\label{fig:SPIEC_net}
\end{figure}

\clearpage
\bibliographystyle{jasa}
\bibliography{care_ref}